\documentclass[runningheads]{llncs}
\usepackage[T1]{fontenc}
\usepackage{graphicx}

\usepackage{cite}
\usepackage{amssymb,amsfonts}
\usepackage{esvect}

\usepackage{macros}
\usepackage{float}
\usepackage{tcolorbox}
\usepackage{stmaryrd}
\usepackage{centernot} 
\usepackage{soul} 
\usepackage{mathpartir,bussproofs}
\usepackage{mathtools}
\usepackage{mathrsfs}
\usepackage{arydshln}
\usepackage{anyfontsize}
\usepackage{thm-restate}
\usepackage{hyperref}
\usepackage{cleveref}
\usepackage{mdframed}
\usepackage{algorithm,algpseudocode}
\usepackage{graphicx}%

\usepackage{color}

\SetSymbolFont{stmry}{bold}{U}{stmry}{m}{n}

\begin{document}
\title{Equational Reasoning Modulo Commutativity in Languages with Binders}
%
%
\author{Ali K. Caires-Santos\inst{1} \orcidID{0009-0004-3183-3686} \and
Maribel Fern\'andez\inst{2} \orcidID{0000-0001-8325-5815}\and
Daniele Nantes-Sobrinho\inst{1,3} \orcidID{0000-0002-1959-8730}}
\authorrunning{A.K Caires-Santos et al.}
%
\institute{University of Bras\'ilia, Brazil \and
King's College London, UK\\
\and
Imperial College London, UK\\
}
\maketitle              
\begin{abstract}
Many formal languages include binders as well as operators that satisfy equational axioms, such as commutativity. Here we consider the nominal language, a general formal framework which provides support for the representation of binders, freshness conditions and $\alpha$-renaming. 
Rather than relying on the usual freshness constraints, we introduce a nominal algebra which employs  permutation fixed-point constraints in $\alpha$-equivalence judgements, seamlessly integrating commutativity into the reasoning process.
We establish its proof-theoretical properties and provide a sound and complete semantics in the setting of nominal sets.
Additionally, we propose a novel algorithm for nominal unification modulo commutativity, which we prove terminating and correct.  By leveraging fixed-point constraints, our approach ensures a finitary unification theory, unlike standard methods relying on freshness constraints.
This framework offers a robust foundation for structural induction and recursion over syntax with binders and commutative operators, enabling reasoning in settings such as first-order logic and the $\pi$-calculus.

\end{abstract}
\section{Introduction}

In computer science and logic, formal languages abstract away concrete syntax details, focusing on abstract syntax to capture structural properties. However, conventional abstract syntax falls short for languages with variable-binding constructs.
%
%
The formalisation  of languages with binding is challenging and requires specialised tools: one of the aims of the POPLMark challenge~\cite{DBLP:conf/tphol/AydemirBFFPSVWWZ05} was to explore and compare methods for the formalisation of binding operators. Many different techniques were proposed to address the challenge, for example, the approach based on  higher-order abstract syntax (HOAS) is described  in the survey~\cite{DBLP:journals/jar/FeltyMP15} by Felty et al.
Here we follow the nominal approach, introduced by Gabbay and Pitts~\cite{DBLP:journals/fac/GabbayP02}: Binders are specified using an abstraction construct of the form $[a]t$, where $a$ is a name and $t$ an arbitrary expression (read ``abstract $a$ in $t$'').

Nominal syntax~\cite{DBLP:journals/fac/GabbayP02} makes a distinction between \emph{names} $a,b,...$ (also called atoms), which can be abstracted but cannot be substituted, and \emph{(meta-)variables} $X,Y,...$ (also called \emph{unkowns}), which can be substituted but not abstracted.
The $\alpha$-equivalence relation $\approx$ is defined using the notion of support, or equivalently, the notion of freshness, a concept that generalises the idea of an atom not being free in a term: $a$ is \emph{fresh} for  $t$, written $a \fresh t$, if $a$ is not in the support of $t$.

More generally, in languages where operators are subject to equational axioms $E$, such as associativity ($\tf{A}$) and commutativity ($\C$), the equivalence relation of interest is generated by $\alpha$ and  $E$.
A well-known example is first-order logic with $\forall$ and $\exists$ quantifiers and conjunction and disjunction connectives; similarly, in Milner's $\pi$-calculus expressions are considered modulo a structural congruence. 
Effective reasoning techniques for such languages are essential, particularly when proving properties about them in a proof assistant. A common subproblem in this context is solving unification problems—finding the appropriate {\em substitution} that equates two expressions—when applying lemmas in automated proofs.


Although unification of first-order terms modulo commutativity is decidable and finitary,   and nominal unification is unitary (every solvable problem has a unique most general unifier), unification modulo $\alpha$ and $\C$ is not finitary if solutions are expressed using substitutions and freshness contexts~\cite{DBLP:journals/mscs/Ayala-RinconSFS21}. This motivated the study of a version of nominal syntax where the $\alpha$-equivalence relation is specified using permutation fixed-point constraints $\pi \fix{} t$ (read ``$\pi$ fixes $t$)'' instead of freshness constraints~\cite{DBLP:journals/lmcs/Ayala-RinconFN19}.
However,  the derivation relation defined in~\cite{DBLP:journals/lmcs/Ayala-RinconFN19} is sound for strong nominal algebras~\cite{10.1561/2500000017,DBLP:journals/corr/abs-1006-3027} only (that is, the carriers should be strong\footnote{A nominal set is strong iff its elements are strongly supported by finite sets, i.e., whenever a permutation fixes an element $x$ of the set then it has to fix the atoms occurring in $x$ pointwise.} nominal sets). Soundness with respect to the standard nominal set models was proved only for a restricted calculus in~\cite{entics:14777}.

\paragraph*{Contributions.} The system presented in this paper is a generalisation of the latter in two ways: judgements here can include arbitrary fixed point constraints in contexts, and equational axioms are taken into account, making the system more expressive while maintaining soundness for standard nominal set models. Devising a proof system for general fixed point judgements that is sound and complete  is challenging; the difficulty lies in the variable rule, which needs to be sufficiently powerful to derive all the valid judgements, but only those that are valid in all nominal set models.  A key idea to obtain soundness and completeness with respect to standard nominal set semantics is to analyse the permutation cycles in fixed-point assumptions.


We also study  satisfiability of equality constraints: we  show that nominal unification is decidable and unitary if $E = \emptyset$, and it is finitary if $\C$ommutativity axioms are included. More precisely, we provide a unification algorithm specified as a set of simplification rules on unification problems, which we show terminating
and correct. 
The normal forms of the simplification system consist of a finite set ($\Upsilon$) of  fixed-point constraints and a substitution. The algorithm outputs a most general solution of the form $(\Upsilon,\sigma)$ if $E = \emptyset$ and  a finite complete set of solutions of the form $(\Upsilon_i,\sigma_i)$ if $E = \C$. Having a finitary theory is important in practice, for example to design  programming languages and theorem provers that support binders and equational axioms, to specify completion procedures for rewriting modulo $\alpha$ and $E$, etc.
A notable feature of our unification algorithm is its treatment of computed {\em substitutions}: it ensures that meta-variables are instantiated while maintaining stability under renaming of bound names—preserving consistency when fresh names are introduced during computation (see Remark~\ref{rmk:substitution}).

Proofs of selected results can be found in the Appendix. An extended version with all the proofs is available in~\cite{arxiv/cairessantos2025}.

\section{Preliminaries}\label{sec:preliminaries}

In this section, we introduce some basic notions and notations used throughout the paper; for more details we refer the reader to~\cite{book/Pitts, DBLP:journals/fac/GabbayP02, DBLP:journals/logcom/Gabbay09}.

\subsubsection*{Nominal Syntax.}
Fix disjoint countably infinite sets $\mathbb{A} = \{a,b,c,\ldots\}$ of {\em atoms} and $\V = \{X,Y,Z,\ldots\}$ of {\em variables}. A {\em signature} $\Sigma$ is a set of {\em term-formers} ($\tf{f},\tf{g},\ldots$) $-$ disjoint from $\A$ and $\V$ $-$ where each $\tf{f}\in \Sigma$ has a fixed \emph{arity} $n\geq 0$, denoted $\tf{f}:n$.
%
A {\em finite} permutation $\pi$ of atoms is a bijection $\A \to \A$ such that the set $\dom{\pi} := \{a \in\A \mid \pi(a) \neq a\}$ is finite. Write $\id$ for the {\em identity permutation}, $\pi\circ \pi'$ for the {\em composition} of  $\pi$ and $\pi'$, and $\pi^{-1}$ for the {\em inverse} of $\pi$.
We use {\em swappings} of atoms of the form $(a\ b)$ to implement the renaming of atoms, for the purpose of $\alpha$-conversion.   Given two permutations $\pi$ and $\rho$, the permutation $\pi^\rho = \rho\circ\pi\circ\rho^{-1}$ denotes the {\em conjugate} of $\pi$ with respect to $\rho$.

 \emph{Nominal terms} are defined inductively by the following grammar:
\[ s,t,u::= a \mid \pi\act X \mid \tf{f^E}(t_1,\ldots, t_n) \mid [a]t \]
where $a$ is an {\em atom}; $\pi\act X$ is a {\em suspension}, where $\pi$ is an atom permutation. Intuitively, $\pi\act X$ denotes that $X$ will get substituted for a term and then $\pi$ will permute the atoms in that term accordingly (i.e., the permutation $\pi$ is suspended). Suspensions of the form $\id\act X$ will be represented simply by $X$; {\em a function symbol} $\tf{f} : n$ of arity $n$ may be equipped with an equational theory $\tf{E}$, hence $\tf{f^E}(t_1,\ldots, t_n)$ denotes the \emph{application} of $\tf{f^E}$ to a tuple $(t_1,\ldots, t_n)$. We write $\tf{f^\emptyset}$, or simply $\tf{f}$, to emphasise that no equational theory is assumed for $\tf{f}$; $[a]t$ denotes the \emph{abstraction} of the atom $a$ over the term $t$. The set of all nominal terms formed from a signature $\Sigma$ will be denoted by $\F(\Sigma,\V)$.

Write $\var{-}$ and $\atm{-}$ to denote the sets of all variables and all atoms occurring in the argument, respectively. A nominal term $t$ is {\em ground} when it does not mention variables, i.e., when $\var{t} = \emptyset$. The set of all ground terms will be denoted by $\F(\Sigma)$. The set of \emph{free atoms} of a ground term $g$ is denoted as $\tf{fn}(t)$ and consists of the atoms, say $a$, that do not occur under the scope of an abstraction $[a]$, that is, $\tf{fn}(a) = \{a\}, \tf{fn}(\tf{f^E}(g_1,\ldots,g_n)) = \bigcup_{i=1}^n\tf{fn}(g_i),$ and $\tf{fn}([a]g) = \tf{fn}(g)\setminus\{a\}$. Write $t \equiv u$ for {\em syntactic identity} of terms.

The {\em object-level permutation action on terms}, denoted as $\pi\act t$, is defined  via the homomorphic application of a permutation on the structure of the term: $ \pi \act a \equiv  \pi(a) ,~\pi \act (\pi' \act X) \equiv (\pi \circ \pi') \act X, \pi \act ([a]t)\equiv [\pi(a)](\pi \act t)$ and $
\pi \act \tf{f^E}(t_1, \ldots, t_n) \equiv \tf{f^E}(\pi \act t_1, \ldots, \pi \act t_n)$.
\emph{Substitutions}, usually denoted by $\sigma,\delta,\ldots$, are finite mappings from variables to terms, which extend homomorphically over terms. We write $Id$ to denote the identity substitution, $\dom{\sigma}$ 
to denote the domain
of $\sigma$. Substitutions and permutations commute: $\pi\act (s\sigma) \equiv (\pi \act s)\sigma$.

\begin{example}Abstractions will be used to represent binding operators (e.g. the lambda operator from the $\lambda$-calculus). Considering the signature $\Sigma = \{\tf{lam}: 1, \tf{app} : 2\}$ and treating $\lambda$-variables as atoms, $\lambda$-terms can be inductively generated by the grammar:
\(
    e::= a\mid \tf{app}(e,e)\mid \tf{lam}([a]e)
\).
Note that $(a \ b)\cdot \tf{lam}([a]a)\equiv \tf{lam}([b]b)$ implements the renaming of $a$ for $b$ using the swapping $\swap{a}{b}$. 
\end{example}

\paragraph*{Groups and Permutations.}
For completeness, we present some elementary notions, notations and properties of groups and permutations that will be used throughout the paper. This subsection can be skipped on a first reading and consulted later as a reference.
A {\em permutation group}
is a group\footnote{A {\em group} is a set $G$ equipped with an operation $\circ:G\times G\to G$ that satisfies four properties: closure, associativity, identity, and invertibility.} $G$ whose elements are permutations on a set and the group operation is the composition $\circ$ of permutations.
 If $S$ is a non-empty subset of a permutation group $G$, then we denote by $\pair{S}$  the subgroup of all elements of 
$G$ that can be expressed as the finite composition of elements in  $S$ and their inverses.
$\pair{S}$ is called the {\em subgroup generated by $S$}.
We will work with the group of {\em finite} permutations of atoms $\A$, denoted by $\Perm{\A}$. Moreover, given $B\subseteq \A$ write $\Perm{B} = \{\pi\in\Perm{\A}{} \mid \dom{\pi}\subseteq B\}$.
Recall that finite permutations can be expressed in two equivalent ways: either as compositions of swapping (a.k.a. transpositions), or as compositions of {\em disjoint permutation cycles}. E.g., $(a \ d)(a \ c)(a \ b)(e \ f)$ and $(a \ b \ c \ d)(e \ f)$ represent the same permutation.
If $H$ and $K$ are subsets of a permutation group $G$ whose operation is $\alert{\circ}$ (in particular, if $H$ and $K$ are subgroups of $G$), the set $\{h\alert{\circ}k \mid h \in H, k \in K\}$ will be denoted by $H\alert{\circ} K$~\footnote{The usual notation is just by juxtaposition $HK$, we will distinguish the operator $H\alert{\circ}K$ to improve readability later on.}.
In general, $H\alert{\circ}K$ is not a subgroup of $G$, even when $H$ and $K$ are subgroups of $G$. However, $H\alert{\circ}K$ is a subgroup of $G$ iff $H\alert{\circ}K=K\alert{\circ}H$.

\subsubsection*{Nominal Sets.}
A {\em $\Perm{\A}$-set}, denoted by $\nom{X}$, is a pair $(|\nom{X}|,\act)$ consisting of an {\em underlying set} $|\nom{X}|$ and a {\em permutation action} $\cdot $, which is a group action on $|\nom{X}|$, i.e., an operation
$\act:\Perm{\A}\times |\nom{X}|\to |\nom{X}|$ such that $\id\act x = x$ and $\pi\act (\pi'\act x) = (\pi\circ\pi')\act x$, for every $x\in |\nom{X}|$ and $\pi,\pi'\in \Perm{\A}$.  We will write $\pi\act_{\nom{X}} x$, when we want to make $\nom{X}$ clear.
For $B \subseteq \A$, we denote $\Fix{B}$ as the set of permutations that fix pointwise the elements of $B$, i.e.
  $
  \Fix{B} = \{\pi\in\Perm{\A}\mid \forall a\in B.~\pi(a) = a\}.$
A set $B\subseteq \A$ {\em supports} an element $x \in |\nom{X}|$ when for all permutations $\pi\in\Perm{\A}$:
$\pi\in\Fix{B} \implies \pi\act x = x $.
Additionally,  $B$ {\em strongly supports} $x\in|\nom{X}|$ if the reverse implication holds.
A {\em nominal set} is a $\Perm{\A}$-set $\nom{X}$ all of whose elements are finitely supported.
$\nom{X,Y,Z}$ will range over nominal sets.

The {\em  support} of an element $x\in|\nom{X}|$ of a nominal set $\nom{X}$ is defined as
$\supp{}{x} = \bigcap\{B\mid \text{$B$ is finite and supports $x$}\}.$
This implies that $\supp{}{x}$ is the {\em least} finite support of $x$.
The support provides a semantics for freshness constraints in the class of nominal sets~\cite{DBLP:journals/logcom/Gabbay09,DBLP:journals/logcom/GabbayM09}:  \text{When } $ a\notin\supp{}{x}$, \text{ we write } $a\fresh_{\tt sem} x$, and read this as ``$a$ is fresh for $x$''.

One of the main properties of the relation $a\fresh_{\tt sem} x$ is the so-called \emph{Choose-a-Fresh-Name Principle}, which says that if $x \in|\nom{X}|$ has a finite support, then it is always possible to choose a new atom $a$ such that $a\fresh_{\tt sem} x$. Therefore, the notion of freshness provides a generative aspect of fresh atoms.

\begin{example}[Some simple nominal sets]\label{ex:nominal-and-strong-sets}
 The $\Perm{\A}$-set $(\A,\act)$ with the action $\pi\act_\A a = \pi(a)$ is a nominal set and $\supp{}{a} = \{a\}$. The group $\Perm{\A}$ equipped with the action by conjugation $\rho\act_{\Perm{\A}}\pi = \pi^{\rho}$ forms a nominal set and $\supp{}{\pi} = \dom{\pi}$. The set of ground terms $\F(\Sigma)$ with the usual permutation action on terms is a nominal set and $\supp{}{g} = \atm{g}$. 
 Another example is the set $\pow{\tt fin}{\A} = \{B\subset\A \mid B \text{ is finite}\}$. Then the $\Perm{\A}$-set $(\pow{\tt fin}{\A},\act)$ with the action $\pi\act_{\pow{\tt fin}{\A}} B = \{\pi\act_\A a\mid a \in B\}$ is a nominal set and $\supp{}{B} = B$. 
\end{example}



\subsubsection*{The new \texorpdfstring{$\new$}{new}-quantifier.} The $\new$-quantifier~\cite{DBLP:journals/fac/GabbayP02} models the quantification used when we informally write ``rename $x$ in $\lambda x.t$ to a new name''. Thus, if $\phi$ is the description of some
property of atomic names $a$, then $\new \atnew{a}.\phi(\atnew{a})$ holds when $\phi(\atnew{a})$ holds for all but finitely many atoms $\atnew{a}$. Usually, $\phi$ is a predicate in higher-order logic~\cite{book/Pitts} or in ZFA or FM set theory~\cite{DBLP:journals/fac/GabbayP02,DBLP:journals/bsl/Gabbay11}. The $\new$-quantifier is defined as
\begin{equation}\label{eq:new_cofin}
    \new \atnew{a}.\phi(\atnew{a}) \text{ holds iff } \{\atnew{a}\in\A \mid \phi(\atnew{a}) \text{ holds}\} \text{ is cofinite}\footnote{A {\em cofinite} subset of a set $M$ is a subset $B$ whose complement in $M$ is a finite set. }.
\end{equation}

Therefore, $\new \atnew{a}.(\ldots)$ means ``for a cofinite number of atoms $\atnew{a}, (\ldots)$'' . The quantifier $\new$ allows to quantify the \emph{generative} aspect of names/atoms. This aligns with the Choose-a-Fresh-Name Principle of the relation $\fresh_{\tt sem} $. The connection between these two notions is established by \emph{Pitts' equivalence}~\cite{DBLP:journals/fac/GabbayP02,book/Pitts}, which defines the freshness of an atom $a$ for an element $x$ in a nominal set using the new-quantifier $\new$ and a fixed-point identity as follows:
\begin{equation}\label{eq:pitts-eq-1}
    a \fresh_{\tt sem} x \iff \new \atnew{c}. (a \ \atnew{c})\cdot x=x.
\end{equation}

The following equivalences will be heavily used throughout the text:
\begin{equation}\label{eq:fresh_fix}
   \begin{aligned}
   a\notin \supp{}{x}&\iff  a \fresh_{\tt sem} x \\
   &\iff \new \atnew{c}. (a \ \atnew{c})\cdot x=x \text{ holds}\\
   & \iff \{\atnew{c}\in\A \mid \newswap{a}{c}\act x = x\} \text{ is cofinite.}
   \end{aligned}
\end{equation}

Let $\catnew{c} = \atnew{c_1}, \ldots, \atnew{c_n}$ be list of atoms. We will write $\newc{c}{}. \phi(\catnew{c})$ as shorthand for $\new \atnew{c_1}. \new \atnew{c_2}. \ldots \new \atnew{c_n}. \phi(\catnew{c})$. Since $\new \atnew{a_1}. \new \atnew{a_2}. \phi$ holds if and only if $\new \atnew{a_2}. \new \atnew{a_1}. \phi$ holds, then the set $\catnew{c}$ can be interchangeably treated as either an unordered list or a finite set of atoms. We will refer to this set as the set of {\em $\new$-quantified names}.

\section{Generalised Permutation Fixed-Points in Nominal Sets}\label{sec:permutation-fix-nominal-sets}

In this section, we  consider {\em generalised fixed-point identities}: a version of fixed-point identities with general \new-quantified atom permutations $\pi$.
That is,
\begin{equation*}
    \newc{c}{}.\pi \act x = x,
\end{equation*}
 where $\pi$ is a finite composition of disjoint cycles, and not only swappings. For example, $\new\atnew{c_1},\atnew{c_2},\atnew{c_3}. \swap{a \ \atnew{c_1}}{\atnew{c_2}}\circ \swap{b}{d}\cdot x = x $ is a generalised fixed-point identity.

In addition, we will split the permutation $\pi$ into  $\pi_{\catnew{c}}$, consisting of the cycles that contain at least one  $\new$-quantified atom,  and  $\pi_{\neg \catnew{c}}$, consisting of cycles whose atoms are not quantified. 


\begin{example}
    Consider the  generalised fixed-point identity $\new{\catnew{c}}.\pi\cdot x=x$ with $\catnew{c} = \atnew{c_1}, \atnew{c_2}, \atnew{c_3},\atnew{c_4}$ and 
    \(
    \pi = (a_1 \ \atnew{c_1} \ \atnew{c_2}) \circ (a_2 \ \atnew{c_3})\circ \highlight{(a_3 \ a_4 \ a_5)}\circ (a_6 \ \atnew{c_4})\circ \highlight{(a_7 \ a_8)}
\). Note that $\pi$
is equivalent to the permutation consisting of the composition of $\pi_{\catnew{c}}$ and  $\pi_{ \neg \catnew{c}}$:
 \(
    \pi = \underbrace{(a_1 \ \atnew{c_1} \ \atnew{c_2}) \circ (a_2 \ \atnew{c_3}) \circ (a_6 \ \atnew{c_4})}_{\pi_{\catnew{c}}}\circ  \underbrace{\highlight{(a_3 \ a_4 \ a_5)}\circ \highlight{(a_7 \ a_8)}}_{\pi_{\neg \catnew{c}}}.
\)
Interesting permutations with different behaviours are: (i) $\pi_1 = (a \ \atnew{c_1}) \circ (b \ \atnew{c_2} \ \atnew{c_3} \ d) \circ (f \ \atnew{c_4})$, with $(\pi_1)_{\catnew{c}}=\pi_1$ and $(\pi_1)_{\neg \catnew{c}}=\id$; and (ii) $\pi_2=(a_1 \ b_1\ d_1)\circ (a_2 \ b_2)$ with  $(\pi_2)_{\neg \catnew{c}}=\pi_2$ and $(\pi_2)_{\catnew{c}}=\id$.
\end{example}

We formalise this decomposition of permutations in the following definition.
\begin{definition}[Classifying Permutation Cycles]\label{def:splitting}
   Let $\pi$ be a permutation and $\catnew{c}$ a finite set of \new-quantified atoms. Define $\pi_{\catnew{c}}$ as the product of disjoint cycles of $\pi$ that mention atoms in $\catnew{c}$, and  $\pi_{\neg\catnew{c}}$ as the product of disjoint cycles of $\pi$ that do not mention atoms in $\catnew{c}$. Define $\pi_{\catnew{c}} := \id$ if $\catnew{c}\cap\dom{\pi} = \emptyset$.  Similarly, $\pi_{\neg\catnew{c}} := \id$ if every disjoint cycle of $\pi$ mentions at least one atom from $\catnew{c}$.
\end{definition}

\begin{example} The following illustrate some properties of splitting permutations:
    \begin{enumerate}
        \item In general, it is not true that $\pnew{\pi_1\circ\pi_2}{\catnew{c}} = \pnew{\pi_1}{\catnew{c}}\circ\pnew{\pi_2}{\catnew{c}}$ for all permutations $\pi_1,\pi_2$. Take $\catnew{c}=\atnew{c_1}$ and the permutations
  $ \pi_1 = (a \ \atnew{c_1})$ \text{ and }
        $\pi_2 = (a \ b \ d)$. Then
     \(
        \pi_1\circ\pi_2 = (a \ b \ d \ \atnew{c_1})
         \neq (a \ \atnew{c_1}) = (\pi_1)_{\atnew{c_1}} \circ (\pi_2)_{\atnew{c_1}}.
    \)

    \item In general, it is not true that $\pnew{\pi^\rho}{\catnew{c}} = (\pi_{\catnew{c}})^\rho$,  for all  permutations $\pi$ and $\rho$. Similarly,  $\pnew{\pi^\rho}{\neg \catnew{c}} = (\pnew{\pi}{\neg \catnew{c}})^\rho$ does not hold in general. Take $\catnew{c}=\atnew{c_1}$ and  consider  
    $    \pi = (a \ b \ \atnew{c_1})\circ (d \ e)$ and $
    \rho = (\atnew{c_1} \ d)$.
    After computing the conjugate and reorganising the cycles we obtain
 $\pnew{\pi^{\rho}}{\atnew{c_1}} = (\atnew{c_1} \ e)\neq (a \ b \ d) =  (\pi_{\atnew{c_1}})^{\rho}$. We can similarly check that $(\pi_{\neg\atnew{c_1}})^{\rho}\neq \npnew{\pi^{\rho}}{\atnew{c_1}}$.

\end{enumerate}
\end{example}

Assuming permutations are written as  compositions of disjoint cycles, the next lemma states that if $\pi\act x = x$ and  $\eta$ is a cycle in $\pi$, then if   one atom of $\eta$ is not in $\supp{}{x}$,  none of the atoms in $\eta$ are in $\supp{}{x}$. Moreover, it also states that all the permutations in the group generated by permutations that fix an element $x$ of a nominal set, also fix $x$.

\begin{restatable}{lemma}{generated}
\label{lemma:generated-group}
 Let $x$ be an element of a nominal set and $\pi, \eta_1,\ldots, \eta_n$ be permutations written as compositions of disjoint cycles.
    \begin{enumerate}
    \item If $\pi\act x=x$,  $\eta$ is a cycle of $\pi$ and $\dom{\eta} - \supp{}{x} \neq \emptyset$
    then $\dom{\eta}\cap\supp{}{x} = \emptyset$.
    \item If $\eta_i\act x = x$, for all $i=1,\ldots,n$, and  $\rho\in \pair{\eta_1,\ldots,\eta_n}$, then $\rho\act x = x$.
    \end{enumerate}
\end{restatable}
For e.g., if $\swap{a_1 \ a_2}{a_3}\circ \swap{a_4}{a_5}\cdot x = x $ and $a_2\notin \supp{}{x}$, then $a_3,a_1\notin \supp{}{x}$.

The following result generalises Pitts' equivalence (\ref{eq:pitts-eq-1}) to any permutation $\pi$.

\begin{restatable}[Generalised Pitts' equivalence]{lemma}{pitts}
\label{lemma:pitts-eq-generalized}\hfill
    \begin{enumerate}
        \item $\newc{c}{}.~\pi\act x =x$ iff $\newc{c}{}.\pi_{\catnew{c}}\act x = x$ and $\newc{c}{}.\pi_{\neg\catnew{c}}\act x = x$.

        \item $\newc{c}{}. ~\pi_{\catnew{c}}\act x = x$ iff  for all $\catnew{c}$, $\catnew{c}\cap\supp{}{x} = \emptyset$ implies $\dom{\pi_{\catnew{c}}}\cap \supp{}{x} = \emptyset$.

        \item $\newc{c}{}.~\pi_{\neg\catnew{c}}\act x = x$ iff for all $\catnew{c}$,  $\catnew{c}\cap\supp{}{x} = \emptyset$ implies $\pi_{\neg\catnew{c}}\act x = x$.

    \end{enumerate}
\end{restatable}


\section{A proof system}\label{sec:proof-system}

In this section, we introduce notions of fixed-point constraint, equality constraint and $\alpha,\C$-equivalence  judgement, together with derivation rules  (Figure~\ref{fig:fixed-rules_new}). Our judgements and rules are more general than the ones in~\cite{DBLP:journals/logcom/GabbayM09, entics:14777,DBLP:journals/lmcs/Ayala-RinconFN19} since we consider {\em general} $\new$-quantified permutations subsuming both freshness and algebraic (here, commutative) conditions (more discussion about expressivity can be found below); as a consequence, we need only one predicate (equality, denoted $\aeq{C}$) and only one set of rules to handle notions of term equality and name freshness. Before discussing the rules, we define generalised fixed-point contexts and judgements, as well as a normalisation process to refine permutations in contexts.

 
\begin{definition}[Generalised fixed-point constraint and context] \hfill 
\begin{enumerate}
\item An {\em $\alpha,\C$-equality constraint} (or simply an {\em equality constraint}) is an expression of the form $t \aeq{C} u$, where $t$ and $u$ are nominal terms. 

       \item A (generalised) {\em fixed-point constraint} is an $\alpha,\C$-equality constraint of the form $\pi \act t \aeq{C} t$, where $\pi$ is a permutation. A fixed-point constraint of the form $\pi \act X \aeq{C} X$ is called {\em primitive}.  
       \item    A (generalised) {\em fixed-point context} is an expression of the form $\newc{c}{}. \Upsilon$, where $\catnew{c} = \atnew{c_1}, \ldots, \atnew{c_n}$ and $\Upsilon$ is a finite set containing only primitive fixed-point constraints. For simplicity, we write $\Upsilon_{\catnew{c}}$ to denote  $\newc{c}{}. \Upsilon$. Symbols such as $\Upsilon,\Psi,\Upsilon',\Psi'$, are used to represent fixed-point contexts.
       \end{enumerate}
\end{definition}
We will adopt the existing notation for fixed-point constraints~\cite{DBLP:journals/lmcs/Ayala-RinconFN19} and abbreviate $\pi \act t \aeq{C} t$ as $\pi \fix{C} t$, writing primitive constraints as $\pi \fix{C} X$.

\begin{definition}[Judgement]
 An {\em $\alpha,\C$-equivalence judgment} is an expression of the form  $\new\catnew{c}.(\Upsilon \vdash t\aeq{C} u)$
  consisting of a fixed-point context $\Upsilon_{\atnew{\pvec{c}'}}$, with $\atnew{\pvec{c}'}\subseteq \catnew{c}$, and an equality constraint $t\aeq{C} u$, which may contain atoms from $\catnew{c}$. We simplify the notation to  $\Upsilon_{\catnew{c}} \vdash t\aeq{C} u$.
\end{definition}
%
%
We  write $\Upsilon_{\catnew{c}}|_X$ to refer to the fixed-point constraints on the variable $X$, that is,  $\Upsilon_{\catnew{c}}|_X := \newc{c}{}. \{\pi\fix{C} X\mid \pi\fix{C} X\in \Upsilon\}$. We extend the notation $\atm{-}$ to contexts $\Upsilon_{\catnew{c}}$. For example, $\atm{\new \atnew{c}.\{(a \ b \ \atnew{c})\fix{C} X\}} = \{a, b, \atnew{c}\}$. The set also includes vacuous $\new$-quantified names, as in $\atm{\new \atnew{c,c_1}.\{(a \ \atnew{c})\fix{C} X\}} = \{a, b, \atnew{c,c_1}\}$. We may slightly abuse notation when convenient and write $\pi\fix{C} X \in \Upsilon_{\catnew{c}}$ instead of $\pi\fix{C} X \in \Upsilon$.

\subsubsection{Expressivity.}\label{sec:normalised}
In the standard nominal approach using freshness constraints, contexts store only information about the freshness of atoms in (meta-)variables — they impose conditions on which atoms cannot occur free in any interpretation of the variables. Generalised fixed-point contexts are more expressive. 
Consider the context $\new \atnew{c}.\Psi_1 = \new \atnew{c}.\{(a \ \atnew{c} \ b)\circ(d \ e)\fix{C} Z\}$. By the generalized Pitts' equivalence (see Lemma~\ref{lemma:pitts-eq-generalized}), this  context implies two things:

\begin{enumerate}
\item The atoms $a, b$ (and also $\atnew{c}$) cannot occur freely in an interpretation of $Z$.

\item $(d \ e)\fix{C} Z$ holds, which does not force the atoms $d$ and $e$ to be fresh in interpretations of $Z$, since neither $d$ nor $e$ are $\new$-quantified or occur in a cycle with a $\new$-quantified name. For example,  $Z$ could be mapped to the term $d+e$ where $+$ is a commutative operator. In this case, $a, b$, and $\atnew{c}$ do not occur freely in $d+e$, while $d$ and $e$ occur freely in $d+e$, but $(a \ \atnew{c} \ b)\act (d+e) \aeq{C} d+e$ and $(d \ e)\act(d+e) \aeq{C} d+e$.
\end{enumerate}

The example above shows that generalised fixed-point contexts are a form of ``hybrid context'': they store more than just freshness information; they also encode relevant fixed-point constraints that cannot be reduced to freshness.

\subsection{Normalisation}\label{sec:normalization}
We can explore the properties of permutations to improve our fixed-point contexts. The idea is to use the properties in  Lemma~\ref{lemma:generated-group} to transform contexts into equivalent, simpler ones,  using first the rule  {\bf (R1)} exhaustively and then  {\bf (R2)}: 
{\small 
\[
\begin{array}{rl}
     \textbf{(R1)}& \new\catnew{c}. (\{\pi_{\catnew{c}}\circ\pi_{\neg \catnew{c}} \fix{C} X\}\cup \Upsilon') \implies     \new\catnew{c}. (\{\pi_{\catnew{c}}\fix{C} X, \pi_{\neg \catnew{c}} \fix{C} X\}\cup \Upsilon') \\[1mm]
     \textbf{(R2)}& \new\catnew{c}. (\{\pi_{\catnew{c}}\fix{C} X, (\rho_{\neg \catnew{c}} \circ \rho'_{\neg \catnew{c}})\fix{C}X \}\cup \Upsilon') \implies     \new\catnew{c}. (\{(\pi_{\catnew{c}}\circ \rho_{\neg \catnew{c}})\fix{C} X,   \rho'_{\neg \catnew{c}} \fix{C} X\}\cup \Upsilon') \\
     & \text{if } \dom{\pi_{\catnew{c}}}\cap \dom{\rho_{\neg\catnew{c}}}\neq \emptyset \text{, where $\rho$ is a cycle disjoint from $\rho'$.}
\end{array}
\]
}
We justify our rules as follows: For rule {\bf (R1)}: From Definition~\ref{def:splitting}, it follows that given $\catnew{c}$, every permutation $\pi$ that is written as a product of disjoint cycles can be rewritten such that $\pi= \pi_{\catnew{c}}\circ \pi_{\neg \catnew{c}}$.  By Lemma~\ref{lemma:pitts-eq-generalized}, the constraint $\pi_{\catnew{c}}\circ \pi_{\neg \catnew{c}} \fix{C} X$ is equivalent to  $\pi_{\catnew{c}} \fix{C} X \wedge \pi_{\neg \catnew{c}} \fix{C} X$. For rule {\bf (R2)} the constraints on the lhs are normal forms w.r.t. ({\bf R1}): From the constraint $\pi_{\catnew{c}}\fix{C} X$ it follows that the atoms in $\pi_{\catnew{c}}$ cannot occur free in the instances of $X$ (they are not in the support of the instances of $X$). Since $\rho$ is a cycle that contains atoms in $\pi_{\catnew{c}}$, by  Lemma~\ref{lemma:generated-group}, all the atoms in the cycle must not be in the support of $X$. From the equivalences in (\ref{eq:fresh_fix}), not being in the support
is equivalent to being the fixed-point with atoms in $\catnew{c}$. The composition $\pi_{\catnew{c}}\circ \rho_{\neg \catnew{c}} $ generates a new permutation $\pnew{\pi'}{\catnew{c}}$ with the atoms of $\pi_{\catnew{c}}$ and $\rho_{\neg \catnew{c}}$ organised as a composition of disjoint cycles.

\begin{example}\label{exa:normalisation}
To illustrate ({\bf R1}), consider  the  context $\Psi_{\atnew{c}}=\new \atnew{c}.\{(\newswap{a}{c}\circ \swap{d}{e}\fix{C} Z\}$ which is equivalent to the context   $\new \atnew{c}.\{\newswap{a}{c}\fix{C} Z, \swap{d}{e}\fix{C} Z\}$ with simpler constraints. The left constraint gives information about freshness whereas the  right constraint is about terms that are fixed by permutation $\swap{d}{e}$. To illustrate ({\bf R2}), consider the context $\new \atnew{c}.\Psi_2 = \new \atnew{c}.\{\newswap{a}{c}\fix{C} X, (a \ b)\circ(d \ e)\fix{C} X\}$. By generalised Pitts' equivalence (Lemma~\ref{lemma:pitts-eq-generalized}), $\newswap{a}{c}\fix{C} X$ alone implies that $a$ cannot occur freely in any interpretation of $X$, as $\atnew{c}$ is $\new$-quantified. In contrast, when considering the constraint $(a \ b)\circ (d \ e)\fix{C} X$ on its own, the names $a,b,d,e$ may occur freely in an interpretation of $X$. However, when the two constraints are considered together, as in $\new \atnew{c}.\Psi_2$, the scenario for the name $b$ changes: it follows from Lemma~\ref{lemma:generated-group} that $b$ also cannot be free in any interpretation of $X$.
Rule ({\bf R2}) composes both permutations as in $\newswap{a}{c}\circ \swap{a}{b}$ obtaining the equivalent context $\new \atnew{c}.\{\newswap{a \ b}{c}\fix{C} X, (d \ e)\fix{C} X\}$ where it is explicit that $a$ and $b$ must be fresh in all the instances of $X$.
\end{example}

\paragraph{\bf Important Assumption.} From now on, unless explicitly stated otherwise, we will assume that {\em every fixed-point context, say $\Upsilon_{\catnew{c}}$, is in normal form with respect to the exhaustive application of rules ({\bf R1}) and ({\bf R2})}. 

Thus, just as we can split a permutation $\pi$ into its disjoint parts $\pi_{\catnew{c}}$ and $\pi_{\neg\catnew{c}}$, we will split a generalised fixed-point context $\Upsilon_{\catnew{c}}$ into its ``freshness''  part (denoted $(\Upsilon_{\catnew{c}})_{\fresh}$) and ``pure fixed-point'' (but non-freshness) part (denoted $(\Upsilon_{\catnew{c}})_{\fix{C}}$). That is, $(\Upsilon_{\catnew{c}})$ will be divided into the following disjoint sets:
$(\Upsilon_{\catnew{c}})_{\fresh} := \newc{c}{}.\{\pi_{\catnew{c}}\fix{C} X \mid \pi\fix{C} X\in \Upsilon\}$ and $(\Upsilon_{\catnew{c}})_{\fix{C}} := \newc{c}{}.\{\pi_{\neg\catnew{c}}\fix{C} X \mid \pi\fix{C} X\in \Upsilon\}$. 
For example, the context $\new\catnew{c}.\Psi_2$ from Example~\ref{exa:normalisation} will be split into $(\new \atnew{c}.\Psi_2)_{\fresh} = \new \atnew{c}.\{(a \ b \ \atnew{c})\fix{C} X\}$ and $(\new \atnew{c}.\Psi_2)_{\fix{C}} = \new \atnew{c}.\{(d \ e)\fix{C} X\}$.

In the following  we use the notation
$\perm{}{(\Upsilon_{\catnew{c}})_{\fix{C}}} := \{\pi_{\neg\catnew{c}} \mid \pi\fix{C} X\in \Upsilon\}$

\subsection{Derivation rules for \texorpdfstring{$\aeq{C}$}{falphaeqC}}\label{sec:derivation-rules}


\begin{figure}[!t]
\small
\begin{mdframed}
\begin{mathpar}
   \inferrule*[right=$(\frule{\faeq{C}}{a})$]
    {\qquad }{\Upsilon_{\catnew{c}} \vdash a\aeq{C} a}
   \and
    \inferrule*[right=$(\frule{\faeq{C}}{var})$]
    {\pi'^{-1}\circ\pi\in\PN{}{\Upsilon_{\catnew{c}}|_X}}
        {\Upsilon_{\catnew{c}} \vdash \pi\act X \aeq{C} \pi'\act X}
   \and
   \inferrule*[right=$\tf{f}\neq \tf{f^C} \ (\frule{\faeq{C}}{\tf{f}})$]
   {\Upsilon_{\catnew{c}} \vdash t_1 \aeq{C} t_1'\\ \ldots \\
        \Upsilon_{\catnew{c}} \vdash t_n\aeq{C} t_n'}
        {\Upsilon_{\catnew{c}} \vdash {\tt f}(t_1,\ldots,t_n) \aeq{C} {\tt f}(t_1',\ldots,t_n')}
    \and
    \inferrule*[right=$(\frule{\faeq{C}}{\tf{f^C}})$]
    {\Upsilon_{\catnew{c}} \vdash s_0\aeq{C}t_i\\
       \Upsilon_{\catnew{c}} \vdash s_1\aeq{C}t_{1-i} \quad  i=0,1}
    {\Upsilon_{\catnew{c}} \vdash \tf{f^C}(s_0,s_1)\aeq{C} \tf{f^C}(t_0,t_1)}
    \and
    \inferrule*[right=$(\frule{\faeq{C}}{[a]})$]
    {\Upsilon_{\catnew{c}} \vdash  t \aeq{C} t'}
    {\Upsilon_{\catnew{c}} \vdash [a]t \aeq{C} [a]t'}
    \and
    \inferrule*[right=$(\frule{\faeq{C}}{ab})$]
    {\Upsilon_{\catnew{c},\atnew{c_1}} \vdash \newswap{a}{c_1}\act t \aeq{C} \newswap{b}{c_1}\act s}
    {\Upsilon_{\catnew{c}} \vdash [a]t \aeq{C} [b]s}
\end{mathpar}
\end{mdframed}
\caption{Derivation rules for $\aeq{C}$.}
\label{fig:fixed-rules_new}
\vspace*{4pt}
\end{figure}


We are now ready to present our derivation rules (Figure~\ref{fig:fixed-rules_new}). The rule  $(\frule{\faeq{C}}{ab})$ for abstraction relies on Pitts' definition of $\alpha$-equivalence on nominal sets ($=_\alpha$) using the $\new$-quantifier~\cite{book/Pitts}:
\begin{equation}\label{eq:pitts-eq-2}
 (a_1,x_1) =_\alpha (a_2,x_2) \iff \new \atnew{c_1}.~\newswap{a_1}{ c_1}\act x_1 = \newswap{a_2}{c_1}\act x_2.
\end{equation}
This means that two abstractions $[a_1]t$ and $[a_2]s$ are equal if  $\new \atnew{c_1}. \newswap{a_1}{c_1}\act t = \newswap{a_2}{c_1}\act s$. As expected, the rule for commutativity $(\frule{\faeq{C}}{\tf{f^C}})$  ensures  that to prove $\Upsilon_{\catnew{c}}\vdash {\tf f}^{\C}(s_0,s_1)\aeq{C} {\tf f}^{\C}(t_0,t_1)$ there must exist derivations of $\Upsilon_{\catnew{c}}\vdash s_0 \aeq{C} t_0 $ and $\Upsilon_{\catnew{c}} \vdash s_1\aeq{C} t_1$; or derivations of $\Upsilon_{\catnew{c}}\vdash s_0 \aeq{C} t_1 $ and $\Upsilon_{\catnew{c}} \vdash s_1\aeq{C} t_0$.

The rule $(\frule{\faeq{C}}{var})$ for handling variables is more involved. It is inspired by  Lemma~\ref{lemma:generated-group} and Lemma~\ref{lemma:pitts-eq-generalized}.
Applying $\pi'^{-1}$ to both sides of $\pi\act X \aeq{C} \pi'\act X$ yields $\pi'^{-1} \act (\pi \act X) \aeq{C} X$, which simplifies to $(\pi'^{-1}\circ\pi) \cdot X \aeq{C} X$. At this point, it is essential to verify whether $\pi'^{-1}\circ\pi$ belongs to the {\em group generated by permutations derived from the context $\Upsilon_{\catnew{c}}|_X$}, denoted as $\PN{}{\Upsilon_{\catnew{c}}|_X}$, whose permutations are known to fix $X$. Here
\begin{equation}\label{permutation-group-X}
    \PN{}{\Upsilon_{\catnew{c}}|_X} = \Perm{\atm{(\Upsilon_{\catnew{c}}|_X)_{\fresh}}}\alert{\circ}\pair{\perm{}{(\Upsilon_{\catnew{c}}|_X)_{\fix{C}}}},
\end{equation}
which is a subgroup of $\Perm{\A}$ because $\rho_1\circ\rho_2 = \rho_2\circ\rho_1$ holds for all $\rho_1 \in \Perm{\atm{(\Upsilon_{\catnew{c}}|_X)_{\fresh}}}$ and all $\rho_2\in\pair{\perm{}{(\Upsilon_{\catnew{c}}|_X)_{\fix{C}}}}$ since $\rho_1$ and $\rho_2$ are disjoint.
The remaining rules are standard and rely on the decomposition of the constraint over the structure of its terms.
\begin{example} \label{examples-judgements} 
The judgement $\new \atnew{c}.\{\newswap{d \ e}{c}\fix{C} X, (a \ b)\fix{C} Y\}\vdash \tf{f^C}([d]X,((a \ b)\act Y)\aeq{C} \tf{f^C}(Y,[e]X)$ is derivable. For $\Upsilon_{\atnew{c}}=\new \atnew{c}.\{\newswap{d \ e}{c}\fix{C} X, (a \ b)\fix{C} Y\}$, we have the following derivation:
\vspace{-4mm}
{\small 
          \begin{prooftree}
            \AxiomC{$(a \ b)\in \Perm{\{\atnew{c}\}}\alert{\circ}\pair{(a \ b)}$}
            \dashedLine

            \UnaryInfC{$(a \ b)\in \PN{}{\Upsilon_{\atnew{c}}|_Y}$}
            \RightLabel{$(\frule{\faeq{C}}{var})$}
            
            \UnaryInfC{$\Upsilon_{\atnew{c}} \vdash (a \ b)\act Y \aeq{C} Y$}
            
            \AxiomC{$(d \ e \ \atnew{c_1})\in \Perm{\{d,e, \atnew{c},\atnew{c_1}\}}\alert{\circ}\pair{\emptyset}$}
            \dashedLine
            \UnaryInfC{$(d \ e \ \atnew{c_1})\in \PN{}{\Upsilon_{\atnew{c},\atnew{c_1}}|_X}$}
            \RightLabel{$(\frule{\faeq{C}}{var})$}
        \UnaryInfC{$\Upsilon_{\atnew{c},\atnew{c_1}} \vdash \newswap{d}{c_1}\act X \aeq{C} \newswap{e}{c_1}\act X$}
        \RightLabel{$(\frule{\faeq{C}}{ab})$}
        \UnaryInfC{$\Upsilon_{\atnew{c},\atnew{c_1}} \vdash [d]X \aeq{C} [e]X$}
          \RightLabel{$(\frule{\faeq{C}}{\tf{f^C}})$}
             \BinaryInfC{$\Upsilon_{\atnew{c}} \vdash \tf{f^C}([d]X,(a \ b)\act Y)\aeq{C} \tf{f^C}(Y,[e]X)$}
         \end{prooftree}
         }
   While the freshness approach~\cite{DBLP:journals/tcs/UrbanPG04} requires $a \fresh Y$ and $b \fresh Y$ to derive $(a \ b)\act Y \aeq{C} Y$, the fixed-point calculus removes this limitation. Here, $(a \ b)\act Y \aeq{C} Y$ is derivable without relying on any freshness assumption. Yet, if desired, it could be derived using freshness; for instance, it could be derived using the context $\new\atnew{c}. \{(a \ b \ \atnew{c})\fix{C}Y \}$. This shows that the fixed-point calculus generalises and extends the freshness calculus, establishing it as proper subcalculus.
\end{example}




\subsection{Properties}\label{sec:properties}
Our calculus satisfies essential structural properties, as summarised in Theorem~\ref{lemma:pitts-eq-generalized}. It upholds {\em equivariance}, ensuring that derivations remain invariant under the action of permutations, thereby preserving consistency across transformations. It also satisfies {\em inversion}, a useful property for inductive proofs on derivations, as it facilitates the bottom-up application of derivation rules. More significantly, and in contrast with previous fixed-point approaches~\cite{DBLP:journals/lmcs/Ayala-RinconFN19,entics:14777}, our framework enables the definition of a sound and complete semantics for commutative nominal algebra with $\new$-quantified fixed-point constraints in nominal sets, as established in Theorem~\ref{thm:soundness-completeness-fix} below. The key insight to achieve this result lies in the $(\frule{\faeq{C}}{var})$ rule, more precisely in the construction of the group  $\PN{}{\Upsilon_{\catnew{c}}|_X}$ (see~(\ref{permutation-group-X})), which carefully distinguishes freshness information from fixed-point behaviour.

\begin{restatable}[Miscellaneous]{theorem}{miscellaneous}\label{thm:miscellaneous}
The following hold:
\begin{enumerate}
  \item (Inversion) \label{thm:inversion} The derivation rules from Figure~\ref{fig:fixed-rules_new} are invertible.
  \item (Equivariance) \label{thm:object-equivariance} $\Upsilon_{\catnew{c}} \vdash  s \aeq{C} t$ iff $\Upsilon_{\catnew{c}} \vdash  \rho\act s\aeq{C} \rho\act t$,
where $\rho$ is an arbitrary permutation (it could have atoms of $\catnew{c}$ in its domain).
  \item (Equivalence) \label{thm:alpha-equivalence} $\Upsilon_{\catnew{c}}\vdash - \aeq{C} -$ is an equivalence relation.
  \item (Strengthening) \label{thm:strengthening}  If $(\Upsilon\uplus\{\pi\fix{C} X\})_{\catnew{c}}\vdash s\aeq{C} t$ and $(\dom{\pi}\setminus\catnew{c})\cap\atm{s,t} = \emptyset$, then $\Upsilon_{\catnew{c}}\vdash s\aeq{C} t$.
 \end{enumerate}
\end{restatable}

\subsubsection{Soundness and Completeness.}\label{sec:soundness-completeness}
We first define the semantics of judgements. Below, we introduce the fundamental notions required to present our results, while the full construction of the semantic model is provided in the appendix.

    Given a signature $\Sigma$, which may contain commutative function symbols, a {\em commutative nominal $\Sigma$-algebra} $\nalg{A}$ (or just $\Sigma$-algebra) consists of:
    \begin{enumerate}
        \item A nominal set $\nom{A} = (|\nom{A}|,\act)$ $-$ the domain.

        \item An injective equivariant\footnote{Given $\Perm{\A}$-sets $\nom{X,Y}$, call a map $f: |\nom{X}|\to |\nom{Y}|$ {\em equivariant} when $\pi\act f(x) = f(\pi\act x)$ for all $\pi\in\Perm{\A}$ and $x\in |\nom{X}|$.} map $\atom^{\nalg{A}}:\A\to |\nom{A}|$ to interpret atoms; we write the interpretation $\atom(a)$ as $a^{\nalg{A}}\in |\nom{A}|$.

        \item An equivariant map $\abs^{\nalg{A}}:\A\times |\nom{A}|\to |\nom{A}|$ such that $\new\atnew{c}. \newswap{a}{c}\act \abs^{\nalg{A}}(a,x) = \abs^{\nalg{A}}(a,x)$, for all $a\in \A$ and $x\in|\nom{A}|$; we use this to interpret abstraction.

        \item An equivariant map $f^{\nalg{A}}:|\nom{A}|^n \to |\nom{A}|$, for each $\tf{f}:n$ in $\Sigma$.

        \item For each $\tf{f^C}$ in $\Sigma$, the associated equivariant map $f^{\C,\nalg{A}}$ satisfies $f^{\C,\nalg{A}}(x_1, x_2) = f^{\C,\nalg{A}}(x_2, x_1)$ for all $x_1, x_2 \in |\nom{A}|$. Sometimes, $\nalg{A}$ will be called a {\em nominal model of $\C$ (over $\Sigma$)}.
    \end{enumerate}
    We use $\nalg{A,B,C}$ to denote $\Sigma$-algebras.

A {\em valuation} $\varsigma$ in a commutative $\Sigma$-algebra $\nalg{A}$ maps unknowns $X\in\V$ to elements
$\varsigma(X) \in |\nom{A}|$. Below we define an equivariant function $\Int{\act}{\nalg{A}}{\varsigma}: \F(\Sigma,\V)\to |\nom{A}|$ to interpret terms w.r.t. a valuation $\varsigma$.

\begin{definition}[Interpretation]\label{def:interpretation-of-terms}
    Let $\nalg{A}$ be a commutative $\Sigma$-algebra. Suppose that $t \in \F(\Sigma, \V)$ and consider a valuation $\varsigma$ in $\nalg{A}$. The {\em interpretation} $\Int{t}{\nalg{A}}{\varsigma}$ is defined inductively as:
    \[
    \begin{array}{l@{\hspace{1cm}}r}
    \begin{array}{rcl}
            \Int{a}{\nalg{A}}{\varsigma} &=&  a^{\nalg{A}}\\
            \Int{[a]t}{\nalg{A}}{\varsigma} &=&\abs^{\nalg{A}}(a,\Int{t}{\nalg{A}}{\varsigma})\\
            \Int{{\tt f^{\C}}(t_1,t_2)}{\nalg{A}}{\varsigma} &=& f^{\C,\nalg{A}}(\Int{t_1}{\nalg{A}}{\varsigma}, \Int{t_2}{\nalg{A}}{\varsigma})\\
        \end{array} &
        \begin{array}{rcl}
            \Int{ \pi\act X}{\nalg{A}}{\varsigma} & =& \pi\act \varsigma(X)\\
            \Int{{\tt f}(t_1,\ldots,t_n)}{\nalg{A}}{\varsigma} & =& f^{\nalg{A}}(\Int{t_1}{\nalg{A}}{\varsigma},\ldots, \Int{t_n}{\nalg{A}}{\varsigma}) \\
        \end{array}
        \end{array}
\]
\end{definition}
%
The interpretation map $\Int{\cdot}{\nalg{A}}{\varsigma}$ is  equivariant. The next definition gives a semantics to our derivations within nominal sets.
%
    For any commutative $\Sigma$-algebra $\nalg{A}$:
        \begin{itemize}
           \item $\Int{\Upsilon_{\catnew{c}}}{\nalg{A}}{\varsigma}$ is {\it valid} if, and only if,  $\newc{c}{}.\pi\act\Int{X}{\nalg{A}}{\varsigma} = \Int{X}{\nalg{A}}{\varsigma}$, for each $\pi\fix{C} X\in\Upsilon$;
          \item  $\Int{\Upsilon_{\catnew{c}}\vdash t \aeq{C} u}{\nalg{A}}{\varsigma}$  is {\it valid}  if, and only if, $\Int{\Upsilon_{\catnew{c}}}{\nalg{A}}{\varsigma}$ (valid) implies $\Int{t}{\nalg{A}}{\varsigma} = \Int{u}{\nalg{A}}{\varsigma}$.
        \end{itemize}
  More generally, $\Int{\Upsilon_{\catnew{c}}\vdash t \aeq{C} u}{\nalg{A}}{}$ is {\em valid} iff  $\Int{\Upsilon_{\catnew{c}}\vdash t \aeq{C} u}{\nalg{A}}{\varsigma}$ is valid for all valuation $\varsigma$ .
 The semantics of equality constraints is defined for all commutative $\Sigma$-algebras: $\Upsilon_{\catnew{c}} \vDash s\aeq{C} t$ when $\Int{\Upsilon_{\catnew{c}} \vdash s\aeq{C} t}{\nalg{A}}{}$ is valid for all $\nalg{A}$ model of $\C$.


\begin{example}
    An example of an algebra involves the set $\pow{\tf{fin}}{\A}$. When equipped with the maps $\atm{a} = \{a\},\abs(S) = S\setminus\{a\}$, and $f(S_1,\ldots,S_n) = \bigcap_i S_i$ for every $\tf{f}:n$ in $\Sigma$, it forms a model of $\C$.
\end{example}

The proof system is sound and complete with respect to the nominal set semantics. 
In previous work, the $(\frule{\faeq{C}}{var})$ rule was simpler: it checked only the domains of the permutations involved, yielding a restricted notion of soundness. The construction of the  subgroup $\PN{}{\Upsilon_{\catnew{c}}|_X}$  in $(\frule{\faeq{C}}{var})$  is essential to obtain:

\begin{restatable}{theorem}{soundcomplete}\label{thm:soundness-completeness-fix}
     For any $\Upsilon_{\catnew{c}},t,u$, the following hold:
    \begin{enumerate}
        \item (Soundness) \label{thm:soundness-fix} If $\Upsilon_{\catnew{c}} \vdash t\aeq{C} u$ then $\Upsilon_{\catnew{c}} \vDash t\aeq{C} u$.

        \item (Completeness) \label{thm:completeness-fix} If $\Upsilon_{\catnew{c}} \vDash t\aeq{C} u$ then $\Upsilon_{\catnew{c}} \vdash t\aeq{C} u$.
    \end{enumerate}
\end{restatable}

\section{Nominal \texorpdfstring{$\C$}{C}-Unification}\label{sec:nominal-c-unification}

In this section, 
we present a rule-based algorithm for solving $\C$-unification problems using fixed-point constraints, proving its termination and correctness.



 \begin{definition}[Nominal $\C$-Unification Problem] A \emph{(nominal) $\C$-unification problem} $\probc$ has the form $\newc{c}{}.Pr$,  where $Pr$ consists of a finite set of (unification) constraints of the form $s\aeq{C}^?t$.
\end{definition}

Solutions for nominal $\C$-unification problems will be expressed using pairs  $\npair{\Psi, \sigma}{\atnew{\pvec{c}'}}$ consisting of a generalised fixed-point context $\Psi_{\atnew{\pvec{c}'}}$ and a substitution $\sigma$, where occurrences of the atoms $\atnew{\pvec{c}'}$ in the image of $\sigma$ are $\new$-quantified.
We call such $\sigma$ a $\new$-substitution, denoted as $\newc{c}{}.\sigma$ or $\sigma_{\atnew{\pvec{c}}}$.  We will omit the subscript and write only $\sigma$ to avoid clutter. 
We will come back to this later in \S\ref{ssec:unif-rules}.

 \begin{definition}[$\C$-Solution/Unifier]
 A  $\C$-unifier ($\C$-{\em solution}) for a  nominal $\C$-unification problem $\probc = \newc{c}{}.Pr$ is a pair $\npair{\Psi, \sigma}{\atnew{\pvec{c}'}}$, such that  $\catnew{c} \subseteq \atnew{\pvec{c}'}$ and
 $\Psi_{\atnew{\pvec{c}'}} \vdash s\sigma \aeq{C}t\sigma$, for each $s\aeq{C}^? t\in Pr$.
 If no $\C$-unifier $\npair{\Psi, \sigma}{\atnew{\pvec{c}'}}$ exists for $\probc$, the problem $\probc$ is {\em unsolvable}. We denote the set of all $\C$-unifiers of $\probc$ by $\mathcal{U}(\probc)$. 
\end{definition}
For simplicity, we will omit $\C$ and primarily refer to $\C$-unifiers and $\C$-solutions of $\C$-unification problems as unifiers and solutions, respectively.
The solutions to a problem $\probc$, elements of the set $\mathcal{U}(\probc)$, can be compared using the instantiation order $\ins{}$ (quasi-order) defined as follows:

\begin{definition}[Instantiation order]\label{def:ordering}
    Given  $\npair{\Psi_1,\sigma_1}{\atnew{\vec{c}_1}}$ and $\npair{\Psi_2,\sigma_2}{\atnew{\vec{c}_2}}$,
    we write $\npair{\Psi_1,\sigma_1}{\atnew{\vec{c}_1}} \ins{} \npair{\Psi_2,\sigma_2}{\atnew{\vec{c}_2}}$ when $\atnew{\vec{c}_1}\subseteq \atnew{\vec{c}_2}$ and there exist a substitution $\delta$ such that  $(\Psi_2)_{\atnew{\vec{c}_2}}\vdash \Psi_1\delta$ and $(\Psi_2)_{\atnew{\vec{c}_2}}\vdash X\sigma_2\aeq{C} X\sigma_1\delta$, for all $X\in \V$.
\end{definition}
%
%
For example, for $\sigma = [Y\mapsto a, X\mapsto Z]$, $\sigma' = [Y\mapsto a, X\mapsto \tf{f}(Y,d), Z\mapsto \tf{f}(Y,d)]$,  we can verify that
\(
    \npair{\{\newswap{a}{c} \fix{C} Z\}, \sigma}{\atnew{c}}\ins{}\npair{\{\newswap{a}{c}\fix{C} Z, \newswap{a}{c_1}\fix{C} Y\}, \sigma'}{\atnew{c},\atnew{c_1}}
\).
In fact, for $\delta=[Z\mapsto \tf{f}(Y,d)]$:
    \begin{enumerate}
        \item $\new \atnew{c},\atnew{c_1}.\{\newswap{a}{c}\fix{C}Z, \newswap{a}{c_1}\fix{C} Y\}\vdash W\sigma' \aeq{c} W\sigma\delta$, for all $W\in \V$.
        \item $\new \atnew{c},\atnew{c_1}.\{\newswap{a}{c}\fix{C} Z, \newswap{a}{c_1}\fix{C} Y\}\vdash \newswap{a}{c}\fix{C} Z\delta$.
    \end{enumerate}
Unlike nominal unification\cite{DBLP:journals/tcs/UrbanPG04} (without equational theories), solvable nominal $\C$-unification problems need not have just one most general unifier. Thus, the role of the most general unifier  is taken on by a complete set of $\C$-unifiers:
\begin{definition}[Complete Set of $\C$-unifiers]
Let $\probc$ be a $\C$-unification problem. A {\em complete set of $\C$-unifiers of $\probc$} is a set of pairs $S$ that satisfies
\begin{enumerate}
    \item $S\subseteq \U{\probc}$;
     \item  For all $\npair{\Psi,\sigma}{\atnew{\pvec{c}'}} \in \U{\probc}$ there exists $\npair{\Phi,\sigma'}{\atnew{\pvec{c}''}}\in S$ s.t.\ $\npair{\Phi,\sigma'}{\atnew{\pvec{c}''}}\ins{} \npair{\Psi,\sigma}{\atnew{\pvec{c}'}}$.
\end{enumerate}
\end{definition}




\subsection{$\texttt{Unify}_{\C}$: a rule-based algorithm for $\C$-unification}\label{ssec:unif-rules}

\begin{figure}[!t]
    \centering
    \begin{tabular}{@{}l@{ }l@{ \hspace{-0.4cm}}c@{ }l@{}}
        \hline
        &\\[-0.3cm]

        $(del)$ & $\newc{c}{}.Pr\uplus\{t\aeq{C}^? t\}$ & $\Longrightarrow$ & $\newc{c}{}.Pr$ \\

        $(\tf{f})$ & $\newc{c}{}.Pr\uplus\{\tf{f}(\tilde{t})_n \aeq{C}^{?} \tf{f}(\tilde{t'})_n\}$ & $\Longrightarrow$ & $\newc{c}{}.Pr\cup\{t_1\aeq{C}^{?} t_1',\ldots,t_n\aeq{C}^{?} t_n'\}$ \\

        $(\tf{f^C})$ & $\newc{c}{}.Pr\uplus\{\tf{f^C}(t_0,t_1) \aeq{C}^{?} \tf{f^C}(s_0,s_1)\}$ & $\Longrightarrow$ & $
        \left\{\begin{array}{@{}l@{}}
            \newc{c}{}.Pr\cup\{t_0\aeq{C}^{?} s_0,t_1\aeq{C}^{?} s_1\},\\[0.05cm]
            \newc{c}{}.Pr\cup\{t_0\aeq{C}^{?} s_1,t_1\aeq{C}^{?} s_0\}
        \end{array}
        \right.$\\[0.5cm]

        $([a])$ & $\newc{c}{}.Pr\uplus\{[a]t \aeq{C}^{?} [a]t'\}$ & $\Longrightarrow$ & $\newc{c}{}.Pr\cup\{t\aeq{C}^{?} t'\}$\\

        $(ab)$ & $\newc{c}{}.Pr\uplus\{[a]t \aeq{C}^{?} [b]t'\}$ & $\Longrightarrow$ & $\new\catnew{c},\atnew{c_1}.Pr\cup\{\newswap{a}{c_1}\act t\aeq{C}^{?} \newswap{b}{c_1}\act t'\}$ \\

        $(var)$ & $\newc{c}{}.Pr\uplus\{\pi\act X \aeq{C}^{?} \pi'\act X\}$ & $\Longrightarrow$ & $\newc{c}{}.Pr\cup\{(\pi'^{-1}\circ\pi)\act X\aeq{C}^? X\}$, \\

        & & & if $\pi'\neq \id$\\[0.1cm]

       $(inst_1)$ & $\newc{c}{}.Pr\uplus\{\pi\act X \aeq{C}^{?} t\}$ & $\overset{[X\mapsto \pi^{-1}\act t]}{\Longrightarrow}$ & $\newc{c}{}.Pr\{X\mapsto \pi^{-1}\act t\}$, if $X\notin\var{t}$\\[0.1cm]

       $(inst_2)$& $\newc{c}{}.Pr\uplus\{t\aeq{C}^{?} \pi\act X \}$ & $\overset{[X\mapsto \pi^{-1}\act t]}{\Longrightarrow}$ & $\newc{c}{}.Pr\{X\mapsto \pi^{-1}\act t\}$, if $X\notin\var{t}$\\
        &\\[-0.3cm]
        \hline
    \end{tabular}
    \caption{Simplification rules for $\aeq{C}$. The notation $(\tilde{t})_n$ abbreviates $(t_1,\ldots,t_n)$.}
    \label{fig:c-simp-rules}
\end{figure}

In this subsection, we introduce a set of \emph{simplification rules} (Figure~\ref{fig:c-simp-rules}) for ``solving'' nominal $\C$-unification problems. These rules simplify problems by transforming constraints into ``simpler'' ones, by running the derivation rules from Figure~\ref{fig:fixed-rules_new} in reverse. The rule $(\frule{\faeq{C}}{\tf{f^C}})$ splits a $\C$-unification problem into two, with the original problem solvable if at least one of the new problems is solvable. This requires working with finite sets of $\C$-unification problems, called \emph{extended $\C$-unification problems}, denoted as $\mathcal{M}$.

The algorithm ${\tt Unif}_{\C}$ starts with $\exprob=\{\probc\}$ and exhaustively applies simplification rules. 
We denote by $\nf{\exprob}$ the normal forms of $\exprob$ by the  reduction relation $\Longrightarrow$ generated by the simplification rules. 

\begin{restatable}[Termination]{theorem}{termination}\label{thm:termination}
The relation $\Longrightarrow$ is terminating.
\end{restatable}
Some constraints in $\probc$ cannot be simplified, we say that they are {\em reduced}. An equality constraint $s \aeq{C}^? t$ is \emph{reduced} if it satisfies one of the following:
\begin{itemize}
    \item $s$ is a suspension on a variable and $t$ is that variable (e.g., $\pi\act X \aeq{C}^? X$).
    \item $s$ and $t$ are distinct atoms (e.g. $a\aeq{C}^? b$).
    \item $s$ and $t$ have different head symbols (e.g., $\tf{f}(u) \aeq{C}^? \tf{g}(v)$).
    \item $s$ and $t$ have different root constructors (e.g., $[a]s' \aeq{C}^? \tf{f}(t)$).
    \item $s$ is a suspension on some variable, $t$ is not a suspension, but $t$ contains occurrences of that variable (e.g. $X \aeq{C}^? \tf{f}(X,a)$).
\end{itemize}

The first form, which is a fixed-point equation, is \emph{consistent}; the others are \emph{inconsistent}.  A problem $\probc = \newc{c}{}.Pr$ is \emph{reduced} if all constraints in $Pr$ are reduced.  
A problem $\probc$  is called {\em solvable} if its normal form $ \nf{\probc} $ contains a reduced problem where all the constraints are consistent.
In this case, we compute $\sol{\probc}$ as the set of pairs $\npair{\Upsilon, \sigma}{\atnew{\pvec{c}'}}$ 
where 
$\atnew{\pvec{c}'}$ includes the set of names generated in the simplification process (i.e. $\atnew{\pvec{c}}\subseteq \atnew{\pvec{c}'}$), 
$\Upsilon_{\atnew{\pvec{c}'}} = \new \atnew{\pvec{c}'}.\{\pi\fix{C} X \mid \pi\act X \aeq{C}^? X \in {\probc'} \text{~solvable~and~}  {\probc'}\in \nf{\probc} \}$, and $\sigma$ is either $Id$ (if no instantiation rules were used) or the composition of substitutions used throughout the simplification of a branch starting in $\probc$ and ending in $\probc'$.

 \begin{example}
    Suppose $\Sigma = \{\wedge:2\}$. Let $\probc = \{[a](X\wedge Y) \aeq{C}^? [b](Y\wedge X)\}$. Then
    \[
        \probc \overset{*}{\Longrightarrow} \{ \new \atnew{c_1}.\{(a \ \atnew{c_1} \ b)\act X\aeq{C}^?  X\},\new \atnew{c_1}.\{(a \ b \ \atnew{c_1})\act X\aeq{C}^? X, (a \ b \ \atnew{c_1})\act Y \aeq{C}^? Y\}\}.
    \]
    The algorithm computes two solutions:
  $\npair{\{(a \ \atnew{c_1} \ b)\fix{C}  X\},[Y\mapsto (a \ \atnew{c_1} \ b)\act X]}{\atnew{c_1}}$ and 
 $ \npair{\{(a \ b \ \atnew{c_1})\fix{C} X, (a \ b \ \atnew{c_1})\fix{C} Y\},Id}{\atnew{c_1}}$ (a complete set of $\C$-solutions of $\probc$). Note that, by taking $\delta = [Y\mapsto X]$, the first  is an instance of the second.
\end{example}

\begin{remark}\label{rmk:substitution}
Substitutions feature $\new$-quantified names, as in $\sigma = [Y \mapsto (a \ \atnew{c_1} \ b) \act X]$. This is a desirable property that shows that solutions are stable under the choice of new names for renaming that can be done throughout computation.  This is achieved using the concept of a $\new$-substitution,  
which provides a systematic way to track freshness information.
\end{remark}

The next result shows that if a $\C$-unification problem $\probc$ is solvable, our algorithm ${\tt Unif}_{\C}$ will output  $\sol{\probc}$ which is a complete set of solutions of $\probc$. 
\begin{restatable}[Correctness]{theorem}{correctness}\label{thm:correctness-simp-rules}
    Let $\probc$ be a problem such that $\probc\overset{*}{\Longrightarrow} \nf{\probc}$. If $\probc$ is solvable, then the following hold: 
    \begin{enumerate}
        \item \textnormal{(Soundness)} $\sol{\probc}\subseteq \U{\probc}$.
       
        \item \textnormal{(Completeness)} If $\npair{\Phi,\tau}{\atnew{\pvec{c}''}}\in \U{\probc}$ then there exists $\npair{\Psi,\sigma}{\atnew{\pvec{c}'}}\in\sol{\probc}$ s.t. $\npair{\Psi,\sigma}{\atnew{\pvec{c}'}}\ins{} \npair{\Phi,\tau}{\atnew{\pvec{c}''}}$. The set $\sol{\probc}$ is a complete set of $\C$-unifiers of $\probc$.
    \end{enumerate}
\end{restatable}

A direct consequence of termination and correctness is that the complete set of $\C$-unifiers of a $\probc$ is always finite.
\begin{corollary}
   $\C$-unification using generalised fixed-point constraints is finitary.
\end{corollary}
\section{Related Work.}
Gabbay and Mathijssen~\cite{DBLP:journals/logcom/GabbayM09} define a sound and complete proof system for nominal algebra, where derivation rules are subject to freshness conditions,
and provide a semantics in the class of nominal sets. 
Nominal sets also give an absolute denotation for first-order logic~\cite{DBLP:journals/jacm/Gabbay16}, and more recently for predicate logic~\cite{dowek2023nominalsemanticspredicatelogic}.

 Applications of nominal techniques in languages where operators satisfy equational axioms were studied in various contexts, using freshness constraints in the definition of $\alpha$-equivalence: e.g., to specify syntactic reasoning algorithms to  show properties of program transformations on higher-order expressions in call-by-need functional languages, taking into account e.g. {\tt letrec} constructs~\cite{DBLP:journals/fuin/Schmidt-Schauss22};  in theorem-proving, due to nominal logic’s equivariance property,  a different form of unification,  called equivariant unification, was investigated in~\cite{DBLP:journals/jar/Cheney10};
 applications in theorem-proving and rewrite-based programming languages also require the use e.g. of {\tt A} and/or {\tt C} operators~\cite{AYALARINCON20193}.  Nominal unification algorithms modulo commutativity are available, which generate infinite complete sets of solutions expressed using freshness constraints and substitutions~\cite{DBLP:conf/frocos/Ayala-RinconSFN17,   DBLP:journals/mscs/Ayala-RinconSFS21}. 
The finitary  algorithm  developed here opens the path for  extensions of programming languages such as  $\alpha$-prolog~\cite{DBLP:conf/iclp/CheneyU04}, building commutativity into the unification algorithm.

Alternative approaches to representing binders, beyond the already-mentioned HOAS, include the $\lambda\sigma$-calculus\cite{DBLP:journals/jfp/AbadiCCL91} and the locally nameless  representation~\cite{DBLP:journals/jar/Chargueraud12}. The $\lambda\sigma$-calculus is based on the definition of  explicit capture-avoiding substitutions. However, a significant drawback is that it does not support reasoning about open terms, as it lacks confluence when applied to terms with meta-variables.
Handling meta-variables effectively is essential for reasoning about binders in  in proof assistants, such as AutoSubst2~\cite{DBLP:conf/cpp/StarkSK19}, which provides support for formalising metatheory in Coq. It is based on the  $\lambda\sigma$-calculus and could benefit from an alternative approach to computing substitutions for meta-variables—one that systematically accounts for the required renamings.


\section{Conclusion}

We introduced a novel framework for equational reasoning modulo commutativity in languages with binders, using nominal techniques with ($\new$-quantified) permutation fixed-point constraints. This approach integrates commutativity into $\alpha$-equivalence reasoning without relying on freshness constraints, ensuring soundness and completeness in nominal set semantics.
%
Additionally, we proposed a terminating and correct $\C$-unification algorithm that provides finite complete set of solutions, addressing a key limitation of previous methods that produced infinite solutions when using freshness constraints. Our results establish a solid foundation for reasoning about syntax with binding and equational properties, with potential applications in theorem proving, programming language semantics, and automated reasoning systems.

Future work includes extending our approach to broader equational theories, e.g. integrating associativity along commutativity, and exploring practical implementations in verification tools.

%
%
%
 \bibliographystyle{splncs04}
 \bibliography{references}

\begin{thebibliography}{10}
\providecommand{\url}[1]{\texttt{#1}}
\providecommand{\urlprefix}{URL }
\providecommand{\doi}[1]{https://doi.org/#1}

\bibitem{DBLP:journals/jfp/AbadiCCL91}
Abadi, M., Cardelli, L., Curien, P., L{\'{e}}vy, J.: Explicit substitutions. J.
  Funct. Program.  \textbf{1}(4),  375--416 (1991),
  \url{https://doi.org/10.1017/S0956796800000186}

\bibitem{book/Pitts}
{Andrew M. Pitts}: ``Nominal Sets: Names and Symmetry in Computer Science''.
  Cambridge Tracts in Theoretical Computer Science, Cambridge University Press
  (2013), \url{https://doi.org/10.1017/CBO9781139084673}

\bibitem{DBLP:conf/frocos/Ayala-RinconSFN17}
Ayala{-}Rinc{\'{o}}n, M., de~Carvalho{-}Segundo, W., Fern{\'{a}}ndez, M.,
  Nantes{-}Sobrinho, D.: On {S}olving {N}ominal {F}ixpoint {E}quations. In:
  Dixon, C., Finger, M. (eds.) Frontiers of Combining Systems - 11th
  International Symposium, FroCoS 2017, Bras{\'{\i}}lia, Brazil, September
  27-29, 2017, Proceedings. Lecture Notes in Computer Science, vol. 10483, pp.
  209--226. Springer (2017),
  \url{https://doi.org/10.1007/978-3-319-66167-4\_12}

\bibitem{DBLP:journals/mscs/Ayala-RinconSFS21}
Ayala{-}Rinc{\'{o}}n, M., de~Carvalho{-}Segundo, W., Fern{\'{a}}ndez, M.,
  Silva, G.F., Nantes{-}Sobrinho, D.: Formalising nominal {C}-unification
  generalised with protected variables. Math. Struct. Comput. Sci.
  \textbf{31}(3),  286--311 (2021),
  \url{https://doi.org/10.1017/S0960129521000050}

\bibitem{DBLP:journals/lmcs/Ayala-RinconFN19}
Ayala{-}Rinc{\'{o}}n, M., Fern{\'{a}}ndez, M., Nantes{-}Sobrinho, D.: On
  {N}ominal {S}yntax and {P}ermutation {F}ixed {P}oints. Log. Methods Comput.
  Sci.  \textbf{16}(1),  19:1–19:36 (2020),
  \url{https://doi.org/10.23638/LMCS-16(1:19)2020}

\bibitem{AYALARINCON20193}
Ayala-Rincón, M., {de Carvalho-Segundo}, W., Fernández, M., Nantes-Sobrinho,
  D., Rocha-Oliveira, A.C.: A formalisation of nominal alpha-equivalence with
  {A}, {C}, and {AC} function symbols. Theoretical Computer Science
  \textbf{781},  3--23 (2019). \doi{https://doi.org/10.1016/j.tcs.2019.02.020}

\bibitem{DBLP:conf/tphol/AydemirBFFPSVWWZ05}
Aydemir, B.E., Bohannon, A., Fairbairn, M., Foster, J.N., Pierce, B.C., Sewell,
  P., Vytiniotis, D., Washburn, G., Weirich, S., Zdancewic, S.: Mechanized
  {M}etatheory for the {M}asses: The {P}oplmark {C}hallenge. In: Hurd, J.,
  Melham, T.F. (eds.) Theorem Proving in Higher Order Logics, 18th
  International Conference, TPHOLs 2005, Oxford, UK, August 22-25, 2005,
  Proceedings. Lecture Notes in Computer Science, vol.~3603, pp. 50--65.
  Springer (2005), \url{https://doi.org/10.1007/11541868\_4}

\bibitem{entics:14777}
Caires-Santos, A.K., Fernández, M., Nantes-Sobrinho, D.: Strong nominal
  semantics for fixed-point constraints. Electronic Notes in Theoretical
  Informatics and Computer Science, Proceedings of MFPS XL.  \textbf{4}, 6 (Dec
  2024), \url{https://entics.episciences.org/14777}

\bibitem{arxiv/cairessantos2025}
Caires-Santos, A.K., Fernández, M., Nantes-Sobrinho, D.: Equational reasoning
  modulo commutativity in languages with binders (extended version) (2025),
  \url{https://arxiv.org/abs/2502.19287}

\bibitem{DBLP:journals/jar/Chargueraud12}
Chargu{\'{e}}raud, A.: The locally nameless representation. J. Autom. Reason.
  \textbf{49}(3),  363--408 (2012),
  \url{https://doi.org/10.1007/s10817-011-9225-2}

\bibitem{DBLP:journals/jar/Cheney10}
Cheney, J.: Equivariant {U}nification. J. Autom. Reason.  \textbf{45}(3),
  267--300 (2010), \url{https://doi.org/10.1007/s10817-009-9164-3}

\bibitem{DBLP:conf/iclp/CheneyU04}
Cheney, J., Urban, C.: alpha-prolog: {A} {L}ogic {P}rogramming {L}anguage with
  {N}ames, {B}inding and alpha-{E}quivalence. In: Demoen, B., Lifschitz, V.
  (eds.) Logic Programming, 20th International Conference, {ICLP} 2004,
  Saint-Malo, France, September 6-10, 2004, Proceedings. Lecture Notes in
  Computer Science, vol.~3132, pp. 269--283. Springer (2004),
  \url{https://doi.org/10.1007/978-3-540-27775-0\_19}

\bibitem{dowek2023nominalsemanticspredicatelogic}
Dowek, G., Gabbay, M.J.: Nominal semantics for predicate logic: algebras,
  substitution, quantifiers, and limits (2023),
  \url{https://arxiv.org/abs/2312.16487}

\bibitem{DBLP:journals/jar/FeltyMP15}
Felty, A.P., Momigliano, A., Pientka, B.: The next 700 challenge problems for
  reasoning with higher-order abstract syntax representations - part 2 - {A}
  survey. J. Autom. Reason.  \textbf{55}(4),  307--372 (2015),
  \url{https://doi.org/10.1007/s10817-015-9327-3}

\bibitem{DBLP:journals/logcom/Gabbay09}
Gabbay, M.J.: Nominal {A}lgebra and the {HSP} {T}heorem. J. Log. Comput.
  \textbf{19}(2),  341--367 (2009), \url{https://doi.org/10.1093/logcom/exn055}

\bibitem{DBLP:journals/bsl/Gabbay11}
Gabbay, M.J.: Foundations of {N}ominal {T}echniques: {L}ogic and {S}emantics of
  {V}ariables in {A}bstract {S}yntax. Bull. Symb. Log.  \textbf{17}(2),
  161--229 (2011), \url{https://doi.org/10.2178/bsl/1305810911}

\bibitem{DBLP:journals/jacm/Gabbay16}
Gabbay, M.J.: Semantics {O}ut of {C}ontext: {N}ominal {A}bsolute {D}enotations
  for {F}irst-{O}rder {L}ogic and {C}omputation. J. {ACM}  \textbf{63}(3),
  25:1--25:66 (2016), \url{https://doi.org/10.1145/2700819}

\bibitem{DBLP:journals/logcom/GabbayM09}
Gabbay, M.J., Mathijssen, A.: Nominal ({U}niversal) {A}lgebra: {E}quational
  {L}ogic with {N}ames and {B}inding. J. Log. Comput.  \textbf{19}(6),
  1455--1508 (2009), \url{https://doi.org/10.1093/logcom/exp033}

\bibitem{DBLP:journals/fac/GabbayP02}
Gabbay, M.J., Pitts, A.M.: A {N}ew {A}pproach to {A}bstract {S}yntax with
  {V}ariable {B}inding. Formal Aspects Comput.  \textbf{13}(3-5),  341--363
  (2002), \url{https://doi.org/10.1007/s001650200016}

\bibitem{DBLP:journals/corr/abs-1006-3027}
Kurz, A., Petrisan, D., Velebil, J.: Algebraic theories over nominal sets. CoRR
   \textbf{abs/1006.3027} (2010), \url{http://arxiv.org/abs/1006.3027}

\bibitem{10.1561/2500000017}
Murawski, A.S., Tzevelekos, N.: Nominal game semantics. Found. Trends Program.
  Lang.  \textbf{2}(4),  191–269 (Mar 2016),
  \url{https://doi.org/10.1561/2500000017}

\bibitem{DBLP:journals/fuin/Schmidt-Schauss22}
Schmidt{-}Schau{\ss}, M., Kutsia, T., Levy, J., Villaret, M., Kutz, Y.D.K.:
  Nominal {U}nification and {M}atching of {H}igher {O}rder {E}xpressions with
  {R}ecursive {L}et. Fundam. Informaticae  \textbf{185}(3),  247--283 (2022),
  \url{https://doi.org/10.3233/FI-222110}

\bibitem{DBLP:conf/cpp/StarkSK19}
Stark, K., Sch{\"{a}}fer, S., Kaiser, J.: Autosubst 2: reasoning with
  multi-sorted de bruijn terms and vector substitutions. In: Mahboubi, A.,
  Myreen, M.O. (eds.) Proceedings of the 8th {ACM} {SIGPLAN} International
  Conference on Certified Programs and Proofs, {CPP} 2019, Cascais, Portugal,
  January 14-15, 2019. pp. 166--180. {ACM} (2019),
  \url{https://doi.org/10.1145/3293880.3294101}

\bibitem{DBLP:journals/tcs/UrbanPG04}
Urban, C., Pitts, A.M., Gabbay, M.J.: Nominal {U}nification. Theor. Comput.
  Sci.  \textbf{323}(1-3),  473--497 (2004),
  \url{https://doi.org/10.1016/j.tcs.2004.06.016}

\end{thebibliography}
%


\appendix

\section{Proofs of Section~\ref{sec:permutation-fix-nominal-sets}}\label{app:permutation-fix-nominal-sets}

\generated*

  \begin{proof}
    \begin{enumerate}
        \item 
   
     If $\pi$ is a product of disjoint cycles, then write $\pi = \mu_1\circ\ldots\circ\mu_n$. Note that $\eta$ is one of these $\mu_{i}$'s, say $\eta = \mu_{i_0}$. Then $a\in \dom{\mu_{i_0}} - \supp{}{x}$. Suppose, by contradiction, that there is an $b\in \dom{\mu_{i_0}}\cap \supp{}{x}$. Then  $b\in\supp{}{x}$ implies $\mu_{n}(b)\in\supp{}{\mu_n\act x}$ and since $\mu_n(b) = b$ this leads to $b\in\supp{}{\mu_n\act x}$, which in turn yields $\mu_{n-1}(b)\in\supp{}{(\mu_{n-1}\mu_n)\act x}$, that is equivalent to $b\in\supp{}{(\mu_{n-1}\mu_n)\act x}$ because $\mu_{n-1}(b) = b$. Proceeding this way, eventually we will reach $b\in\supp{}{(\mu_{i_0+1}\ldots\mu_{n-1}\mu_n)\act x}$. Applying $\mu_i$ on both sides yields
     \(
          \mu_{i_0}(b)\in\supp{}{(\mu_{i_0}\mu_{i_0+1}\ldots\mu_{n-1}\mu_n)\act x}.
     \)
     
     Since $\mu_{i_0-1}$ and $\mu_{i_0}$ are disjoint, we have $\mu_{i_0-1}(\mu_{i_0}(b)) = \mu_{i_0}(b)$ and so we get
     \(
          \mu_{i_0}(b)\in\supp{}{(\mu_{i_0-1}\mu_{i_0}\mu_{i_0+1}\ldots\mu_n)\act x}.
     \)
     Therefore, applying the rest of the cycles, we obtain $\mu_{i_0}(b)\in\supp{}{\pi\act x}$ and, consequently, $\mu_{i_0}(b)\in\supp{}{x}$ because $\pi\act x = x$. Since there is a $l>0$ such that $\mu_{i_0}^l(b) = a$, by repeating this argument as many times as necessary, we obtain $a = \mu_{i_0}^l(b) \in \supp{}{x}$, reaching a contradiction. 

     \item First, note that $\eta_i\act x = x$ iff $\eta_i^{-1}\act x = x$. Since $\rho\in \pair{\eta_1,\ldots,\eta_n}$, there is a $r\in \mathbb{N}$ such that $\rho = \zeta_1\zeta_2\ldots\zeta_r$ where $\zeta_i\in \{\eta_1,\ldots,\eta_n\}$ or $\zeta_i\in \{\eta_1^{-1},\ldots,\eta_n^{-1}\}$. Thus, $\zeta_i\act x = x$ for all $i=1,\ldots,r$. Therefore,
     \[
        (\zeta_1\zeta_2\ldots\zeta_r)\act x = (\zeta_1\zeta_2\ldots\zeta_{r-1})\act(\zeta_r\act x) =  (\zeta_1\zeta_2\ldots\zeta_{r-1})\act x = \ldots =  \zeta_1\act x = x.
     \]
    \end{enumerate}
 \end{proof}

\pitts*

\begin{proof}
    \begin{enumerate}
        \item Suppose $\newc{c}{}.\pi\act x = x$. By definition, $D_1=\{\catnew{c} \mid \pi\act x  = x\}$ is cofinite. Since the set $D_2 = \{\catnew{c} \mid \catnew{c}\cap\supp{}{x} = \emptyset\}$ is cofinite, it follows that $D_1\cap D_2$ is still cofinite. Take $\catnew{c}\in D_1\cap D_2$. Then $\pi\act x = x$ and $\catnew{c}\cap \supp{}{x} = \emptyset$ hold simultaneously. Since $\pi\act x = x$ and each cycle of $\pi_{\catnew{c}}$ mention atoms from $\catnew{c}$, it follows by Lemma~\ref{lemma:generated-group} that $\dom{\pi_{\catnew{c}}}\cap \supp{}{x} = \emptyset$ and, by the definition of support, $\pi_{\catnew{c}}\act  x = x$.

        Now, it remains to prove that $\pi_{\neg\catnew{c}}\act x = x$. Using the fact that $\pi_{\catnew{c}}$ and $\pi_{\catnew{c}}$ commute because they are disjoint, we get
        \begin{align*}
            \pi\act x = x &\Longrightarrow (\pi_{\catnew{c}}\circ\pi_{\neg\catnew{c}})\act x = x\\
            &\Longrightarrow (\pi_{\neg\catnew{c}}\circ \pi_{\catnew{c}})\act x = x\\
            &\Longrightarrow \pi_{\neg\catnew{c}}\act(\pi_{\catnew{c}}\act x) = x\\
            &\Longrightarrow \pi_{\neg\catnew{c}}\act x = x.
        \end{align*}
        Thus, the sets $\{\catnew{c} \mid \pi_{\catnew{c}}\act x = x\}$ and $\{\catnew{c} \mid \pi_{\neg\catnew{c}}\act  x = x\}$ are cofinite and the result follows.

        Conversely, suppose $\newc{c}{}.\pi_{\catnew{c}}\act x = x$ and $\newc{c}{}.\pi_{\neg\catnew{c}}\act x = x$. This implies that the sets $D_1 = \{\catnew{c} \mid \pi_{\catnew{c}}\act  x = x\}$ and $D_2 = \{\catnew{c} \mid \pi_{\neg\catnew{c}}\act  x = x\}$ are cofinite. By assumption, both $D_1$ and $D_2$ are cofinite. Consequently, their intersection $D_1 \cap D_2$ is also cofinite. For all $\catnew{c} \in D_1 \cap D_2$, $\pi_{\catnew{c}}\act x = x$ and $\pi_{\neg\catnew{c}}\act x = x$ hold simultaneously. Applying $\pi_{\neg\catnew{c}}$ to both sides of $\pi_{\catnew{c}}\act x = x$, we obtain $\pi\act x = x$. Therefore, $\{\catnew{c} \mid \pi\act  x = x\}$ is cofinite and so $\newc{c}{}.\pi\act x = x$.

        \item Suppose, by contradiction, that there is a $\catnew{c}$ satisfying $\catnew{c}\cap \supp{}{x} = \emptyset$ such that $\dom{\pi_{\catnew{c}}}\cap \supp{}{x} \neq \emptyset$. Then there is an atom $a$ (that cannot be in $\catnew{c}$) such that $a\in \dom{\pi_{\catnew{c}}}$ and $a\in \supp{}{x}$. Thus, for some $l\geq 1$, we can write
    \[
        \pi_{\catnew{c}} = (a\  \pi_{\catnew{c}}(a) \ \ldots \  \pi_{\catnew{c}}^l(a))\circ\rho
    \]
    where $\rho$ is some permutation such that $\dom{\rho}\cap \{a, \pi_{\catnew{c}}(a),\ldots,\pi_{\catnew{c}}^l(a)\} = \emptyset$. Since each cycle of $\pi_{\catnew{c}}$ mention at least one atom from $\catnew{c}$, it follows that there is some $1\leq k \leq l$ such that $\pi_{\catnew{c}}^k(a) \in \catnew{c}$. As a consequence, $\pi_{\catnew{c}}^k(a)\notin\supp{}{x}$. By Lemma~\ref{lemma:generated-group}, it follows that $a\notin \supp{}{x}$, a contradiction.

     Conversely, suppose that $\dom{\pi_{\catnew{c}}}\cap \supp{}{x} = \emptyset$ for all $\catnew{c}\cap\supp{}{x} = \emptyset$. By the definition of support, this is equivalent to say that $\pi_{\catnew{c}}\act x = x$ for all $\catnew{c}\cap\supp{}{x} = \emptyset$. In other words, $\{\catnew{c} \mid \catnew{c}\cap \supp{}{x} = \emptyset\} \subseteq \{\catnew{c} \mid \pi_{\catnew{c}}\act  x = x\}$. Therefore, $\{\catnew{c} \mid \pi_{\catnew{c}}\act  x = x\}$ is cofinite and so $\newc{c}{}.  \pi_{\catnew{c}}\act  x = x$.

     \item Assume $\newc{c}{}.\pi_{\neg\catnew{c}}\act x = x$. By definition, this is equivalent to say that $\{\catnew{c} \mid \pi_{\neg\catnew{c}}\act x \neq x\}$ is finite. Suppose, by contradiction, that $\pi_{\neg \atnew{\pvec{c}'}}\act x \neq x$ for some $\atnew{\pvec{c}'} = \atnew{c_1'},\ldots,\atnew{c_n'}$ such that $\atnew{\pvec{c}'}\cap\supp{}{x} = \emptyset$. Take another list $\atnew{\pvec{c}''} = \atnew{c_1''},\ldots,\atnew{c_n''}$ such that $\atnew{\pvec{c}''}\cap(\supp{}{x}\cup\dom{\pi}\cup\atnew{\pvec{c}'}) = \emptyset$ (there are a cofinite amount of them). Then $\pi_{\neg\atnew{\pvec{c}'}} = \pi_{\neg\atnew{\pvec{c}''}}$ and so $\pi_{\neg\atnew{\pvec{c}''}}\act x \neq x$, proving that the set $\{\catnew{c} \mid \pi_{\neg\catnew{c}}\act x\neq x\}$ is cofinite, a contradiction. The converse follows by the definition of $\new$.

    \end{enumerate}
\end{proof}

\section{Proofs of Section~\ref{sec:properties} - Properties}\label{app:properties}

\miscellaneous*

\begin{proof}
\begin{enumerate}
    \item We want to prove that the conclusion of each of these rules implies its respective premise. But this follows directly from the fact that each rule corresponds to a unique class of terms. For example, if $\Upsilon_{\catnew{c}} \vdash  \tf{f^C}(t_0,t_1) \faeq{C} \tf{f^C}(t_1,t_0)$, then there is a proof $\Pi$ of this judgment. However, the only possible rule applicable as a last step is $(\frule{\faeq{C}}{\tf{f^C}})$, so we either have $\Upsilon_{\catnew{c}}\vdash t_0 \aeq{C} t_0$ and $\Upsilon_{\catnew{c}} \vdash t_1 \aeq{C} t_1$, or $\Upsilon_{\catnew{c}} \vdash t_0 \aeq{C} t_1$ and $\Upsilon_{\catnew{c}} \vdash  t_1 \aeq{\C} t_0$. The same reasoning applies to all the other rules.

    \item It's sufficient to prove only the left-to-right case. The proof is by induction on the last rule applied.
        \begin{itemize}
            \item The last rule is $(\frule{\faeq{C}}{a})$. In this case, $\Upsilon_{\catnew{c}} \vdash  a \aeq{C}a$. By $(\frule{\faeq{C}}{a})$, we have $\Upsilon_{\catnew{c}} \vdash  \rho\act a \aeq{C}\rho\act a$.

            \item The last rule is $(\frule{\faeq{C}}{var})$. In this case, $\Upsilon_{\catnew{c}} \vdash \pi_1\act X \aeq{C} \pi_2\act X$. By Inversion (Theorem~\ref{thm:miscellaneous}(\ref{thm:inversion})), $\pi_2^{-1}\circ\pi_1\in \PN{}{\Upsilon_{\catnew{c}}|_X}$. Since $(\rho\circ\pi_2)^{-1}\circ(\rho\circ\pi_1) = \pi_2^{-1}\circ\pi_1$, the result follows by an application of rule $(\frule{\faeq{C}}{var})$.

            \item The rules $(\frule{\faeq{C}}{\tf{f}}),(\frule{\faeq{C}}{\tf{f^C}}), (\frule{\faeq{C}}{[a]})$ and $(\frule{\faeq{C}}{ab})$ follow by induction. Here we prove only the case of the rule $(\frule{\faeq{C}}{ab})$. In this case, we have $\Upsilon_{\catnew{c}} \vdash [a]t \aeq{C} [b]s$. By Inversion (Theorem~\ref{thm:miscellaneous}(\ref{thm:inversion})), we obtain $\Upsilon_{\catnew{c},\atnew{c_1}} \vdash  \newswap{a}{c_1}\act t \aeq{C} \newswap{b}{c_1}\act s$ where $\atnew{c_1}$ is taken such that it doesn't occur in $\Upsilon_{\catnew{c}},a,b,t,s,\rho$. By induction, we have $\Upsilon_{\catnew{c},\atnew{c_1}} \vdash  \rho\act(\newswap{a}{c_1}\act t) \aeq{C} \rho\act(\newswap{b}{c_1}\act s)$, which is equivalent to $\Upsilon_{\catnew{c},\atnew{c_1}}  \vdash \newswap{\rho(a)}{c_1}\act (\rho\act t) \aeq{C} \newswap{\rho(b)}{c_1}\act (\rho\act s)$. Applying the $(\frule{\faeq{C}}{ab})$ rule, we obtain $\Upsilon_{\catnew{c}}  \vdash [\rho(a)]\rho\act t \aeq{C} [\rho(b)]\rho\act s$, which is equivalent to $\Upsilon_{\catnew{c}}  \vdash \rho\act[a]t \aeq{C} \rho\act[b]s$.

            \item  We will prove only the interesting cases.
    \begin{itemize}
        \item {\it Reflexivity.} The variable case is trivial because $\pi^{-1}\circ\pi = \id$ and $\id \in \PN{}{\Upsilon_{\catnew{c}}|_X}$ always holds.

        \item {\it Symmetry.} Suppose $\Upsilon_{\catnew{c}} \vdash \pi_1\act X \aeq{C} \pi_2\act X$. By Inversion (Theorem~\ref{thm:miscellaneous}(\ref{thm:inversion})), we obtain $\pi_2^{-1}\circ\pi_1 \in \PN{}{\Upsilon_{\catnew{c}}|_X}$. Hence $\pi_1^{-1}\circ \pi_2 = (\pi_2^{-1}\circ\pi_1)^{-1} \in  \PN{}{\Upsilon_{\catnew{c}}|_X}$. Therefore, the result follows by rule $(\frule{\faeq{C}}{var})$.

        \item {\it Transitivity.}

        Given nominal terms $s,t,u$ and derivations $\Upsilon_{\catnew{c}} \vdash s \aeq{C} u$ and $\Upsilon_{\atnew{\pvec{c}'}}\vdash u \aeq{C} t$, we will show that $\Upsilon_{\catnew{c}\cup\atnew{\pvec{c}'}} \vdash s \aeq{C} t$. We establish the result by induction on $s$, whose structure influences the other terms. We present only the suspension case here, as it is the most interesting. For $s \equiv \pi_1\act X$, we have $\Upsilon_{\catnew{c}} \vdash \pi_1\act X \aeq{C} u$. This forces $u \equiv \pi_2\act X$, which, in turn, forces $t\equiv \pi_3\act X$. Consequently, $\Upsilon_{\catnew{c}} \vdash \pi_1\act X \aeq{C} \pi_2\act X$ and $\Upsilon_{\catnew{c}} \vdash \pi_2\act X \aeq{C} \pi_3\act X$. By Inversion (Theorem~\ref{thm:miscellaneous}(\ref{thm:inversion})), $\pi_2^{-1}\circ\pi_1 \in \PN{}{\Upsilon_{\catnew{c}}|_X}$ and $\pi_3^{-1}\circ\pi_2 \in \PN{}{\Upsilon_{\atnew{\pvec{c}'}}|_X}$. Since both $\PN{}{\Upsilon_{\catnew{c}}|_X}$ and $\PN{}{\Upsilon_{\atnew{\pvec{c}'}}|_X}$ are subsets of $\PN{}{\Upsilon_{\catnew{c}\cup\atnew{\pvec{c}'}}|_X}$ we have that $\pi_2^{-1}\circ\pi_1$ and
            $\pi_3^{-1}\circ\pi_2$ are in $\PN{}{\Upsilon_{\catnew{c}\cup\atnew{\pvec{c}'}}|_X}$.
            Therefore, $\pi_3^{-1}\circ\pi_1 = (\pi_3^{-1}\circ\pi_2)\circ(\pi_2^{-1}\circ\pi_1)$ is also in $\PN{}{\Upsilon_{\catnew{c}\cup\atnew{\pvec{c}'}}|_X}$
            and thus the result follows by rule $(\frule{\faeq{C}}{var})$.
    \end{itemize}
        \end{itemize}

        \item Direct consequence Equivariance and Equivalence.

        \item The proof is by induction on the structure of $s$. Here we show only the interesting cases.
       \begin{itemize}

        \item $s \equiv \pi_1 \act Y$. This forces $t \equiv \pi_2 \act Y$. Then we have $(\Upsilon\uplus\{\pi\fix{C} X\})_{\catnew{c}} \vdash \pi_1 \act Y \aeq{C} \pi_2 \act Y$. By Inversion (Theorem~\ref{thm:miscellaneous}(\ref{thm:inversion})), it follows that $\pi_2^{-1}\circ\pi_1\in \PN{}{(\Upsilon\uplus\{\pi\fix{C} X\})_{\catnew{c}}|_Y}$.

        \begin{itemize}
            \item $Y\not\equiv X$.

            Then $\PN{}{(\Upsilon\uplus\{\pi\fix{C} X\})_{\catnew{c}}|_Y} = \PN{}{\Upsilon_{\catnew{c}}|_Y}$ and so the result follows by rule $(\frule{\faeq{C}}{var})$.

            \item $Y\equiv X$.

            In this case, the condition $(\dom{\pi}\setminus\catnew{c})\cap\atm{s,t} = \emptyset$ is the same as
            \[
                (\dom{\pi}\setminus\catnew{c})\cap(\dom{\pi_1}\cup\dom{\pi_2}) = \emptyset.
            \]
            This implies $\dom{\pi_2^{-1}\circ\pi_1}\cap (\dom{\pi}\setminus\catnew{c}) = \emptyset$. Consequently, we have that $\pi_2^{-1}\circ\pi_1\in\PN{}{\Upsilon_{\catnew{c}}|_X}$ and the result follows by an application of rule $(\frule{\faeq{C}}{var})$.
        \end{itemize}

        \item $s \equiv [a]s'$. We consider two cases for $t$:

            \begin{itemize}
                \item Case 1: $t \equiv [a]t'$.

                In this situation, we have $(\Upsilon\uplus\{\pi\fix{C} X\})_{\catnew{c}} \vdash [a]s' \aeq{C} [a]t'$. By Inversion (Theorem~\ref{thm:miscellaneous}(\ref{thm:inversion})), we obtain $(\Upsilon\uplus\{\pi\fix{C} X\})_{\catnew{c}} \vdash s' \aeq{C} t'$ and since $(\dom{\pi}\setminus\catnew{c})\cap \atm{s,t} = \emptyset$ implies $(\dom{\pi}\setminus\catnew{c})\cap \atm{s',t'} = \emptyset$, the inductive hypothesis gives us $\Upsilon_{\catnew{c}} \vdash s'\aeq{C} t'$, and the result follows by applying rule $(\frule{\faeq{C}}{ab})$.

                \item Case 2: $t \equiv [b]t'$.

                    Here, we have $(\Upsilon\uplus\{\pi\fix{C} X\})_{\catnew{c}} \vdash [a]s' \aeq{C} [b]t'$. By Inversion (Theorem~\ref{thm:miscellaneous}(\ref{thm:inversion})), we obtain $(\Upsilon\uplus\{\pi\fix{C} X\})_{\catnew{c},\atnew{c_1}} \vdash \newswap{a}{c_1} \act s' \aeq{C} \newswap{b}{c_1} \act t'$ for some $\atnew{c_1}\notin \atm{\Upsilon_{\catnew{c}},a,b,s',t',\pi}$.  By the choice of $\atnew{c_1}$, we have that $\dom{\pi}\setminus\catnew{c} = \dom{\pi}\setminus\catnew{c},\atnew{c_1}$ and so $(\dom{\pi}\setminus\catnew{c})\cap\atm{s,t} = \emptyset$ implies $(\dom{\pi}\setminus\catnew{c},\atnew{c_1})\cap\atm{s,t} = \emptyset$. We claim that
                    \[
                        (\dom{\pi}\setminus\catnew{c},\atnew{c_1})\cap\atm{\newswap{a}{c_1}\act s',\newswap{b}{c_1}\act t'}) = \emptyset.
                    \]
                    Indeed, suppose, by contradiction, that there is an atom, say $d$, such that $d\in \dom{\pi}\setminus\catnew{c},\atnew{c_1}$ and $d\in \atm{\newswap{a}{c_1}\act s',\newswap{b}{c_1}\act t'})$. Then $\newswap{a}{c_1}(d)\in \atm{s'}$. Note that $d$ cannot be $a$ or $\atnew{c_1}$. Thus $\newswap{a}{c_1}(d) = d$ and hence $d\in \atm{s'}$, contradicting $(\dom{\pi}\setminus\catnew{c},\atnew{c_1})\cap\atm{s,t} = \emptyset$. 

                    Now, by induction, we have $\Upsilon_{\catnew{c},\atnew{c_1}} \vdash \newswap{a}{c_1} \act s' \aeq{C} \newswap{b}{c_1} \act t'$ and the result follows by applying the rule $(\frule{\faeq{C}}{ab})$. 
            \end{itemize}
        \end{itemize}
\end{enumerate}

\end{proof}


\section{Proofs of Section~\ref{sec:soundness-completeness} - Soundness and Completeness}\label{app:soundness-completeness}

\soundcomplete*

Since the proof of Completeness is much more elaborate, we will separate the proofs in different sections.

\subsection{Soundness}\label{ssec:soundness}

\begin{proof}[of Soundness]
     Let $\nalg{A}$ be a model of $\C$. We must show for any valuation $\varsigma$ that if $\Upsilon_{\catnew{c}} \vdash t\aeq{C} u$ then $\Int{\Upsilon_{\catnew{c}} \vdash t\aeq{C} u}{\nalg{A}}{\varsigma}$ is valid.

     The proof proceeds by induction on the last rule applied in the derivation of the judgment. To illustrate, we show the cases where the last rule applied is...
     \begin{itemize}
         \item $(\frule{\faeq{C}}{var})$. In this case, we have
            \begin{prooftree}
                \AxiomC{$\rho^{-1}\circ \pi\in \PN{}{\Upsilon_{\catnew{c}}|_X}$}
                \RightLabel{$(\frule{\faeq{C}}{var})$}
                \UnaryInfC{$\Upsilon_{\catnew{c}} \vdash \pi\act X \aeq{C} \rho\act X$}
            \end{prooftree}
            Suppose $\Int{\Upsilon_{\catnew{c}}}{\nalg{A}}{\varsigma}$ is valid. Then:
            \(
                \newc{c}{}. \eta\act \Int{Y}{\nalg{A}}{\varsigma} =  \Int{Y}{\nalg{A}}{\varsigma} \text{ for all $\eta\fix{C} Y\in \Upsilon$.}
            \)
            In particular,
            \[
                \newc{c}{}. \eta\act \Int{X}{\nalg{A}}{\varsigma} =  \Int{X}{\nalg{A}}{\varsigma} \text{ for all $\eta\fix{C} X\in \Upsilon|_X$.}
            \]
            By Pitts' equivalence (Lemma~\ref{lemma:pitts-eq-generalized}) it follows, for any $\eta\fix{C} Y\in \Upsilon$, that for all $\catnew{c}\cap \supp{}{\Int{X}{\nalg{A}}{\varsigma}} = \emptyset$,
            (i) $\dom{\eta_{\catnew{c}}}\cap\supp{}{\Int{X}{\nalg{A}}{\varsigma}}=\emptyset$,
                and
                (ii) $\eta_{\neg\catnew{c}}\act \Int{X}{\nalg{A}}{\varsigma} =  \Int{X}{\nalg{A}}{\varsigma}$. From the condition $\rho^{-1}\circ\pi \in \PN{}{\Upsilon_{\catnew{c}}|_X}$ and Lemma~\ref{lemma:generated-group}, it follows that $(\rho^{-1}\circ\pi) \act \Int{X}{\nalg{A}}{\varsigma} =  \Int{X}{\nalg{A}}{\varsigma}$ holds.

            \item $(\frule{\faeq{C}}{ab})$. In this case, we have
           \begin{prooftree}
               \AxiomC{$\Upsilon_{\catnew{c},\atnew{c_1}} \vdash \newswap{a}{c_1}\act t' \aeq{C} \newswap{b}{c_1}\act u'$}
          \RightLabel{$(\frule{\faeq{C}}{ab})$}
          \UnaryInfC{$\Upsilon_{\catnew{c}} \vdash  [a]t' \aeq{C} [b]u'$}
           \end{prooftree}

           Suppose $\Int{\Upsilon_{\catnew{c},\atnew{c_1}}}{\nalg{A}}{\varsigma}$ is valid. We aim to show that $\Int{[a]t'}{\nalg{A}}{\varsigma} = \Int{[b]u'}{\nalg{A}}{\varsigma}$, which is equivalent to:
            \[
                \abs^{\nalg{A}}(a,\Int{t'}{\nalg{A}}{\varsigma}) = \abs^{\nalg{A}}(b,\Int{u'}{\nalg{A}}{\varsigma}).
            \]
            From the validity of $\Int{\Upsilon_{\catnew{c},\atnew{c_1}}}{\nalg{A}}{\varsigma}$, it follows that:
            \[
                \newc{c}{},\atnew{c_1}.\pi\act \Int{X}{\nalg{A}}{\varsigma} = \Int{X}{\nalg{A}}{\varsigma} \text{ for all } \pi\fix{C} X\in\Upsilon.
            \]
            Since $\atnew{c_1}$ is taken such that $\atnew{c_1}\notin\atm{\Upsilon_{\catnew{c}}}$ we have that $\atnew{c_1}\notin \dom{\pi}$ for all $\pi\fix{C} X\in \Upsilon$. Furthermore, we can assume w.l.o.g. that $\atnew{c_1}$ is taken fresh for $\Int{t'}{\nalg{A}}{\varsigma}$, $\Int{u'}{\nalg{A}}{\varsigma}$  and $\Int{X}{\nalg{A}}{\varsigma}$ for all $X\in \var{\Upsilon_{\catnew{c}}}$. Then,
           \[
                \newc{c}{}.\pi\act \Int{X}{\nalg{A}}{\varsigma} = \Int{X}{\nalg{A}}{\varsigma} \text{ holds for all } \pi\fix{C} X\in\Upsilon.
            \]
            Thus, $\Int{\Upsilon_{\catnew{c}}}{\nalg{A}}{\varsigma}$ is valid. By induction, we get
            \[
                \Int{\newswap{a}{c_1}\act t'}{\nalg{A}}{\varsigma} = \Int{\newswap{b}{c_1}\act u'}{\nalg{A}}{\varsigma}.
            \]

            \begin{claim}[1]
                We claim that $a\notin\supp{}{\Int{u'}{\nalg{A}}{\varsigma}}$ and $b\notin\supp{}{\Int{t'}{\nalg{A}}{\varsigma}}$.  We will prove only $a\notin\supp{}{\Int{u'}{\nalg{A}}{\varsigma}}$, because the proof for $b \notin\supp{}{\Int{t'}{\nalg{A}}{\varsigma}}$ follows by a similar argument. Suppose, by contradiction, that $a \in \supp{}{\Int{u'}{\nalg{A}}{\varsigma}}$. Under this assumption, we have $a \in \supp{}{\newswap{b}{c_1} \act \Int{u'}{\nalg{A}}{\varsigma}}$. Consequently, it follows that $a \in \supp{}{\newswap{a}{c_1} \act \Int{t'}{\nalg{A}}{\varsigma}}$. This implication leads to $\atnew{c_1} \in \supp{}{\Int{t'}{\nalg{A}}{\varsigma}}$. However, this result contradicts the assumption that $\atnew{c_1}$ is fresh for $\Int{t'}{\nalg{A}}{\varsigma}$. 
            \end{claim}

            Now, let's use this information to proceed with the proof. Using that $a\notin\supp{}{\Int{u'}{\nalg{A}}{\varsigma}}$, we infer $a\notin\supp{}{\abs^{\nalg{A}}(b, \Int{u'}{\nalg{A}}{\varsigma})}$. Moreover, by definition, we have $\new \atnew{c}. \newswap{b}{c}\act \abs^{\nalg{A}}(b, \Int{u'}{\nalg{A}}{\varsigma}) = \abs^{\nalg{A}}(b, \Int{u'}{\nalg{A}}{\varsigma})$, which by Pitts' equivalence means that $b \notin\supp{}{\abs^{\nalg{A}}(b, \Int{u'}{\nalg{A}}{\varsigma})}$. Therefore, we obtain $(a \ b) \act \abs^{\nalg{A}}(b, \Int{u'}{\nalg{A}}{\varsigma}) = \abs^{\nalg{A}}(b, \Int{u'}{\nalg{A}}{\varsigma})$. Hence, we can rewrite it  as follows:
            \begin{align*}
                \abs^{\nalg{A}}(b,\Int{u'}{\nalg{A}}{\varsigma}) &= (a \ b)\act\abs^{\nalg{A}}(b,\Int{u'}{\nalg{A}}{\varsigma}) \\
                &= \abs^{\nalg{A}}(a,(a \ b)\act \Int{u'}{\nalg{A}}{\varsigma}).
            \end{align*}
            To complete the proof, note that all we need to do is to show that $(a \ b) \act \Int{u'}{\nalg{A}}{\varsigma} = \Int{t'}{\nalg{A}}{\varsigma}$. Then let's prove this final step.

            \begin{claim}[2]
                We claim that $(a \ b)\act \Int{u'}{\nalg{A}}{\varsigma} = \Int{t'}{\nalg{A}}{\varsigma}$. On one hand, by Claim (1), we know that $b\notin\supp{}{\Int{t'}{\nalg{A}}{\varsigma}}$ and since $\atnew{c_1}\notin\supp{}{\Int{t'}{\nalg{A}}{\varsigma}}$, this implies $\newswap{b}{c_1}\act \Int{t'}{\nalg{A}}{\varsigma} = \Int{t}{\nalg{A}}{\varsigma}$. On the other hand, we already know that $\newswap{a}{c_1}\act\Int{t'}{\nalg{A}}{\varsigma} = \newswap{b}{c_1}\act\Int{u'}{\nalg{A}}{\varsigma}$ holds by induction. Thus,
                \begin{align*}
                    & \newswap{a}{c_1}\act\Int{t'}{\nalg{A}}{\varsigma} =  \newswap{b}{c_1}\act\Int{u'}{\nalg{A}}{\varsigma} \\
                    \Longleftrightarrow{} & \newswap{b}{c_1}\act(\newswap{a}{c_1}\act\Int{t'}{\nalg{A}}{\varsigma}) =  \Int{u'}{\nalg{A}}{\varsigma}\\
                    \Longleftrightarrow{}& ((a \ b)\circ\newswap{b}{c_1}\circ\newswap{a}{c_1})\act\Int{t'}{\nalg{A}}{\varsigma} =  (a \ b)\act\Int{u'}{\nalg{A}}{\varsigma}\\
                    \Longleftrightarrow{}& \newswap{b}{c_1}\act\Int{t'}{\nalg{A}}{\varsigma} =  (a \ b)\act\Int{u'}{\nalg{A}}{\varsigma}\\
                    \Longleftrightarrow{}& \Int{t'}{\nalg{A}}{\varsigma} =  (a \ b)\act\Int{u'}{\nalg{A}}{\varsigma}.
                \end{align*}
            \end{claim}

             This proves completes the proof that $\abs^{\nalg{A}}(b,\Int{u'}{\nalg{A}}{\varsigma}) =  \abs^{\nalg{A}}(a,\Int{t'}{\nalg{A}}{\varsigma})$.

                \end{itemize}
\end{proof}

\subsection{Completeness}\label{app:completeness}

The objective of this section is to show that the calculus is complete w.r.t. the nominal set semantics. The construction is long and complex, therefore, we will just present the roadmap of the constructions and results needed to prove Completeness, and only give the complete proof of the most challenging result (Lemma\ref{alemma:valid-judge-pi-generated}). We refer the reader interested in the details to the extended version of this paper~\cite{arxiv/cairessantos2025}.




\subsubsection*{Free Term Models.}\label{app:free-models}
The set of nominal terms $\F(\Sigma, \V)$ cannot form a model of commutative ($\C$) because it is not a nominal set under the usual permutation action. Specifically, the support of a suspension, such as $\rho \act X$, cannot be precisely determined, as $X$ represents an unknown term. The standard approach is to work with the {\em free term models}, consisting of ground terms $\F(\Sigma\cup \D)$, over a signature $\Sigma$ and a set of term-formers $\D$ disjoint from $\Sigma$.  Here we assume that $\Sigma$ contains commutative function symbols, but $\D$ doesn't. The set of {\em free terms} is defined by the following grammar:
\[ g,g' ::=  a \mid [a]g \mid {\tf f} (g_1,\ldots, g_n) \mid {\tf d}(a_1,\ldots, a_k)\]
Here $\tf{f}:n$ ranges over $\Sigma$ and $\tf{d}:k$ range over elements of $\D$.

As expected, to obtain nominal models for a theory $\C$ we will work with the set of free terms quotient by provable equality modulo $\C$. 
%
This quotient set, defined next, effectively incorporates both $\alpha$-equivalence and commutative behaviour, enabling us to construct a model that satisfies the desired properties in terms of nominal set semantics.

\begin{definition}[Free terms up to $\C$]
Write $\lin{g}$ for the set of free terms $g'$ such that $\emptyset\vdash g \aeq{C} g'$.
The {\em set of free terms $\F(\Sigma\cup \D)$ up to $\C$}, denoted as  $\lin{\F}(\Sigma\cup \D)$, is the set of equivalence classes (modulo $\C$) of free terms, i.e.,  $\lin{\F}(\Sigma\cup \D)=\{\lin{g} \mid g \text{ is a free term}\}$.
\end{definition}

In the following, we will abbreviate $\lin{\F}(\Sigma\cup \D)$ as  $\cF$. The next lemma presents a collection of diverse properties of free terms modulo $\C$. Notably, item \ref{alemma:ground-algebra-nominal-set} establishes that $\lin{\F}_{\C}$ is a nominal set, while item \ref{athm:supp-free-names} defines the support of the equivalence class of a free term modulo $\C$.

We will use the following elementary result from nominal set theory. 
\begin{lemma}[Pitts~\cite{book/Pitts}]\label{alemma:quotient-nominal}
    If $\nom{X}$ is a nominal set and $\sim$ is an equivariant equivalence relation on $\nom{X}$ and $x \in \nom{X}/_{\sim}$, then $\supp{}{x} = \bigcap\{\supp{}{g}\mid g\in x\}$.
\end{lemma}

\begin{lemma}[Properties of $\cF$]\label{alemma:ground-algebra-properties} Let $g,g'\in \F(\Sigma\cup \D)$. The following hold: \hfill

   \begin{enumerate}
    \item \label{alemma:ground-algebra-nominal-set} The set $\cF$ with the action $\pi\act\lin{g} = \lin{\pi\act g}$ is a nominal set and $\supp{}{\lin{g}} = \bigcap\{\supp{}{g'}\mid g'\in \lin{g}\}$.
    \item  \label{alemma:free-names-equivariant} The map $\tf{fn}(-):\F(\Sigma\cup \D) \to \pow{\tf{fin}}{\A}$ is equivariant.
    \item \label{alemma:free-names-preserve} If $g' \in \lin{g}$  then $\tf{fn}(g) = \tf{fn}(g')$.
    \item  \label{athm:supp-free-names} $\supp{}{\lin{g}} = \tf{fn}(g)$.
    \end{enumerate}
\end{lemma}

\begin{proof}
\begin{enumerate}
    \item Direct consequence of Lemma~\ref{alemma:quotient-nominal}.

    \item For all $g\in \lin{\F}(\Sigma\cup\D)$ and all permutation $\pi$, we must prove that $\tf{fn}(\pi\act g) = \pi\act\tf{fn}(g)$. This is done by induction on the ground term $g$. As an illustrative case, let $g\equiv [a]g'$. By induction $\tf{fn}(\pi\act g') = \pi\act \tf{fn}(g')$. Thus,
    \begin{align*}
         \tf{fn}(\pi\act[a]g') = \tf{fn}([\pi(a)]\pi\act g') = \tf{fn}(\pi\act g') \setminus \{\pi(a)\}
         &= \pi\act \tf{fn}(g') \setminus \pi\act \{a\} \\
         &= \pi\act(\tf{fn}(g')\setminus\{a\})\\
         &= \pi\act \tf{fn}([a]g').
    \end{align*}

    \item Suppose $g \sim g'$. Then by definition there is a derivation $~\vdash g \aeq{C} g'$. The proof follows by induction on the last rule applied to obtain  $~\vdash g \aeq{C} g'$. The only non-trivial case is when the last rule applied is $(\frule{\faeq{C}}{ab})$. In this case, $g \equiv [a]g_1$ and $g' \equiv [b]g_1'$ and $~\vdash [a]g_1 \aeq{C} [b]g_1'$. Then by Inversion (Theorem~\ref{thm:miscellaneous}(\ref{thm:inversion})), we have that $~\vdash \newswap{a}{c_1}\act g_1 \aeq{C} \newswap{b}{c_1}\act g_1'$ where $\atnew{c_1}\notin \atm{\catnew{c},a,b,g_1,g_1'}$. Thus $\newswap{a}{c_1}\act g_1 \sim \newswap{b}{c_1}\act g_1'$. By induction, we have $\tf{fn}(\newswap{a}{c_1}\act g_1) = \tf{fn}(\newswap{b}{c_1}\act g_1')$.
    
    \begin{claim}[1]
         We claim that $a\notin \tf{fn}(g_1')$. Suppose, by contradiction, that $a\in \tf{fn}(g_1')$. Then
        \begin{equation*}
            \begin{tabular}{@{}l@{ }ll@{}}
               & $a \in\tf{fn}(g_1')$ &  \\
              $\Longrightarrow$ & $a \in\tf{fn}(\newswap{b}{c_1} \act g_1')$ & (item~\ref{alemma:free-names-equivariant})\\
              $\Longrightarrow$ & $a \in\tf{fn}(\newswap{a}{c_1} \act g_1)$ & \\
           $\Longrightarrow$ & $\atnew{c_1} \in\tf{fn}(g_1)$ & (item~\ref{alemma:free-names-equivariant})
            \end{tabular}
        \end{equation*}
        This leads to a contradiction, because we took $\atnew{c_1}$ such that $\atnew{c_1}\notin \atm{g_1}$ and $\tf{fn}(g_1) \subseteq \atm{g_1}$. This proves the claim.
    \end{claim}

    \begin{claim}[2]
         We also claim that $\tf{fn}(g_1) = (a \ b)\act \tf{fn}(g_1')$. Indeed, from $a,\atnew{c_1} \notin \tf{fn}(g_1')$ we conclude $\newswap{a}{c_1}\act \tf{fn}(g_1') = \tf{fn}(g_1')$. Thus,
        \begin{equation*}
            \begin{tabular}{@{}l@{ }l@{}l@{}}
               & $\tf{fn}(\newswap{a}{c_1}\act g_1) = \tf{fn}(\newswap{b}{c_1}\act g_1')$ \\
               $\Longrightarrow$ & $\newswap{a}{c_1}\act \tf{fn}(g_1) = \newswap{b}{c_1}\act\tf{fn}(g_1')$ & (item~\ref{alemma:free-names-equivariant})\\
               $\Longrightarrow$ &  $\tf{fn}(g_1) = ((\atnew{c_1} \ a)\circ\newswap{b}{c_1})\act\tf{fn}(g_1')$ & \\
               $\Longrightarrow$ &  $\tf{fn}(g_1) = ((b \ a)\circ(\atnew{c_1} \ a))\act\tf{fn}(g_1')$ & \\
               $\Longrightarrow$ &   $\tf{fn}(g_1) = (a \ b)\act(\newswap{a}{c_1}\act\tf{fn}(g_1'))$ & \\
               $\Longrightarrow$ &  $\tf{fn}(g_1) = (a \ b)\act \tf{fn}(g_1')$ &
            \end{tabular}
        \end{equation*}
        This finish the prove of the second claim.
    \end{claim}
    
   Now, let's prove $\tf{fn}(g_1) \setminus\{a\} = \tf{fn}(g_1')\setminus\{b\}$. To show this, first observe that from $a\notin \tf{fn}(g_1')$, we get $a\notin \tf{fn}(g_1')\setminus\{b\}$ and, since $b\notin \tf{fn}(g_1')\setminus\{b\}$, this yields $(a \ b)\act(\tf{fn}(g_1')\setminus\{b\}) = \tf{fn}(g_1')\setminus\{b\}$. Therefore,
        \begin{equation*}
            \begin{tabular}{@{}l@{ }lll@{}}
                $\tf{fn}(g_1)\setminus\{a\}$ & $=$ &  $((a \ b)\act\tf{fn}(g_1'))\setminus((a \ b)\act\{b\})$ & by Claim (2) \\
                 & $=$ & $(a \ b)\act(\tf{fn}(g_1')\setminus\{b\})$ \\
                 & $=$ & $\tf{fn}(g_1')\setminus\{b\}$ 
            \end{tabular}
        \end{equation*}
        This completes the proof that $\tf{fn}(g) = \tf{fn}(g')$.

    \item Take an arbitrary but fixed $g\in\F(\Sigma\cup\D)$.

    \paragraph*{Step 1: Show that $\tf{fn}(g)$ supports $[g]$.}  We will do this by showing, using induction on the structure of $g$, that for any permutation $\pi$, the following holds:
    \[
        \pi\in\Fix{\tf{fn}(g)} \implies \pi\act \lin{g} = \lin{g}.
    \]
    The only interesting case is when $g \equiv [a]g_1$. This implies $\tf{fn}([a]g_1) = \tf{fn}(g_1)\setminus\{a\}$. Then $\pi\in\Fix{\tf{fn}([a]g_1)}$ implies $\pi(b) = b$ for all $b \in  \tf{fn}(g_1)\setminus\{a\}$.

            \begin{itemize}
                \item If $a\notin \tf{fn}(g_1)$, then $ \tf{fn}(g_1) \setminus \{a\} =  \tf{fn}(g_1)$ and then $\pi\in\Fix{\tf{fn}(g_1)}$ which, by induction, implies that $\pi\act \lin{g_1} = \lin{g_1}$. This means that $~\vdash \pi\act g_1 \aeq{C} g_1$, that is, $\pi\act g_1 \sim g_1$.

                \begin{enumerate}
                    \item $\pi(a) = a$: In this case,
                    \begin{prooftree}
                        \AxiomC{$\vdash \pi\act g_1 \aeq{C} g_1$}
                        \UnaryInfC{$\vdash [\pi(a)]\pi\act g_1\aeq{C} [a]g_1$}
                    \end{prooftree}
                    which implies that $\pi\act \lin{[a]g_1} = \lin{[a]g_1}$.

                    \item $\pi(a) \neq a$: In this case, the condition $a\notin\tf{fn}(g_1)$ implies $\pi(a)\notin \tf{fn}(g_1)$ by item~\ref{alemma:free-names-preserve}. Take $\atnew{c_1}\notin \atm{g_1}\cup\{a,\pi(a)\}\cup\dom{\pi}$. Then, $\newswap{a}{c_1}$ and $\newswap{\pi(a)}{c_1}$ are in $\Fix{\tf{fn}(g_1)}$ which, by induction, implies $\newswap{a}{c_1}\act \lin{g_1} = \lin{g_1}$ and $\newswap{\pi(a)}{c_1}\act \lin{g_1} = \lin{g_1}$, that is, $\newswap{\pi(a)}{c_1}\act g_1 \sim g_1$ and $g_1 \sim \newswap{a}{c_1}\act g_1$. Hence
                    \begin{equation*}
                        \begin{tabular}{@{}r@{ }lll@{}}
                             & $g_1 \sim \newswap{a}{c_1}\act g_1$ &  & \\
                             $\Longrightarrow$ & $\pi\act g_1 \sim \pi\act (\newswap{a}{c_1}\act g_1)$ &\\
                            $\Longrightarrow$ & $g_1 \sim \newswap{\pi(a)}{c_1}\act (\pi\act g_1)$ & \\
                            $\Longrightarrow$ & $\newswap{a}{c_1}\act g_1 \sim \newswap{\pi(a)}{c_1}\act (\pi\act g_1)$ & \\
                            $\Longrightarrow$ & $\newswap{\pi(a)}{c_1}\act (\pi\act g_1) \sim \newswap{a}{c_1}\act g_1$ &
                        \end{tabular}
                    \end{equation*}
                    This implies $~\vdash \newswap{\pi(a)}{c_1}\act (\pi\act g_1)  \aeq{C} \newswap{a}{c_1}\act g_1$. By an application of rule $(\frule{\faeq{C}}{ab})$, we obtain $\vdash [\pi(a)]\pi\act g_1 \aeq{C} [a]g_1$ and so $\pi\act \lin{[a]g_1} = \lin{[a]g_1}$.
                \end{enumerate}

                \item If $a\in \tf{fn}(g_1)$, then
                \begin{enumerate}
                    \item $\pi(a) = a$: In this case, $\pi\in\Fix{\tf{fn}(g_1)}$. By induction, $\pi\act\lin{g_1} = \lin{g_1}$, i.e., $~\vdash \pi\act g_1 \aeq{C} g_1.$ The result follows by an application of rule $(\frule{\faeq{C}}{[a]})$.

                    \item $\pi(a) \neq a$: In this case, we claim $a\notin\tf{fn}(\pi\act g_1)$.  Suppose, by contradiction, that $a\in\tf{fn}(\pi\act g_1)$. Then $a\in \tf{fn}(\pi\act g_1)\setminus\{\pi(a)\}$ because $\pi(a)\neq a$. This is equivalent to $a\in \pi\act(\tf{fn}(g_1)\setminus\{a\})$, which in turn is equivalent to $\pi^{-1}(a) \in \tf{fn}(g_1)\setminus\{a\}$. From the hypothesis that $\pi\in\Fix{\tf{fn}([a]g_1)} = \Fix{\tf{fn}(g_1)\setminus\{a\}}$, we get that $\pi(\pi^{-1}(a)) = \pi^{-1}(a)$, i.e., $\pi^{-1}(a) = a$. Applying $\pi$ on both sides yields $\pi(a) = a$, a contradiction.

                    Moreover, we also claim $\pi\act g_1 \sim (\pi(a) \ a)\act g_1$.  By the induction hypothesis, it is sufficient to prove that $(a \ \pi(a))\circ\pi\in\Fix{\tf{fn}(g_1)}$. Since $a\in \tf{fn}(g_1)$, let's see how $(a \ \pi(a))\circ\pi$ interacts with $a$ before analysing other atoms in $\tf{fn}(g_1)$. Then
                            \[
                                ((a \ \pi(a))\circ\pi)(a) = (a \ \pi(a))(\pi(a)) = a
                            \]
                            Now, take $b\in\tf{fn}(g_1)$. Then
                            \begin{align*}
                                ((a \ \pi(a))\circ\pi)(b) &= (a \ \pi(a))(\pi(b)) \\
                                &= \pi(b)\\
                                &= b
                            \end{align*}
                        because $\pi\in\Fix{\tf{fn}(g_1)\setminus\{a\}}$. Now, take $c_1\notin \atm{g_1}\cup\{a,\pi(a)\}\cup\dom{\pi}$. Then, $\newswap{a}{c_1}\in \Fix{\tf{fn}(\pi\act g_1)}$ because $a\notin \tf{fn}(\pi\act g_1)$ and $\atnew{c_1}\notin \tf{fn}(\pi\act g_1)$, the latter being a consequence of the fact the $\tf{fn}(\pi\act g_1)\subseteq \atm{\pi\act g_1}$. By induction, we have $\newswap{a}{c_1}\act\lin{\pi\act g_1} = \lin{\pi\act g_1}$, that is, $\pi\act g_1 \sim \newswap{a}{c_1}\act (\pi\act g_1)$. Hence,
                   \begin{equation*}
                        \begin{tabular}{@{}l@{ }l@{}l@{}}
                            & $\pi\act g_1 \sim (\pi(a) \ a)\act g_1$ &  \\
                            $\Longrightarrow$ & $\pi\act g_1 \sim \newswap{\pi(a)}{c_1}^{(\atnew{c_1} \ a)}\act g_1$ & \\
                            $\Longrightarrow$ & $\newswap{a}{c_1}\act (\pi\act g_1) \sim \newswap{\pi(a)}{c_1}\act (\newswap{a}{c_1}\act g_1)$ & \\
                            $\Longrightarrow$ & $\pi\act g_1 \sim \newswap{\pi(a)}{c_1}\act (\newswap{a}{c_1}\act g_1)$ & \\
                             $\Longrightarrow$ & $(\atnew{c_1} \ \pi(a))\act (\pi\act g_1) \sim \newswap{a}{c_1}\act g_1$ & \\
                            $\Longrightarrow$ & $\newswap{\pi(a)}{c_1}\act (\pi\act g_1) \sim \newswap{a}{c_1}\act g_1$ &
                        \end{tabular}
                    \end{equation*}
                    This implies $\vdash \newswap{\pi(a)}{c_1}\act (\pi\act g_1)  \aeq{C} \newswap{a}{c_1}\act g_1$. By an application of rule $(\frule{\faeq{C}}{ab})$, we obtain $\vdash [\pi(a)]\pi\act g_1 \aeq{C} [a]g_1$ and so $\pi\act \lin{[a]g_1} = \lin{[a]g_1}$.
                \end{enumerate}
            \end{itemize}

        \paragraph*{Step 2: Show that $\supp{}{g} = \tf{fn}(g)$.} Since $\tf{fn}(g)$ supports $[g]$. By the minimality of the support, we obtain $\supp{}{g} \subseteq \tf{fn}(g)$. It remains to prove the other inclusion: $\tf{fn}(g) \subseteq \supp{}{\lin{g}}$. Suppose, by contraction, that there exists $a\in \tf{fn}(g)$ such that $a\notin \supp{}{\lin{g}}$. Choose $b\notin \tf{fn}(g)$. Then $b\notin \supp{}{\lin{g}}$ because we proved that $\supp{}{\lin{g}} \subseteq \tf{fn}(g)$. Since $a,b \notin\supp{}{\lin{g}}$, we have $(a \ b)\act \lin{g} = \lin{g}$. Thus, $(a \ b)\act g \sim g$. By item~\ref{alemma:free-names-preserve}, it follows that $\tf{fn}((a \ b)\act g) = \tf{fn}(g)$. Since $b\notin \tf{fn}(g)$, it follows by item~\ref{alemma:free-names-equivariant} that $a\notin \tf{fn}((a \ b)\act g) = \tf{fn}(g)$, a contradiction. Therefore, this proves the claim that $\supp{}{[g]} = \tf{fn}(g)$.

        \paragraph*{Step 3: Completing the proof.}

        As a direct consequence of item~\ref{alemma:free-names-preserve}, we conclude that $\supp{}{\lin{g}}$ does not  depend on the representative of the class $\lin{g}$, and this completes the proof.
\end{enumerate}
\end{proof}

Now, we construct a nominal commutative $\Sigma$-algebra  $\nalg{F}$ as follows:

\begin{itemize}
    \item The nominal set $\nom{F}=(|\nom{F}|, \cdot )$ with domain $|\nom{F}|=\cF$.
     \item Define the injective equivariant map $\atom^{\nalg{F}}(a) = \lin{a}$, for all $a\in \A$;
        \item Define the equivariant map $\abs^{\nalg{F}}(a,\lin{g}) = \lin{[a]g}$ for all $a\in\A$ and all $g\in \F(\Sigma\cup \D)$;
        \item Define the equivariant map
        $f^{\nalg{F}}(\lin{g_1},\ldots,\lin{g_n}) = \lin{\tf{f}(g_1,\ldots,g_n)}$,  for $g_1,\ldots,g_n \in \F(\Sigma\cup \D)$ and each $\tf{f}:n$ in $\Sigma$.
\end{itemize}

\begin{theorem}\label{alemma:ground-algebra-model}
    $\nalg{F}$ is a nominal model of $\C$.
\end{theorem}

\begin{proof} The proof is standard and can be found in~\cite{arxiv/cairessantos2025}. 
     \begin{itemize}
         \item $\lin{\F}_\C$ is a nominal set by Lemma~\ref{alemma:ground-algebra-properties}(\ref{alemma:ground-algebra-nominal-set}).

         \item The equivariance of each map follows by the fact that $\pi\act \lin{g} = \lin{\pi\act g}$.

         \item It's easy to see that $\atom^{\nalg{F}}$ is injective because $\lin{a} = \{a\}$ for all $a\in\A$.

         \item The following condition is satisfied: $\new\atnew{c}. \newswap{a}{c} \act\abs^{\nalg{F}}(a,\lin{g}) = \abs^{\nalg{F}}(a,\lin{g})$ for all $a\in \A$ and all $\lin{g}\in\lin{\F}(\Sigma\cup\D)$. Indeed, suppose, by contradiction, that $\new \atnew{c}.\newswap{a}{c} \act \abs^{\nalg{F}}(a,\lin{g}) = \abs^{\nalg{F}}(a,\lin{g})$ does not  hold, that is, the set $D_1 = \{\atnew{c} \mid \newswap{a}{c} \act \abs^{\nalg{F}}(a,\lin{g}) =\abs^{\nalg{F}}(a,\lin{g})\}$ is finite. By Theorem~\ref{athm:supp-free-names}, we have $a\notin \supp{}{\abs^{\nalg{F}}(a,\lin{g})}$. Define $D_2 = \{a\}\cup\supp{}{\abs^{\nalg{F}}(a,\lin{g})}$.  Then, for any atom $\atnew{c_1}\notin D_2$, it follows that $\newswap{a}{c_1}\act \abs^{\nalg{F}}(a,\lin{g}) = \abs^{\nalg{F}}(a,\lin{g})$. Therefore, $\A\setminus D_2 \subseteq D_1$, which is a contradiction with the fact that $D_1$ is finite.

         \item  Given $g_1,g_2\in \F(\Sigma\cup\D)$. For each $\tf{f^C}$ in $\Sigma$, note that we have
         \begin{prooftree}
             \AxiomC{ }
             \RightLabel{(refl)}
             \UnaryInfC{$~\vdash g_1 \aeq{C} g_1$}
             \AxiomC{ }
             \RightLabel{(refl)}
             \UnaryInfC{$~\vdash g_2 \aeq{C} g_2$}
             \RightLabel{$(\frule{\faeq{C}}{\tf{f^C}})$}
             \BinaryInfC{$~\vdash \tf{f^C}(g_1,g_2) \aeq{C} \tf{f^C}(g_2,g_1)$}
         \end{prooftree}
         Thus, $\tf{f^C}(g_1,g_2) \sim \tf{f^C}(g_2,g_1)$ and hence $\lin{\tf{f^C}(g_1,g_2)} = \lin{\tf{f^C}(g_2,g_1)}$. Consequently, $f^{\C,\nalg{F}}(\lin{g_1},\lin{g_2}) = f^{\C,\nalg{F}}(\lin{g_2},\lin{g_1})$, proving the result.
    \end{itemize}
\end{proof}

$\nalg{F}$ will be called the {\em free-term model of $\C$}.

\subsubsection*{Useful auxiliary results.}
Below, we present some results from nominal set theory that will be useful:

\begin{itemize}
    \item If $\nom{X}$ is a nominal set, then the powerset $\pow{}{\nom{X}}$ with the action $\pi\act S = \{\pi\act x \mid x \in S\}$, is not necessarily a nominal set. However, the restriction of $\pow{}{\nom{X}}$,
    \[
         \pow{\tf{fs}}{\nom{X}} = \{S\subseteq \nom{X} \mid \text{$S$ is finitely supported}\},
    \]
    equipped with the same action, is a nominal set.

    \item Suppose $S$ is a set, all of whose elements have finite support. If $\bigcup \{\supp{}{x} \mid x\in S\}$ is finite then $\supp{}{S}$ exists and
    \(
        \supp{}{S} = \bigcup \{\supp{}{x} \mid x\in S\}.
    \)
    The proof can be found in~\cite{DBLP:journals/bsl/Gabbay11}.
\end{itemize}

\begin{lemma}\label{lemma:power-set-algebra}
    For all natural $n\geq 1$, the nominal set $\pow{\tf{fs}}{\A^n}$ is a model of $\C$ when equipped with the following structure:
\begin{enumerate}
    \item $\atom(a) = \{(a,\ldots,a)\}$ for all $a\in\A$

    \item $ \abs(a,S) = S\setminus\{x\in S\mid a\in\supp{}{x}\}$ for all $a\in\A$ and all $S\in \pow{\tf{fs}}{\A^n}$

    \item Every $\tf{f}:n$ is mapped to $f(S_1,\ldots,S_n) = S_1\cup\ldots\cup S_n$ for all $S_i\in  \pow{\tf{fs}}{\A^n}$.
\end{enumerate}
\end{lemma}

\begin{proof}
    It's not difficult to see that:
   \begin{itemize}
       \item all mappings are equivariant.
       \item $\atom$ is injective.
       \item $\abs$ satisfies $\new\atnew{c}.\newswap{a}{c}\act\abs(a,S) = \abs(a,S)$ for all $a\in\A$ and $S\in  \pow{\tf{fs}}{\A^n}$. This is a consequence of Pitts' equivalence and $a\notin \supp{}{\abs(a,S)}$ by construction.
       \item Each $\tf{f^C}:n$ in the signature is mapped to $f(S_1,S_2) = S_1\cup S_2$ and $f(S_1,S_2) = f(S_2,S_1)$  for all $S_1,S_2\in  \pow{\tf{fs}}{\A^n}$. Thus, it is indeed a model of $\C$.
   \end{itemize}
\end{proof}

\begin{lemma} \label{alemma:valid-fix-point-judgment}
The following hold:
\begin{enumerate}
     \item \label{alemma:domain-preservation} If $\Upsilon_{\catnew{c}} \vDash \pi\act X \aeq{C} X$, then $\dom{\pi}\subseteq \atm{\Upsilon_{\catnew{c}}|_X}$.
     \item \label{alemma:valid-pi-c-in-generated} If $\Upsilon_{\catnew{c}} \vDash \pi_{\catnew{c}}\act X \aeq{C} X$, then $\pi_{\catnew{c}}\in \Perm{\atm{(\Upsilon_{\catnew{c}}|_X)_{\fresh}}}$.
\end{enumerate}
\end{lemma}

\begin{proof}
     \begin{enumerate}
         \item If $\pi = \id$, then the result follows trivially. Assume that $\pi\neq \id$. Suppose, by contradiction, that there is an atom $a\in \dom{\pi}$ such that $a\notin\atm{\Upsilon_{\catnew{c}}|_X}$. Consider the algebra $\nalg{A} = \pow{\tf{fin}}{\A}$ and the valuation $\varsigma$ given by $\varsigma(X) = \{a\}$ and $\varsigma(Y) = \{b\}$ for all $Y\not\equiv X$, where $b$ is fresh. Then, it is not difficult to see that $\Int{\Upsilon_{\catnew{c}}}{\nalg{A}}{\varsigma}$ is valid. By hypothesis, this means that $\pi\act \Int{X}{\nalg{A}}{\varsigma} = \Int{X}{\nalg{A}}{\varsigma}$, i.e., $\pi(a) = a$, which in turn implies that $a\notin \dom{\pi}$, a contradiction.

        \item Let's start proving the following claims

        \begin{claim}[1]
            We claim that $\dom{\pi_{\catnew{c}}}\cap \supp{}{\Int{X}{\nalg{A}}{\varsigma}} = \emptyset$ for all $\catnew{c}\cap \supp{}{\Int{X}{\nalg{A}}{\varsigma}} = \emptyset$. In fact, $\Upsilon_{\catnew{c}} \vDash \pi_{\catnew{c}}\act X \aeq{C} X$ means that for all models $\nalg{A}$ and all valuations $\varsigma$,
        \begin{equation}\label{eq:pi-c-in-generated}
            \text{If $\Int{\Upsilon_{\catnew{c}}}{\nalg{A}}{\varsigma}$ is valid, then $\pi_{\catnew{c}}\act \Int{X}{\nalg{A}}{\varsigma} = \Int{X}{\nalg{A}}{\varsigma}$.}
        \end{equation}
        $\Int{\Upsilon_{\catnew{c}}}{\nalg{A}}{\varsigma}$ being valid, by definition, means that $\newc{c}{}. \pi'\act \Int{Y}{\nalg{A}}{\varsigma} = \Int{Y}{\nalg{A}}{\varsigma}$ holds for all $\pi'\fix{C} Y\in (\Upsilon_{\catnew{c}})_{\fresh}$. By generalized Pitts's equivalence (Lemma~\ref{lemma:pitts-eq-generalized}) we have, for all $\pi'\fix{C} Y\in (\Upsilon_{\catnew{c}})_{\fresh}$, that
        \begin{itemize}
            \item $\dom{\pnew{\pi'}{\catnew{c}}}\cap \supp{}{\Int{Y}{\nalg{A}}{\varsigma}} = \emptyset$ for all $\catnew{c}$ such that $\catnew{c}\cap(\bigcup\supp{}{\Int{Y}{\nalg{A}}{\varsigma}}) = \emptyset$. (note that $\pi' = \pnew{\pi'}{\catnew{c}}$ for all $\pi'\fix{C} Y\in (\Upsilon_{\catnew{c}})_{\fresh}$.)
        \end{itemize}
        By (\ref{eq:pi-c-in-generated}), we have $\pi_{\catnew{c}}\act \Int{X}{\nalg{A}}{\varsigma} = \Int{X}{\nalg{A}}{\varsigma}$ holds for all $\catnew{c}\cap\supp{}{\Int{X}{\nalg{A}}{\varsigma}} = \emptyset$. As a consequence of Lemma~\ref{lemma:generated-group}, we have $\dom{\pi_{\catnew{c}}}\cap \supp{}{\Int{X}{\nalg{A}}{\varsigma}} = \emptyset$.
        \end{claim}

        \begin{claim}[2]
            We claim that  $\dom{\pi_{\catnew{c}}}\subseteq \atm{(\Upsilon_{\catnew{c}}|_X)_{\fresh}}$.  Suppose, by contradiction, that there is an atom $a\in \dom{\pi_{\catnew{c}}}$ such that $a\notin \atm{(\Upsilon_{\catnew{c}}|_X)_{\fresh}}$. Then $a\in \atm{(\Upsilon_{\catnew{c}}|_X)_{\fix{C}}}\setminus{\catnew{c}}$ by item~\ref{alemma:domain-preservation}. Consider the valuation $\varsigma^*$ given by $\varsigma^*(X) = \atm{(\Upsilon_{\catnew{c}}|_X)_{\fix{C}}}\setminus{\catnew{c}}$ and $\varsigma^*(Y) = \{b\}$ for all $Y\not\equiv X$, where $b$ is fresh. Thus, $\supp{}{\Int{X}{\nalg{A}}{\varsigma}} = \atm{(\Upsilon_{\catnew{c}}|_X)_{\fix{C}}}\setminus{\catnew{c}}$ and $\supp{}{\Int{Y}{\nalg{A}}{\varsigma}} = \{b\}$ for all $Y\not\equiv X$. Moreover, it's not hard to see that $\Int{\Upsilon_{\catnew{c}}}{\pow{\tf{fin}}{\A}}{\varsigma^*}$ is valid. Therefore, $a \in \dom{\pi_{\catnew{c}}}\cap \supp{}{\Int{X}{\nalg{A}}{\varsigma}}$, which is a contradiction with Claim (1).
        \end{claim}

        As a consequence of Claim (2), we have that $\pi_{\catnew{c}}\in \Perm{\atm{(\Upsilon_{\catnew{c}}|_X)_{\fresh}}}$.
     \end{enumerate}
\end{proof}

The next lemma is the most challenging as it establishes a characterisation of the permutations on semantic judgements of the form $\Upsilon_{\catnew{c}} \vDash \pi\act X\aeq{C} X$, which by definition, are the judgements that are valid in every model and under every valuation: It must be the case that the permutation is generated by the permutations in the context.

\begin{lemma}[Characterisation of valid fixed-points]\label{alemma:valid-judge-pi-generated} \hfill

    If $\Upsilon_{\catnew{c}} \vDash \pi\act X\aeq{C} X$, then $\pi\in \PN{}{\Upsilon_{\catnew{c}}|_X}$.
\end{lemma}

\begin{proof}
If $\pi = \id$, then the result follows trivially. Assume that $\pi\neq \id$. Here are some useful information:

\begin{itemize}
    \item By generalized Pitts' equivalence (Lemma~\ref{lemma:pitts-eq-generalized}), from $\Upsilon_{\catnew{c}} \vDash \pi\act X\aeq{C} X$, we have $ \Upsilon_{\catnew{c}} \vDash \pi_{\catnew{c}}\act X\aeq{C} X$ and $\Upsilon_{\catnew{c}} \vDash \pi_{\neg\catnew{c}}\act X\aeq{C} X$.

    \item By Lemma~\ref{alemma:valid-fix-point-judgment}(\ref{alemma:valid-pi-c-in-generated}), we have $\pi_{\catnew{c}}\in \Perm{\atm{(\Upsilon_{\catnew{c}}|_X)_{\fresh}}}$.

    \item By Lemma~\ref{alemma:valid-fix-point-judgment}(\ref{alemma:domain-preservation}),  $\dom{\pi_{\neg\catnew{c}}} \subseteq\atm{(\Upsilon_{\catnew{c}}|_X)_{\fix{C}}}\setminus{\catnew{c}}$.

    \item To simplify the proof we abbreviate $\PN{}{\Upsilon_{\catnew{c}}|_X} = \Perm{(\Upsilon_{\catnew{c}}|_X)_{\fresh}}\pair{(\Upsilon_{\catnew{c}}|_X)_{\fix{C}}}$.
\end{itemize}

Then, it is sufficient to prove that $\pi_{\neg\catnew{c}} \in \pair{(\Upsilon_{\catnew{c}}|_X)_{\fix{C}}}$. Let $(a_1,\ldots,a_n)$ be a list enumerating all atoms being mentioned in all permutations of $(\Upsilon_{\catnew{c}}|_X)_{\fix{C}}$. Then $(a_1,\ldots,a_n) \in \A^n$ where $\A^n$ is the $n$-ary Cartesian power of the nominal set $\A$. The set $\A^n$ is a nominal set when equipped with the usual pointwise action. The orbit of $(a_1,\ldots,a_n)$ over $\pair{(\Upsilon_{\catnew{c}}|_X)_{\fix{C}}}$, is the set ${\cal O} = \{\pi\act (a_1,\ldots,a_n) \mid \pi\in \pair{(\Upsilon_{\catnew{c}}|_X)_{\fix{C}}}\}$ .
Since every element of ${\cal O}$ has finite support and $\bigcup_{x\in {\cal O}} \supp{}{x}  = \{a_1,\ldots,a_n\}$, it follows that ${\cal O}$ itself has finite support and $\supp{}{{\cal O}} = \bigcup_{x\in {\cal O}} \supp{}{x}$. This proves that ${\cal O}\in \pow{\tf{fs}}{\A^n}$. Take the model $\nalg{A}$ as $\pow{\tf{fs}}{\A^n}$ and the valuation
\[
    \varsigma(Y) = \left\{\begin{array}{lc}
        {\cal O}, & Y\equiv X \\
        \{(d_Y,\ldots,d_Y)\} & Y\not\equiv X
    \end{array}\right.
\]
where $d_Y$ is a fresh atom for $\Upsilon_{\catnew{c}}|_Y$.

\begin{itemize}
    \item For every $\rho \fix{C} Y \in \Upsilon_{\catnew{c}}$ such that $Y\not\equiv X$, it holds that $\dom{\rho} \cap \supp{}{\varsigma(Y)} = \emptyset$. Consequently, by the definition of support, we have $\rho \act \varsigma(Y) = \varsigma(Y)$ and so $\newc{c}{}. \rho \act \varsigma(Y) = \varsigma(Y).$

    For  those $\rho\fix{C} Y\in (\Upsilon_{\catnew{c}})_{\fresh}$, we have that $\rho_{\catnew{c}} = \rho$, and given that $\rho \act \varsigma(Y) = \varsigma(Y)$ is true for all $\catnew{c}$ satisfying $\catnew{c} \cap \bigcup_{Y} \supp{}{\varsigma(Y)} = \emptyset$, it follows, by the generalized Pitts' equivalence (Lemma~\ref{lemma:pitts-eq-generalized}), that $\newc{c}{}.\rho \act \varsigma(Y) = \varsigma(Y)$.

    \item Now, for those $\rho\fix{C} X\in \Upsilon_{\fresh}$, we have that $\dom{\rho}\cap \supp{}{\cal O} = \emptyset$ and hence by the same reason, we have $\newc{c}{}.\rho\act \varsigma(X) = \varsigma(X)$.

    \item By the definition of ${\cal O}$, we have that $\rho \act {\cal O} = {\cal O}$ holds for all $(\Upsilon_{\catnew{c}}|_X)_{\fix{C}}$. In particular, $\rho \act {\cal O} = {\cal O}$ for all permutation $\rho$ such that $\rho\fix{C} X\in (\Upsilon_{\catnew{c}}|_X)_{\fix{C}}$. So $\newc{c}{}.\rho\act \varsigma(X) = \varsigma(X)$.
\end{itemize}

Consequently, $\Int{\Upsilon_{\catnew{c}}}{\nalg{A}}{\varsigma}$ is valid.

From $\Upsilon \vDash \pi_{\neg\catnew{c}}\act X \aeq{C} X$, we have $\pi_{\neg\catnew{c}}\act {\cal O} = {\cal O}$. Then for all element $x\in{\cal O}$ there is an element $y\in{\cal O}$ such that $\pi_{\neg\catnew{c}}\act x = y$. By the definition of ${\cal O}$, $x = \pi_x\act (a_1,\ldots,a_n)$ and $y = \pi_y\act (a_1,\ldots,a_n)$ where $\pi_x,\pi_y\in \pair{(\Upsilon_{\catnew{c}}|_X)_{\fix{C}}}$. Hence $ \pi_{\neg\catnew{c}}\act x = y$ implies $\pi_{\neg\catnew{c}}\act (\pi_x\act (a_1,\ldots,a_n)) = \pi_y\act (a_1,\ldots,a_n)$, leading to
\[
    (\pi_y^{-1}\circ\pi_{\neg\catnew{c}}\circ\pi_x)\act (a_1,\ldots,a_n) = (a_1,\ldots,a_n).
\]
This implies $(\pi_y^{-1}\circ\pi_{\neg\catnew{c}}\circ\pi_x)(a_i) = a_i$ for all $i=1,\ldots,n$ and thus $\dom{\pi_y^{-1}\circ\pi_{\neg\catnew{c}}\circ\pi_x}\cap\{a_1,\ldots,a_n\} = \emptyset.$

As observed at the beginning, we have $\dom{\pi_{\neg\catnew{c}}}\subseteq \atm{(\Upsilon_{\catnew{c}}|_X)_{\fix{C}}}\setminus\catnew{c}$. Then the only possibility is $\pi_y^{-1}\circ\pi_{\neg\catnew{c}}\circ\pi_x = \id\in \pair{(\Upsilon_{\catnew{c}}|_X)_{\fix{C}}}$. Therefore,
\[
    \pi_{\neg\catnew{c}} = \pi_y\circ(\pi_y^{-1}\circ\pi_{\neg\catnew{c}}\circ\pi_x)\circ \pi_x^{-1} \in \pair{(\Upsilon_{\catnew{c}}|_X)_{\fix{C}}},
\]
Since $\pi = \pi_{\catnew{c}}\circ\pi_{\neg\catnew{c}}$, we obtain $\pi\in \Perm{(\Upsilon_{\catnew{c}}|_X)_{\fresh}}\pair{(\Upsilon_{\catnew{c}}|_X)_{\fix{C}}}$.
\end{proof}

\begin{lemma}\label{alemma:var-judge-valid}
    If $\Upsilon_{\catnew{c}} \vDash t \aeq{C} u$ and $t$ is a suspension, i.e., $t\equiv \pi_1\act X$, then $u$ must be a suspension, i.e.,  $u\equiv \pi_2\act X$, for some $\pi_2$.
\end{lemma}

\begin{proof}
        Just observe that for each $u \not\equiv \pi_2\act X$, it is possible to find an algebra $\nalg{A}$ and a valuation $\varsigma$ that validate the context, but $\Int{\pi_1\act X}{\nalg{A}}{\varsigma} \neq \Int{u}{\nalg{A}}{\varsigma}$. Here, we will analyse only the case where $u \equiv \pi_2\act Y$.

        Take $a_1,a_2,a_3\notin \atm{\Upsilon_{\catnew{c}},\pi_1,\pi_2}$. Consider the algebra $\nalg{A} = \pow{\tf{fin}}{\A}$ with the valuation $\varsigma$ defined by $\varsigma(X) = \{a_1,a_2\}, \varsigma(Y) = \{a_1,a_3\}$, and $\varsigma(Z) = \emptyset$ for all $Z\not\equiv  X,Y$. Then $\Int{\Upsilon_{\catnew{c}}}{\nalg{A}}{\varsigma}$ is clearly valid but
        \[
           \pi_1\act \Int{X}{\nalg{A}}{\varsigma} = \{a_1,a_2\} \neq \{a_1,a_3\} =  \pi_2\act \Int{Y}{\nalg{A}}{\varsigma}.
       \]
\end{proof}

\subsubsection{Proof of Completeness.}

The rest of the section will be based on the assumption that $\Upsilon_{\catnew{c}}\vDash t \aeq{C} u$.

\begin{definition}
    Let $\mathcal{A}$ be the atoms in $\atm{\Upsilon_{\catnew{c}},t,u}$. Let $\mathcal{X}$ be the variables mentioned anywhere in $\Upsilon_{\catnew{c}}, t$ or $u$.
    \begin{enumerate}
        \item For each $X\in\mathcal{X}$ pick the following data:
    \begin{itemize}
        \item let ${\cal A}_X$ be the set of atoms $a\in{\cal A}\setminus\catnew{c}$ such that there is no $\pi\fix{C} X\in(\Upsilon_{\catnew{c}}|_X)_{\fresh}$ with $a\in \dom{\pi}$.
        \item let $a_{X_1},\ldots,a_{X_{k_X}}$ be the atoms in ${\cal A}_X$ in some arbitrary but fixed order.
        \item let $\tf{d}_X:k_X$ be a fresh term-former.
    \end{itemize}

    \item For each $X \notin \mathcal{X}$, let $\tf{d}_X : 0$ be a term-former. Let $\sf{D}$ be the set of all $\tf{d}$'s (so for each $X\in\V$ we have a $\tf{d}_X\in \sf{D}$).
    \end{enumerate}
\end{definition}

\begin{definition}
    Let $\sigma$ be the following substitution:
    \[
        X\sigma \equiv \left\{\begin{array}{ll}
            \tf{d}_X(a_{X_1},\ldots,a_{X_{k_X}}) &  (X\in\mathcal{X})\\
            \tf{d}_X() & (X\notin\mathcal{X})
        \end{array}\right.
    \]
\end{definition}

Thus, $\sigma$ maps all variables mentioned in $t$ or $u$ to an appropriate ground term such that the support we know.

Let $\Sigma^+ = \Sigma\cup\sf{D}$ and consider $\F(\Sigma^+)$ the set of free (nominal) terms. For each $X\in{\cal X}$, define $R_X$ as the set of all identities of the following form:
\[
    \tf{d}_X(a_{X_1},\ldots,a_{X_{k_X}}) \approx \pi\act \tf{d}_X(a_{X_1},\ldots,a_{X_{k_X}}),
\]
for all $\pi\fix{C} X\in (\Upsilon_{\catnew{c}}|_X)_{\fix{C}}$. Let $R := \bigcup_{X\in{\cal X}} R_X$. Denote by ${\cal R}$ the equivariant equivalence closure of $R$. Define the relation ${\cal E}$ by
\[
    (g_1,g_2)\in {\cal E} \text{ iff } \vdash g_1\aeq{C} g_2 \text{ or } (g_1,g_2)\in {\cal R}.
\]

\begin{lemma}
    ${\cal E}$ is an equivariant equivalence relation on $\lin{\F}(\Sigma^+)$.
\end{lemma}

\begin{proof}
    \begin{enumerate}
        \item {\em Reflexivity}, {\em Symmetry} and {\em Equivariance} follows because both relations $\aeq{C}$ and ${\cal R}$ are reflexive, symmetric and equivariant.

        \item {\em Transitivity}:  Suppose
        $(g_1,g_2)\in {\cal E}$ and $(g_2,g_3)\in {\cal E}$.
        \begin{enumerate}
            \item If$~\vdash g_1\aeq{C} g_2$, then we have two cases:

            \begin{itemize}
                \item Case$~\vdash g_2\aeq{C} g_3$. In this case,$~\vdash g_1\aeq{C} g_3$ follows by the fact that $\aeq{C}$ is transitive and so $(g_1,g_3)\in {\cal E}$.

                \item Case $(g_2,g_3)\in {\cal R}$. In this case, $g_2$ is of the form
                \[
                    g_2 \equiv \tf{d}_X(a_{X_1}',\ldots,a_{X_{k_X}}')
                \]
                which, by$~\vdash g_1\aeq{C} g_2$ forces $g_1 \equiv g_2$. Thus, $(g_1,g_3)\in {\cal E}$.
            \end{itemize}

            \item  If $(g_1,g_2)\in {\cal R}$, then we have two cases:\begin{itemize}
                \item Case$~\vdash g_2\aeq{C} g_3$. This case is similar to the second bullet of the previous case (a).

                \item Case $(g_2,g_3)\in {\cal R}$. In this case, $(g_1,g_3)\in{\cal R}$ follows by the fact that ${\cal R}$ is transitive and hence $(g_1,g_3)\in {\cal E}$.
            \end{itemize}
        \end{enumerate}
    \end{enumerate}
\end{proof}

As a consequence, the set $\F(\Sigma^+)/{{\cal E}}$ forms a nominal set, which we denote by $\lin{\F}_{\cal E}$. This set satisfies properties analogous to those stated in Lemma~\ref{alemma:ground-algebra-properties}. The proof follows a very similar structure, with the only difference lying in item~\ref{alemma:free-names-preserve}.

\begin{lemma}
    If $g'\in \lin{g}$, then $\tf{fn}(g) = \tf{fn}(g')$.
\end{lemma}

\begin{proof}
    Suppose  $g'\in \lin{g}$. Then $(g,g')\in{\cal E}$, so$~\vdash g\aeq{C} g'$ or $(g,g')\in {\cal R}$. If$~\vdash g\aeq{C} g'$, then the proof is the same as the one in Lemma~\ref{alemma:ground-algebra-properties}(\ref{alemma:free-names-preserve}). If $(g,g')\in {\cal R}$, then we have two possibilities:

    \begin{itemize}
        \item $g \equiv  \tf{d}_X(a_{X_1},\ldots,a_{X_{k_X}})$.

        In this case, $g' \equiv g$ or $g' \equiv \pi \act \tf{d}_X(a_{X_1}, \ldots, a_{X_{k_X}})$ for some permutation $\pi$ such that $\pi\fix{C} X \in (\Upsilon_{\catnew{c}}|_X)_{\fix{C}}$. In both scenarios, $\tf{fn}(g) = \tf{fn}(g')$, as $\pi$ merely rearranges some of the atoms in $a_{X_1}, \ldots, a_{X_{k_X}}$.

        \item  $g \equiv \tf{d}_X(a_{X_1}',\ldots,a_{X_{k_X}}')$,

        In this case,
        \[
            g \equiv \mu\act \tf{d}_X(a_{X_1},\ldots,a_{X_{k_X}})
        \]
        where $\mu = (a_{X_1} \ a_{X_1}')\ldots(a_{X_{k_X}} \ a_{X_{k_X}}')$. Then $g' \equiv g$ or 
        \[
            g'\equiv \mu\act(\pi\act \tf{d}_X(a_{X_1},\ldots,a_{X_{k_X}}))
        \]
        for some permutation $\pi$ such that $\pi\fix{C} X \in (\Upsilon_{\catnew{c}}|_X)_{\fix{C}}$. Then
        \begin{align*}
            \tf{fn}(g) &= \tf{fn}(\mu\act \tf{d}_X(a_{X_1},\ldots,a_{X_{k_X}}))\\
            &= \mu\act\tf{fn}( \tf{d}_X(a_{X_1},\ldots,a_{X_{k_X}})\\
            &= \mu\act\tf{fn}(\pi\act \tf{d}_X(a_{X_1},\ldots,a_{X_{k_X}}))\\
            &= \tf{fn}(\mu\act(\pi\act \tf{d}_X(a_{X_1},\ldots,a_{X_{k_X}})))\\
            &= \tf{fn}(g').
        \end{align*}
    \end{itemize}
\end{proof}

Additionally, the set $\lin{\F}_{\cal E}$, equipped with the same equivariant maps defined in Subsection~\ref{app:free-models}, forms an algebra that is a model of $\C$. We denote this algebra by $\nalg{F}_{\cal E}$.

Define $\varsigma^*$ on $\nalg{F}_{\cal E}$ by $\varsigma^*(X) =\lin{X\sigma}$ for every $X\in\V$.

\begin{lemma}\label{alemma:interpretation}
The following hold:
\begin{enumerate}
\item \label{alemma:substitution-interpretation}
    $\Int{t}{\nalg{F}_{\cal E}}{\varsigma^*} = \lin{t\sigma}$ and $\Int{u}{\nalg{F}_{\cal E}}{\varsigma^*} = \lin{u\sigma}$.
\item \label{alemma:context-valid}
    $\Int{\Upsilon_{\catnew{c}}}{\nalg{F}_{\cal E}}{\varsigma^*}$ is valid.
\end{enumerate}
\end{lemma}

\begin{proof}
    \begin{enumerate}
        \item Direct by induction on the structure of $t$ and $u$. To illustrate, consider the case $t \equiv \pi\act X$.
        \begin{align*}
            \Int{\pi\act X}{\nalg{F}_{\cal E}}{\varsigma^*} = \pi\act \Int{X}{\nalg{F}_{\cal E}}{\varsigma^*} = \pi\act\varsigma(X) = \pi\act \lin{X\sigma} = \lin{\pi\act(X\sigma)} = \lin{(\pi\act X)\sigma}.
        \end{align*}

        \item Let $\pi\fix{C} X\in \Upsilon_{\catnew{c}}$. Then $X\in {\cal X}$.
    \begin{itemize}
        \item Suppose $\pi\fix{C} X\in (\Upsilon_{\catnew{c}})_{\fresh}$. Let $\catnew{c}$ be such that  $\catnew{c}\cap\supp{}{ \Int{X}{\nalg{F}_{\cal E}}{\varsigma^*}} = \emptyset$. By construction $\dom{\pi}\cap{\cal A}_X = \emptyset$   so $\dom{\pi}\cap \supp{}{\varsigma^*(X)} = \emptyset$ which, by the definition of support, gives us $\pi\act \varsigma^*(X) \equiv \varsigma^*(X)$. Consequently. we have $\pi\act \lin{\varsigma^*(X)} = \lin{\varsigma^*(X)}$, that is, $\pi\act \Int{X}{\nalg{F}_{\cal E}}{\varsigma^*} = \Int{X}{\nalg{F}_{\cal E}}{\varsigma^*}$. This proves that the set $\{\catnew{c} \mid  \pi\act \Int{X}{\nalg{F}_{\cal E}}{\varsigma^*} = \Int{X}{\nalg{F}_{\cal E}}{\varsigma^*}\}$ is cofinite. Thus, $\newc{c}{}.\pi\act \Int{X}{\nalg{F}_{\cal E}}{\varsigma^*} = \Int{X}{\nalg{F}_{\cal E}}{\varsigma^*}$.

        \item Suppose $\pi\fix{C} X\in (\Upsilon_{\catnew{c}})_{\fix{C}}$.  By construction,
        \begin{align*}
             \lin{\tf{d}_X(a_{X_1},\ldots,a_{X_{k_X}})} &= \lin{\pi\act\tf{d}_X(a_{X_1},\ldots,a_{X_{k_X}})}\\
             &= \pi\act \lin{\tf{d}_X(a_{X_1},\ldots,a_{X_{k_X}})},
        \end{align*}
        so $\pi\act \Int{X}{\nalg{F}_{\cal E}}{\varsigma^*} = \Int{X}{\nalg{F}_{\cal E}}{\varsigma^*}$ holds. Since this holds for all $\catnew{c}$ such that  $\catnew{c}\cap\supp{}{ \Int{X}{\nalg{F}_{\cal E}}{\varsigma^*}} = \emptyset$, the result follows.
    \end{itemize}
    \end{enumerate}
\end{proof}

Let $\Pi$ be a proof of the judgment$~\vdash t\sigma \aeq{C} u\sigma$. By the definition of $\sigma$, the terms $t\sigma$ and $u\sigma$ belong to $\F(\Sigma^+)$, and the subterm $X\sigma$ occurs in the position where $X$ occurs in $t$ or $u$.  Now, we aim to find a method to reverse $t\sigma$ or $u\sigma$ into terms in $\Sigma$, in such a way that it allows us to recover a derivation of $\Upsilon_{\catnew{c}} \vdash t\aeq{C} u$.

\begin{definition}
    Let $\mathcal{A}^+$ be $\mathcal{A}$ extended with:
    \begin{itemize}
        \item a set $\mathcal{C}$ of atoms mentioned anywhere in $\Pi$ (that were not already in $\mathcal{A}$).

        \item a set $\mathcal{B}$ of fresh atoms, in bijection with $\mathcal{A}$ — for convenience, we fix a bijection and write $b_{X_i}$ for the atom corresponding under that bijection with $a_{X_i}$ — and

        \item one fresh atom $e$ (so $e$ does not occur in $\mathcal{A}, \mathcal{C}$ or $\mathcal{B}$).
    \end{itemize}
    \[
        {\cal A}^+ = {\cal A\cup C\cup B}\cup\{e\}.
    \]
\end{definition}

Let ${\cal B}_X$ be the set of $b_{X_i}$'s corresponding with $a_{X_i}$'s. Denote by $({\cal A}_X \ {\cal B}_X)$ the permutation $(a_{X_1} \ b_{X_1})\circ \ldots\circ (a_{X_{k_X}} \ b_{X_{k_X}})$.

\begin{definition}
    Define $\Upsilon^+_{\atnew{\pvec{c}'}}$ to be $\Upsilon_{\catnew{c}}$ extended with the following constraints:
    \begin{itemize}
        \item $\atnew{\pvec{c}'} = \catnew{c}\cup\{c^*\}$ where $\atnew{c^*}$ is a fresh name.

        \item $(\atnew{c_1^*} \ \atnew{c_2^*} \ \ldots \ \atnew{c_n^*} \ e \ \atnew{c^*})\fix{C} X$ for all $X\in{\cal X}$, where ${\cal C} = \{\atnew{c_1^*},\ldots,\atnew{c_n^*}\}$.

        \item $\pi^{({\cal A}_X \ {\cal B}_X)}\fix{C} Y$ for all $\pi\fix{C} Y\in \Upsilon_{\fix{C}}$.
    \end{itemize}
\end{definition}

\begin{definition}
    For the rest of this subsection let $g$ and $h$ range over ground terms in $\Sigma^+$ that mention only atoms from $\mathcal{A}^+\setminus\mathcal{B}\cup\{e\}$. Define an \emph{inverse mapping} $(-)^{-1}$ from such ground terms to terms in $\Sigma$ inductively as follows:
    $$
    \begin{array}{rcl}
    a^{-1}&\equiv& a \\
    \tf{d}_X()^{-1} &\equiv & e\\
    \tf{f}(g_1,\ldots,g_n)^{-1}&\equiv &\tf{f}(g_1^{-1},\ldots,g_n^{-1}) \\
    ([a]g)^{-1}&\equiv& [a]g^{-1}\\
    \tf{d}_X(a_{X_1}', \ldots, a_{X_{k_{\scalebox{.4}{$X$}}}}')^{-1} &\equiv& \pi_X(a_{X_1}', \ldots, a_{X_{k_{\scalebox{.4}{$X$}}}}')\act X
    \end{array}
    $$
    where $\pi_X(a_{X_1}', \ldots, a_{X_{k_{\scalebox{.4}{$X$}}}}') = (a_{X_1}' \ b_{X_1})\circ \ldots\circ (a_{X_{k_{\scalebox{.4}{$X$}}}}' \ b_{X_{k_{\scalebox{.4}{$X$}}}})\circ (b_{X_1} \ a_{X_1})\circ \ldots\circ (b_{X_{k_{\scalebox{.4}{$X$}}}} \ a_{X_{k_{\scalebox{.4}{$X$}}}}).$
\end{definition}

\begin{lemma}\label{alemma:d-X}
    $\tf{d}_{X}(a_{X_1}, \ldots, a_{X_{k_X}})^{-1} \equiv X$.
\end{lemma}

\begin{proof}
    Consequence of $\pi_X(a_{X_1}, \ldots, a_{X_{k_X}}) = \id$.
\end{proof}

The inverse mapping is equivariant (for the terms we care about):

\begin{lemma}\label{alemma:equivariance-for-completeness}
    \sloppy{If $\pi\in \bigcap_{X\in\var{g^{-1}}}\Perm{\atm{(\Upsilon_{\atnew{\pvec{c}'}}|_X)_{\fresh}}\cup{\cal C}}\alert{\circ}\pair{\perm{}{(\Upsilon_{\atnew{\pvec{c}'}}|_X)_{\fix{C}}}}$, then  $\Upsilon^+_{\atnew{\pvec{c}'}} \vdash  (\pi\act g)^{-1} \aeq{C} \pi\act g^{-1}$.}
\end{lemma}

\begin{proof}
The proof is by induction on the structure of $g$. The only non-trivial case is when $g \equiv \tf{d}_X(a_{X_1}',\ldots,a_{X_{k_X}}')$. To prove $\Upsilon^+_{\atnew{\pvec{c}'}} \vdash  (\pi\act g)^{-1} \aeq{C} \pi\act g^{-1}$, we must show that $(\pi^{-1})^{({\cal A}_X \ {\cal B}_X)^{-1}} \in \PN{}{\Upsilon^+_{\atnew{\pvec{c}'}}|_X} $. This is equivalent to prove that $\pi^{({\cal A}_X \ {\cal B}_X)} \in \PN{}{\Upsilon^+_{\atnew{\pvec{c}'}}|_X}$ because $({\cal A}_X \ {\cal B}_X)^{-1} = ({\cal A}_X \ {\cal B}_X)$. In this case, $\var{g^{-1}} = \{X\}$ and so $\pi\in \Perm{\atm{(\Upsilon_{\atnew{\pvec{c}'}}|_X)_{\fresh}}\cup{\cal C}}\alert{\circ}\pair{\perm{}{(\Upsilon_{\atnew{\pvec{c}'}}|_X)_{\fix{C}}}}$. To simply notation, let's call it just $\Perm{(\Upsilon_{\atnew{\pvec{c}'}}|_X)_{\fresh}\cup{\cal C}}\alert{\circ}\pair{(\Upsilon_{\atnew{\pvec{c}'}}|_X)_{\fix{C}}}$. Applying $({\cal A}_X \ {\cal B}_X)$ on both sides, we get $\pi^{({\cal A}_X \ {\cal B}_X)}\in  \Perm{(\Upsilon_{\atnew{\pvec{c}'}}|_X)_{\fresh}\cup{\cal C}}\alert{\circ}\pair{(\Upsilon_{\atnew{\pvec{c}'}}^{({\cal A}_X \ {\cal B}_X)}|_X)_{\fix{C}}}$. Since
\begin{itemize}
    \item $\PN{}{\Upsilon^+_{\atnew{\pvec{c}'}}|_X} =  \Perm{(\Upsilon^+_{\atnew{\pvec{c}'}}|_X)_{\fresh}\cup{\cal C}}\alert{\circ}\pair{(\Upsilon_{\atnew{\pvec{c}'}}^+|_X)_{\fix{C}}}$
    \item $\Perm{(\Upsilon_{\atnew{\pvec{c}'}}|_X)_{\fresh}\cup{\cal C}} \subseteq \Perm{(\Upsilon^+_{\atnew{\pvec{c}'}}|_X)_{\fresh}\cup{\cal C}}$
    \item $\pair{(\Upsilon_{\atnew{\pvec{c}'}}^{({\cal A}_X \ {\cal B}_X)}|_X)_{\fix{C}}} \subseteq \pair{(\Upsilon_{\atnew{\pvec{c}'}}^+|_X)_{\fix{C}}}$
\end{itemize}
it follows that $\pi^{({\cal A}_X \ {\cal B}_X)} \in  \PN{}{\Upsilon^+_{\atnew{\pvec{c}'}}|_X}$.
\end{proof}

\begin{lemma}\label{alemma:inversion-properties}
The following hold:
\begin{enumerate}
\item \label{alemma:substitution-inverse-fixed-point}
    $\Upsilon^+_{\atnew{\pvec{c}'}}\vdash (t\sigma)^{-1} \aeq{C} t$ and $\Upsilon^+_{\atnew{\pvec{c}'}}\vdash (u\sigma)^{-1} \aeq{C} u$.
\item \label{alemma:ground-derivation-implies-derivation}
    If $~\vdash t\sigma \aeq{C} u\sigma$, then $\Upsilon_{\catnew{c}} \vdash t \aeq{C} u$.
\end{enumerate}
\end{lemma}

\begin{proof}
    \begin{enumerate}
        \item The proof is by induction on $t$. The only interesting case is when $t\equiv \pi_1\act X$. First note that

\begin{itemize}
    \item $X\sigma \in \F(\Sigma^+)$, and $X\sigma$ only involves atoms from $\mathcal{A}^+ \setminus \mathcal{B} \cup \{e\}$.

    \item By Lemma~\ref{alemma:var-judge-valid}, $t \equiv \pi_2 \act X$. Since $\Upsilon_{\catnew{c}} \vDash \pi_1 \act X \aeq{C} \pi_2 \act X$ if and only if $\Upsilon_{\catnew{c}} \vDash (\pi_2^{-1} \circ \pi_1) \act X \aeq{C} \pi_2 \act X$, we may assume w.l.o.g. that $s \equiv \pi \act X$ and $t \equiv X$.

    \item By Lemma~\ref{alemma:d-X}, $(X\sigma)^{-1} \equiv X$ and so $\var{(X\sigma)^{-1}} = \{X\}$.

    \item By Lemma~\ref{alemma:valid-judge-pi-generated}, we get
    \[
        \pi \in\PN{}{\Upsilon_{\catnew{c}}|_X} \subseteq \Perm{\atm{(\Upsilon_{\atnew{\pvec{c}'}}|_X)_{\fresh}}\cup {\cal C}}\alert{\circ}\pair{\perm{}{(\Upsilon_{\atnew{\pvec{c}'}}|_X)_{\fix{C}}}}
    \]
\end{itemize}

Applying Lemma~\ref{alemma:equivariance-for-completeness}, we conclude that $\Upsilon^+_{\atnew{\pvec{c}'}} \vdash (\pi \act (X\sigma))^{-1} \aeq{C} \pi \act (X\sigma)^{-1}$. Finally, by applying Lemma~\ref{alemma:d-X} once more, and using that $(\pi\act X)\sigma \equiv \pi\act(X\sigma)$, we derive $\Upsilon^+_{\atnew{\pvec{c}'}} \vdash ((\pi \act X)\sigma)^{-1} \aeq{C} \pi \act X$.

\item Firstly, let's prove the following claim: If $~\vdash s\sigma \aeq{C} t\sigma$ then $\Upsilon^+_{\atnew{\pvec{c}'}} \vdash (s\sigma)^{-1} \aeq{C} (t\sigma)^{-1}$.

     The proof is by induction on the last rule used to obtain the derivation $~\vdash s\sigma \aeq{C} t\sigma$.

         \begin{itemize}
             \item If the last rule applied is $(\frule{\faeq{C}}{a})$, then by Inversion (Theorem~\ref{thm:miscellaneous}(\ref{thm:inversion})), $s\sigma \equiv a \equiv t\sigma$. Consequently, $(a\sigma)^{-1} \equiv a \equiv (t\sigma)^{-1}$. Therefore, the derivation $\Upsilon^+_{\atnew{\pvec{c}'}} \vdash (a\sigma)^{-1} \aeq{C} (a\sigma)^{-1}$ follows by rule $(\frule{\faeq{C}}{a})$.

             \item If the last rule applied is $(\frule{\faeq{C}}{\tf{f}})$, then by Inversion (Theorem~\ref{thm:miscellaneous}(\ref{thm:inversion})), $s\sigma \equiv \tf{f}(g_1,\ldots,g_n)$, $t\sigma \equiv \tf{f}(g_1',\ldots,g_n')$,
              and $\vdash g_i \aeq{C} g_i'$ for all $i = 1,\ldots,n$.
              By the construction of $\sigma$, it forces $s \equiv \tf{f}(s_1,\ldots,s_n)$ and $t \equiv \tf{f}(t_1,\ldots,t_n)$ because the substitution does not  change the structure of the term $s$ and $t$. Thus, $s_i\sigma \equiv g_i$ and $t_i\sigma \equiv g_i'$ for all $i=1,\ldots,n$. Then, for each $i=1,\ldots,n$, the derivation $\vdash g_i \aeq{C} g_i'$ become $\vdash s_i\sigma \aeq{C} t_i\sigma$. By induction, we have that  $\Upsilon^+_{\atnew{\pvec{c}'}} \vdash (s_i\sigma)^{-1} \aeq{C} (t_i\sigma)^{-1}$. Then, by an application of the rule $(\frule{\faeq{C}}{\tf{f}})$, we obtain $\Upsilon^+_{\atnew{\pvec{c}'}} \vdash (\tf{f}(s_1,\ldots,s_n)\sigma)^{-1} \aeq{C} (\tf{f}(t_1,\ldots,t_n)\sigma)^{-1}$.

              \item If the last rule applied is $(\frule{\faeq{C}}{\tf{f^C}})$, then the it follows similarly to the previous case.

               In this case, by Inversion (Theorem~\ref{thm:miscellaneous}(\ref{thm:inversion})), $s\sigma \equiv \tf{f^C}(g_0,g_1)$ and $t\sigma \equiv \tf{f^C}(g_0',g_1')$, and
               \begin{itemize}
            \item  either $\vdash g_0 \aeq{C} g_0'$ and $\vdash g_1 \aeq{C} g_1'$

            This case is analogous to the case of the rule $(\frule{\faeq{C}}{\tf{f}})$.

            \item  or $\vdash g_0 \aeq{C} g_1'$ and $\vdash g_1 \aeq{C} g_0'$.

            This case is also analogous to the case of the rule $(\frule{\faeq{C}}{\tf{f}})$, but we will prove it anyway.

            By the construction of $\sigma$, we conclude that $s\equiv \tf{f^C}(s_0,s_1)$ and $t\equiv\tf{f^C}(t_0,t_1)$ and, for $i=0,1$, $s_i\sigma \equiv g_i$ and $t_i\sigma \equiv g_i'$. Thus, $\vdash g_0 \aeq{C} g_1'$ and and $\vdash g_1 \aeq{C} g_0'$ becomes, respectively, $\vdash s_0\sigma \aeq{C} t_1\sigma$ and $\vdash s_1\sigma \aeq{C} t_0\sigma$. The result follows by rule $(\frule{\faeq{C}}{\tf{f^C}})$.
         \end{itemize}

          \item If the last rule applied is $(\frule{\faeq{C}}{[a]})$, then by Inversion (Theorem~\ref{thm:miscellaneous}(\ref{thm:inversion})), $s\sigma \equiv [a]g_1, t\sigma \equiv [a]g_2$, and $~\vdash g_1 \aeq{C} g_2$. Again, the definition of $\sigma$ forces $s \equiv [a]s'$ and $t\equiv [a]t'$, which in turn yields $s'\sigma \equiv g_1$ and $t'\sigma \equiv g_2$. Thus, $~\vdash g_1 \aeq{C} g_2$ becomes $~\vdash s'\sigma \aeq{C} t'\sigma$. By induction, $\Upsilon^+_{\atnew{\pvec{c}'}}\vdash s'\aeq{C} t'$ and the result follows by rule $(\frule{\faeq{C}}{[a]})$.

         \item If the last rule applied is $(\frule{\faeq{C}}{ab})$, then by Inversion (Theorem~\ref{thm:miscellaneous}(\ref{thm:inversion})), $ s\sigma \equiv [a]g_1, t\sigma \equiv [b]g_2$, and $~\vdash \newswap{a}{c_1}\act g_1 \aeq{C} \newswap{b}{c_1}\act g_2$ where $\atnew{c_1}\notin a,b,g_1,g_2,\atnew{\pvec{c}'}$.

         The definition of $\sigma$ forces $s \equiv [a]s'$ and $t\equiv [b]t'$, which in turn yields $s'\sigma \equiv g_1$ and $t'\sigma\equiv g_2$. Thus,
              \begin{prooftree}
                 \AxiomC{$\vdash \newswap{a}{c_1}\act g_1 \aeq{C} \newswap{b}{c_1}\act g_2$}
                 \dashedLine
                  \UnaryInfC{$\vdash \newswap{a}{c_1}\act (s'\sigma) \aeq{C} \newswap{b}{c_1}\act (t'\sigma)$}
                  \dashedLine
                  \UnaryInfC{$\vdash (\newswap{a}{c_1}\act s')\sigma \aeq{C} (\newswap{b}{c_1}\act t')\sigma$}
              \end{prooftree}
              By induction, we get
             \begin{prooftree}
                 \AxiomC{$\Upsilon^+_{\atnew{\pvec{c}',c_1}}\vdash ((\newswap{a}{c_1}\act s')\sigma)^{-1} \aeq{C} (\newswap{b}{c_1}\act t')\sigma)^{-1}$}
                 \dashedLine
                 \UnaryInfC{$\Upsilon^+_{\atnew{\pvec{c}',c_1}}\vdash (\newswap{a}{c_1}\act (s'\sigma))^{-1} \aeq{C} (\newswap{b}{c_1}\act (t'\sigma))^{-1}$}
                 \dashedLine
                 \UnaryInfC{$\Upsilon^+_{\atnew{\pvec{c}',c_1}}\vdash (\newswap{a}{c_1}\act g_1)^{-1} \aeq{C} (\newswap{b}{c_1}\act g_2)^{-1}$}
             \end{prooftree}

             Observe that $\var{g_1^{-1}}= \var{g_2^{-1}}$. If  $\var{g_1^{-1}} \neq \var{g_2^{-1}}$, let's say $Z\in \var{g_1^{-1}}$ and $Z\notin\var{g_2^{-1}}$. By the definition of $\sigma$ and $-^{-1}$, this means that $\tf{d}_Z(a_{Z_1},\ldots,a_{Z_{k_Z}})$ occurs in $g_1$ at the same position that $Z$ occurs in $g_1^{-1}$. However, since $~\vdash \newswap{a}{c_1}\act g_1 \aeq{C} \newswap{b}{c_1}\act g_2$, this means that $\tf{d}_Z(a_{Z_1},\ldots,a_{Z_{k_Z}})$ must occur in $g_2$ at the same position that it occurs in $g_1$. This implies that $Z\in\var{g_2^{-1}}$, which is a contradiction.

             We claim that $a\notin {\cal A}_X$ for every $X\in\var{g_1^{-1}}$. Suppose, by contradiction, that $a\in {\cal A}_X$ for some $X\in\var{g_1^{-1}}$.

             From $~\vdash \newswap{a}{c_1}\act g_1 \aeq{C} \newswap{b}{c_1}\act g_2$, we have $\tf{fn}(\newswap{a}{c_1}\act g_1) = \tf{fn}(\newswap{b}{c_1}\act g_2)$. This implies $a\notin \tf{fn}(g_2)$.

             Since $\var{g_2^{-1}} = \var{g_1^{-1}}$, it follows that $X\in \var{g_2^{-1}}$ and hence $a\in \tf{d}_X(a_{X_1},\ldots,a_{X_{k_X}})$ and $\tf{d}_X(a_{X_1},\ldots,a_{X_{k_X}})$ occurs in $g_2$. Then $a\in \tf{fn}(g_2)$,  a contradiction.

             Since $a\notin {\cal A}_X$ for every $X\in\var{g_1^{-1}}$. By definition this means that for every $X\in\var{g_1^{-1}}$ there is some $\pi\fix{C} X\in\Upsilon_{\fresh}|_X$ such that $a\in \dom{\pi}$. Thus
              \[
                     \newswap{a}{c_1}\in \bigcap_{X\in\var{g_1^{-1}}}\Perm{\atm{\Upsilon_{\fresh}|_{X}}\cup\atnew{\pvec{c}',c_1}\cup{\cal C}}\pair{\Upsilon_{\fix{C}}|_{X}}
              \]
              By Lemma~\ref{alemma:equivariance-for-completeness},  we obtain the derivation $\Upsilon^+_{\atnew{\pvec{c}',c_1}}\vdash (\newswap{a}{c_1}\act g_1)^{-1} \aeq{C} \newswap{a}{c_1}\act g_1^{-1}$.

              Similarly, we have $b\notin {\cal A}_X$ for every $X\in\var{g_2^{-1}}$ and consequently $\Upsilon^+_{\atnew{\pvec{c}',c_1}}\vdash (\newswap{b}{c_1}\act g_2)^{-1} \aeq{C} \newswap{b}{c_1}\act g_2^{-1}$. Therefore, $\Upsilon^+_{\atnew{\pvec{c}',c_1}}\vdash (a \ \atnew{c_1})\act g_1^{-1} \aeq{C} (b \ \atnew{c_1})\act g_2^{-1}$ follows by transitivity. Therefore, the derivation $ \Upsilon^+_{\atnew{\pvec{c}'}} \vdash [a]g_1^{-1} \aeq{C} [b]g_2^{-1}$ follows by rule $(\frule{\faeq{C}}{ab}).$
         \end{itemize}

        \paragraph*{Completing the proof} Now that we proved the claim, let's proceed with the proof of the lemma: by Lemma~\ref{alemma:substitution-inverse-fixed-point}, symmetry and transitivity of $\aeq{C}$, we deduce $\Upsilon^+_{\atnew{\pvec{c}'}}\vdash s\aeq{C} t$ and by Strengthening (Theorem~\ref{thm:miscellaneous}(\ref{thm:strengthening})) and Lemma~\ref{alemma:vacuous-quantification}, we obtain $\Upsilon_{\catnew{c}}\vdash s\aeq{C} t$.
    \end{enumerate}
\end{proof}

With all of this established, we can finally prove Completeness.

\begin{proof}[of Completeness]
    Suppose $\Upsilon_{\catnew{c}} \vDash t\aeq{C} u$, so $\Int{\Upsilon_{\catnew{c}} \vdash t\aeq{C} u}{\nom{F}}{\varsigma^*}$ is valid. We want to show that there is a derivation of $\Upsilon_{\catnew{c}} \vdash t\aeq{C} u$. In fact, by Lemma~\ref{alemma:interpretation}(\ref{alemma:context-valid}) we have $\Int{t}{\nalg{F}}{\varsigma^*} = \Int{u}{\nalg{F}}{\varsigma^*}$. By Lemma~\ref{alemma:inversion-properties}(\ref{alemma:substitution-interpretation}), $\Int{t}{\nalg{F}}{\varsigma^*} = \lin{t\sigma}$ and $\Int{u}{\nalg{F}}{\varsigma^*} = \lin{u\sigma}$. Therefore $\lin{t\sigma} = \lin{u\sigma}$, that is, $~\vdash t\sigma \aeq{C} u\sigma$. It follows by Lemma~\ref{alemma:inversion-properties}(\ref{alemma:ground-derivation-implies-derivation}) that $\Upsilon_{\catnew{c}} \vdash t\aeq{C} u$, as desired.
\end{proof}

\section{Proof of Preservation by substitution}\label{app:preservation-substitutiton}

The next results aim to prove the property of Preservation by substitution. We begin establishing by establishing a lemma stating that fix-point constraints are closed under composition and inverses.

\begin{lemma}\label{alemma:fix-point-composition-inverse}
    \begin{enumerate}
        \item \label{alemma:fix-point-composition} If $\Upsilon_{\catnew{c}} \vdash \pi_1 \fix{C} t$  and $\Upsilon_{\catnew{c}} \vdash \pi_2 \fix{C} t$, then $\Upsilon_{\catnew{c}} \vdash (\pi_1\circ \pi_2) \fix{C} t$.
        \item \label{alemma:fix-point-inverse} If $\Upsilon_{\catnew{c}} \vdash \pi\fix{C} t$, then  $\Upsilon_{\catnew{c}} \vdash \pi^{-1}\fix{C} t$.
    \end{enumerate}
\end{lemma}

\begin{proof}
Direct consequence of Equivariance and Equivalence in Theorem~\ref{thm:miscellaneous}.
\end{proof}

The following lemma formalizes in the calculus the valid result in nominal set theory, which states that any swapping $(a \ b)$ of two fresh names $a$ and $b$ for an element $x$ fixes the element $x$.

\begin{lemma}\label{alemma:two-fresh-names-fix}
    Let $\atnew{c_1},\atnew{c_2}\notin \atm{\Upsilon_{\catnew{c}},t}$. Then  $\Upsilon_{\catnew{c},\atnew{c_1},\atnew{c_2}} \vdash (\atnew{c_1} \ \atnew{c_2})\act t \aeq{C} t$.
\end{lemma}

\begin{proof}
The proof is direct by induction on the structure of $t$. Here we show only the case when $t\equiv \pi\act X$. We aim to prove $\Upsilon_{\catnew{c},\atnew{c_1},\atnew{c_2}} \vdash (\atnew{c_1} \ \atnew{c_2}) \act (\pi\act X) \aeq{C} (\pi\act X)$. Given that $\atnew{c_1},\atnew{c_2} \notin \dom{\pi}$, we have
     \[
         \pi^{-1}\circ(\atnew{c_1} \ \atnew{c_2})\circ\pi = (\atnew{c_1} \ \atnew{c_2})\in \PN{}{\Upsilon_{\catnew{c},\atnew{c_1},\atnew{c_2}}|_X}.
     \]
     Therefore, the result follows by rule $(\frule{\faeq{C}}{var})$.
\end{proof}

The next lemma provides a method for handling vacuous quantifications in a derivation.

\begin{lemma}\label{alemma:vacuous-quantification}
     Suppose $\atnew{c_1}\notin \atm{\Upsilon_{\catnew{c}}}$. If $\Upsilon_{\catnew{c}} \vdash t\aeq{C} u$ then $\Upsilon_{\catnew{c},\atnew{c_1}}\vdash t\aeq{C} u$. The converse holds if $\atnew{c_1}\notin \atm{\Upsilon_{\catnew{c}},t,u}$.
\end{lemma}

\begin{proof}
 By induction on the last rule applied to obtain $\Upsilon_{\catnew{c}} \vdash t\aeq{C} u$. For instance, if the last rule is $(\frule{\faeq}{var})$, then we obtain $\Upsilon_{\catnew{c}} \vdash \pi_1\act X\aeq{C} \pi_2\act X$. By Inversion (Theorem~\ref{thm:miscellaneous}(\ref{thm:inversion})), it follows that $\pi_2^{-1}\circ\pi_1\in \PN{}{\Upsilon_{\catnew{c}}|_X}$. Since $\atnew{c_1}\notin \atm{\Upsilon_{\catnew{c}}}$, we  have $\PN{}{\Upsilon_{\catnew{c}}|_X}\subseteq\PN{}{\Upsilon_{\catnew{c},\atnew{c_1}}|_X} $. Then $\pi_2^{-1}\circ\pi_1\in \PN{}{\Upsilon_{\catnew{c},\atnew{c_1}}|_X}$ and so the result follows by rule $(\frule{\faeq{C}}{var})$.

\end{proof}

The following lemma formalizes a fundamental property of nominal sets, stating that if a permutation's domain contains only fresh names for an element $x$, then $x$ is fixed by that permutation.

\begin{lemma}\label{alemma:fix-point-formed-by-fresh-names}
    Let $I$ be a non-empty, finite set of indices, and let $\{\pi_i \mid i \in I\}$ be a set of permutations. Suppose $\Upsilon_{\catnew{c}, \atnew{c_1}} \vdash \newswap{a}{c_1}\fix{C} t$ for all $a \in \left(\bigcup_{i \in I} \dom{\pnew{\pi_i}{\catnew{c}}}\right) \cup \catnew{c}$, where $\atnew{c_1} \notin \atm{\Upsilon_{\catnew{c}},t, \{\pi_i \mid i \in I\}}$. Then, $\Upsilon_{\catnew{c}} \vdash \pi'\fix{C} t$ for all $\pi'$ such that $\dom{\pi'} \subseteq \left(\bigcup_{i \in I} \dom{\pnew{\pi_i}{\catnew{c}}}\right) \cup \catnew{c}$.
\end{lemma}

\begin{proof}
    We proceed by induction on the number of disjoint cycles in the decomposition of $\pi'$.

    \paragraph*{Base case.}

    Suppose $\pi'$ is a single cycle, say $\pi' = (a_1 \ a_2 \ \ldots \ a_m)$, which can be expressed as product of swappings:
     \[
        \pi' = (a_1 \ a_m)\circ (a_1 \ a_{m-1})\circ(a_1 \ a_3)\circ(a_1 \ a_2)
    \]
    For each $k\in\{2, \ldots, m\}$: $(a_1 \ a_k) = \newswap{a_1}{c_1}\circ \newswap{a_k}{c_1} \circ \newswap{a_1}{c_1}.$ By the assumption in the statement, we have $\Upsilon_{\catnew{c},\atnew{c_1}} \vdash \newswap{a_k}{c_1}\fix{C} t$ for all $k\in \{2, \ldots, m\}$. By Lemma~\ref{alemma:fix-point-composition-inverse}(\ref{alemma:fix-point-composition}), we have
    \[
        \Upsilon_{\catnew{c},\atnew{c_1}} \vdash \newswap{a_1}{c_1}\circ \newswap{a_k}{c_1} \circ \newswap{a_1}{c_1}\fix{C} t
    \]
    which implies $\Upsilon_{\catnew{c},\atnew{c_1}} \vdash (a_1 \ a_k)\fix{C} t$. By applying Lemma~\ref{alemma:fix-point-composition-inverse}(\ref{alemma:fix-point-composition}) again, we obtain:
    \[
        \Upsilon_{\catnew{c},\atnew{c_1}} \vdash (a_1 \ a_m)\circ (a_1 \ a_{m-1})\circ(a_1 \ a_3)\circ(a_1 \ a_2)\fix{C} t.
    \]
    which is the same as $\Upsilon_{\catnew{c},\atnew{c_1}} \vdash (a_1 \ a_2 \ \ldots \ a_m)\fix{C} t$. By Lemma~\ref{alemma:vacuous-quantification}, we conclude that $\Upsilon_{\catnew{c}} \vdash \pi'\fix{C} t$.

    \paragraph*{Inductive step:} Now, assume the result holds for all permutations $\pi'$ (under the conditions of the statement) that decompose into $n-1$ disjoint cycles. Consider a permutation $\pi'$ such that $\dom{\pi'} \subseteq \left(\bigcup_{i \in I} \dom{\pnew{\pi_i}{\catnew{c}}}\right) \cup \catnew{c}$ and consisting of $n$ disjoint cycles: $\pi' = \eta_1\circ\ldots\circ\eta_n.$ Since $\dom{\eta_1 \circ \ldots \circ \eta_{n-1}} \subseteq \dom{\pi'}$, it follows:
    \[
        \dom{\eta_1\circ\ldots\circ\eta_{n-1}} \subseteq \left(\bigcup_{i\in I}\dom{\pnew{\pi_i}{\catnew{c}}}\right)\cup\catnew{c}.
    \]
    By the induction hypothesis, we have $\Upsilon_{\catnew{c},\atnew{c_1}} \vdash (\eta_1\circ\ldots\circ\eta_{n-1})\fix{C} t.$ Using a similar argument to the base case, we show that $\Upsilon_{\catnew{c},\atnew{c_1}} \vdash \eta_n\fix{C} t$. By  Lemma~\ref{alemma:fix-point-composition-inverse}(\ref{alemma:fix-point-composition}), it follows that $\Upsilon_{\catnew{c},\atnew{c_1}} \vdash \pi'\fix{C} t$ which, by Lemma~\ref{alemma:vacuous-quantification}, yields $\Upsilon_{\catnew{c}} \vdash \pi'\fix{C} t$.
\end{proof}

The following lemma describes a particular derivation case, which will be useful for establishing the next result.

\begin{lemma}\label{alemma:characterization-fix-commutative}
    Suppose $\atnew{c_1}\notin \atm{\Upsilon_{\catnew{c}},a,\tf{f^C}(t_0,t_1)}$. Then $\Upsilon_{\catnew{c},\atnew{c_1}} \vdash \newswap{a}{c_1}\fix{C} \tf{f^C}(t_0,t_1)$ iff $\Upsilon_{\catnew{c},\atnew{c_1}}\vdash \newswap{a}{c_1}\fix{C} t_0$ and $\Upsilon_{\catnew{c},\atnew{c_1}}\vdash \newswap{a}{c_1}\fix{C} t_1$.
\end{lemma}

\begin{proof}
\begin{description}
    \item[$(\Leftarrow)$] This case follows by applying the rule $(\frule{\faeq{C}}{\tf{f^C}})$.

    \item[$(\Rightarrow)$]

    The proof follows by an analysis on every possible pair of terms $(t_0,t_1)$. Since $\Upsilon_{\catnew{c},\atnew{c_1}} \vdash \newswap{a}{c_1}\fix{C} \tf{f^C}(t_0,t_1)$ iff $\Upsilon_{\catnew{c},\atnew{c_1}} \vdash \newswap{a}{c_1}\fix{C} \tf{f^C}(t_1,t_0)$ holds, it's sufficient to analyse only half of the possibilities. The only non-trivial cases are the cases where $t_0$ and $t_1$ are either both suspensions or both abstractions.

    \begin{enumerate}
        \item $t_0 \equiv \pi\act X$ and $t_1 \equiv \pi'\act Y$.

        In this case, $\Upsilon_{\catnew{c},\atnew{c_1}} \vdash \newswap{a}{c_1}\fix{C} \tf{f^C}(\pi\act X,\pi'\act Y)$ implies $ \Upsilon_{\catnew{c},\atnew{c_1}} \vdash \tf{f^C}(\newswap{a}{c_1}\act(\pi\act X),\newswap{a}{c_1}\act(\pi'\act Y))\aeq{C} \tf{f^C}(\pi\act X,\pi'\act Y).$

        \begin{enumerate}
            \item $Y\not\equiv X$.

            By Inversion (Theorem~\ref{thm:miscellaneous}(\ref{thm:inversion})), the valid branch is: $\Upsilon_{\catnew{c},\atnew{c_1}} \vdash \newswap{a}{c_1}\act(\pi\act X) \aeq{C}\pi\act X$ and $\Upsilon_{\catnew{c},\atnew{c_1}} \vdash \newswap{a}{c_1}\act(\pi'\act Y)\aeq{C}\pi'\act Y$, which are precisely the derivations we wanted.

            \item $Y\equiv X$.

            By Inversion (Theorem~\ref{thm:miscellaneous}(\ref{thm:inversion})), there are two possible branches:
            \begin{itemize}
               \item $\Upsilon_{\catnew{c},\atnew{c_1}}\vdash \newswap{a}{c_1}\act(\pi\act X) \aeq{C}\pi\act X$ and $\Upsilon_{\catnew{c},\atnew{c_1}}\vdash \newswap{a}{c_1}\act(\pi'\act X)\aeq{C}\pi'\act X$, and the result follows trivially.

            \item $\Upsilon_{\catnew{c},\atnew{c_1}}\vdash \newswap{a}{c_1}\act(\pi\act X)\aeq{C}\pi'\act X$ and $\Upsilon_{\catnew{c},\atnew{c_1}}\vdash \newswap{a}{c_1}\act(\pi'\act X)\aeq{C}\pi\act X$.

            Note that, by the Equivariance (Theorem~\ref{thm:miscellaneous}(\ref{thm:object-equivariance})) and the symmetry of $\aeq{C}$, these two derivations are equivalent, so it's enough to focus on just one. We will work with the first one. Applying Inversion (Theorem~\ref{thm:miscellaneous}(\ref{thm:inversion})), we obtain: $\pi'^{-1} \circ\newswap{a}{c_1}\circ \pi\in \PN{}{\Upsilon_{\catnew{c},\atnew{c_1}}|_X}$. To simplify the proof, let's call $\gamma := \pi'^{-1} \circ\newswap{a}{c_1}\circ \pi$

            \paragraph*{Case  $\gamma = \id$.} In this case we have
            \[
               \pi'^{-1}(a)  = \gamma(\atnew{c_1}) = \id(\atnew{c_1}) = \atnew{c_1}.
            \]
            and
            \[
                \atnew{c_1}  = \gamma(\pi^{-1}(a)) = \id(\pi^{-1}(a)) = \pi^{-1}(a).
            \]
            Therefore, $\newswap{a}{c_1}^{\pi'^{-1}} =\newswap{\pi'^{-1}(a)}{c_1} = \id \in \PN{}{\Upsilon_{\catnew{c},\atnew{c_1}}|_X}$. Similarly, $\newswap{a}{c_1}^{\pi^{-1}} =\newswap{\pi^{-1}(a)}{c_1} = \id \in \PN{}{\Upsilon_{\catnew{c},\atnew{c_1}}|_X}$. Thus, the result follows by rule $(\frule{\fix{C}}{var})$.

            \paragraph*{Case $\gamma \neq \id$.}
            As observed $\gamma(\pi^{-1}(a)) = \atnew{c_1}$ and $\gamma(\atnew{c_1}) = \pi'^{-1}(a)$. Since $\dom{\gamma}$ is finite, there is a $m\geq 3$ such that $\gamma^m(\pi^{-1}(a)) = \pi^{-1}(a)$. Then
            \[
                \gamma = \underbrace{(\pi^{-1}(a) \ \atnew{c_1} \ \pi'^{-1}(a) \ \gamma^4(\pi^{-1}(a)) \ \ldots \ \gamma^{m-1}(\pi^{-1}(a)))}_{\gamma'}\circ \rho'
            \]
            where $\rho'$ is disjoint from $\gamma'$. Then $\gamma' \in \Perm{\atm{(\Upsilon_{\catnew{c},\atnew{c_1}}|_X)_{\fresh}}}$ and consequently $\pi^{-1}(a),\pi'^{-1}(a) \in \atm{(\Upsilon_{\catnew{c},\atnew{c_1}}|_X)_{\fresh}}$. So it follows that $\newswap{a}{c_1}^{\pi'^{-1}} = \newswap{\pi'^{-1}(a)}{c_1} \in \PN{}{\Upsilon_{\catnew{c},\atnew{c_1}}|_X}$ and  $\newswap{a}{c_1}^{\pi^{-1}} = \newswap{\pi^{-1}(a)}{c_1} \in \PN{}{\Upsilon_{\catnew{c},\atnew{c_1}}|_X}$ and the result follows by $(\frule{\faeq{C}}{var})$.
         \end{itemize}
        \end{enumerate}

        \item $t_0$ and $t_1$ are abstractions.
        \begin{enumerate}
            \item $t_0 \equiv [a]t_0'$ and $t_1 \equiv [a]t_1'$.

            Then $\Upsilon_{\catnew{c},\atnew{c_1}} \vdash \newswap{a}{c_1}\fix{C} \tf{f^C}([a]t_0',[a]t_1')$ implies $ \Upsilon_{\catnew{c},\atnew{c_1}} \vdash \tf{f^C}(\newswap{a}{c_1}\act[a]t_0',\newswap{a}{c_1}\act[a]t_1')\aeq{C} \tf{f^C}([a]t_0',[a]t_1')$. By Inversion (Theorem~\ref{thm:miscellaneous}(\ref{thm:inversion})), we have two possible branches:

        \begin{itemize}
            \item  $\Upsilon_{\catnew{c},\atnew{c_1}} \vdash \newswap{a}{c_1}\act [a]t_0'\aeq{C} [a]t_0'$ and $\Upsilon_{\catnew{c},\atnew{c_1}} \vdash \atnew{c_1}\act [a]t_1'\aeq{C} [a]t_1'$.

            Then the result follows.

            \item $\Upsilon_{\catnew{c},\atnew{c_1}} \vdash \newswap{a}{c_1}\act [a]t_0'\aeq{C} [a]t_1'$ and $\Upsilon_{\catnew{c},\atnew{c_1}} \vdash \newswap{a}{c_1}\act [a]t_1'\aeq{C} [a]t_0'$.

            Note that, by Equivariance (Theorem~\ref{thm:miscellaneous}(\ref{thm:object-equivariance})) and Equivalence (Theorem~\ref{thm:miscellaneous}(\ref{thm:alpha-equivalence})), these two derivations are equivalent, so it's sufficient with just one. We will work with the first one.

            We claim that $\Upsilon_{\catnew{c},\atnew{c_1}} \vdash [a]t_0' \aeq{C} [a]t_1'$. In fact, observe that the derivation $\Upsilon_{\catnew{c},\atnew{c_1}} \vdash \newswap{a}{c_1} \act [a]t_0' \aeq{C} [a]t_1'$ is equivalent to $\Upsilon_{\catnew{c},\atnew{c_1}} \vdash [\atnew{c_1}]\newswap{a}{c_1} \act t_0' \aeq{C} [a]t_1'$. By applying Inversion (Theorem~\ref{thm:miscellaneous}(\ref{thm:inversion})), we derive $\Upsilon_{\catnew{c},\atnew{c_1},\atnew{c_2}} \vdash ((\atnew{c_1} \ \atnew{c_2}) \circ \newswap{a}{c_1}) \act t_0' \aeq{C} \newswap{a}{c_2} \act t_1'$ where $\atnew{c_2}\notin \atm{\Upsilon_{\catnew{c},\atnew{c_1}},a,t_0',t_1'}$. By Equivariance  (Theorem~\ref{thm:miscellaneous}(\ref{thm:object-equivariance}), this results in $\Upsilon_{\catnew{c},\atnew{c_1},\atnew{c_2}} \vdash (\newswap{a}{c_2}^{-1} \circ (\atnew{c_1} \ \atnew{c_2}) \circ \newswap{a}{c_1}) \act t_0' \aeq{C} t_1'$. Since $(\newswap{a}{c_2}^{-1} \circ (\atnew{c_1} \ \atnew{c_2}) \circ \newswap{a}{c_1}) = (\atnew{c_1} \ \atnew{c_2})$, we conclude that $\Upsilon_{\catnew{c},\atnew{c_1},\atnew{c_2}} \vdash (\atnew{c_1} \ \atnew{c_2}) \act t_0' \aeq{C} t_1'$. On the other hand, because $\atnew{c_1},\atnew{c_2}\notin \atm{\Upsilon_{\catnew{c}},a,t_0',t_1'}$, Lemma~\ref{alemma:two-fresh-names-fix} gives us $\Upsilon_{\catnew{c},\atnew{c_1},\atnew{c_2}} \vdash (\atnew{c_1} \ \atnew{c_2})\act t_0' \aeq{C} t_0'$. This combined with Equivalence (Theorem~\ref{thm:miscellaneous}(\ref{thm:alpha-equivalence})) provides us $\Upsilon_{\catnew{c},\atnew{c_1},\atnew{c_2}} \vdash t_0' \aeq{C} t_1'$. Thus, by Lemma~\ref{alemma:vacuous-quantification}, we obtain $\Upsilon_{\catnew{c},\atnew{c_1}} \vdash  t_0' \aeq{C} t_1'$, which, by rule $(\frule{\faeq{C}}{[a]})$, leads to $\Upsilon_{\catnew{c},\atnew{c_1}} \vdash [a]t_0' \aeq{C} [a]t_1'$.

            Now, using the claim we just proved and Equivalence (Theorem~\ref{thm:miscellaneous}(\ref{thm:alpha-equivalence})), we obtain $\Upsilon_{\catnew{c},\atnew{c_1}} \vdash \newswap{a}{c_1}\act [a]t_0'\aeq{C} [a]t_0'$ and $\Upsilon_{\catnew{c},\atnew{c_1}} \vdash \newswap{a}{c_1}\act [a]t_1'\aeq{C} [a]t_1'$ and so the result follows.
            \end{itemize}

            \item $t_0\equiv [a]t_0'$ and $t_1 \equiv [d']t_1'$.

            Then $\Upsilon_{\catnew{c},\atnew{c_1}}\vdash \newswap{a}{c_1}\fix{C} \tf{f^C}([a]t_0',[d']t_1')$ implies $\Upsilon_{\catnew{c},\atnew{c_1}}\vdash \tf{f^C}( \newswap{a}{c_1}\act[a]t_0', \newswap{a}{c_1}\act  [d']t_1')\aeq{C} \tf{f^C}([a]t_0',[d']t_1')$. By Inversion (Theorem~\ref{thm:miscellaneous}(\ref{thm:inversion})), there are two branches to consider:
        \begin{itemize}
             \item $\Upsilon_{\catnew{c},\atnew{c_1}}\vdash  \newswap{a}{c_1}\act[a]t_0'\aeq{C} [a]t_0'$ and $\Upsilon_{\catnew{c},\atnew{c_1}}\vdash  \newswap{a}{c_1}\act[d']t_1'\aeq{C} [d']t_1'$.

             Then the result follows.

            \item $\Upsilon_{\catnew{c},\atnew{c_1}}\vdash  \newswap{a}{c_1}\act[a]t_0'\aeq{C} [d']t_1'$ and $\Upsilon_{\catnew{c},\atnew{c_1}}\vdash  \newswap{a}{c_1}\act[d']t_1'\aeq{C} [a]t_0'$.

            Note that, by Equivariance (Theorem~\ref{thm:miscellaneous}(\ref{thm:object-equivariance})) and Equivalence (Theorem~\ref{thm:miscellaneous}(\ref{thm:alpha-equivalence})), these two derivations are equivalent, so it's sufficient to work with just one. We will work with the first one.

            We claim that $\Upsilon_{\catnew{c},\atnew{c_1}} \vdash  [d']t_1' \aeq{C} [a]t_0'$. Indeed, observe that the derivation $\Upsilon_{\catnew{c},\atnew{c_1}} \vdash \newswap{a}{c_1} \act [a]t_0' \aeq{C} [d']t_1'$ can be rewritten as $\Upsilon_{\catnew{c},\atnew{c_1}} \vdash [\atnew{c_1}]\newswap{a}{c_1} \act t_0' \aeq{C} [d']t_1'$. By applying Inversion (Theorem~\ref{thm:miscellaneous}(\ref{thm:inversion})), we derive $\Upsilon_{\catnew{c},\atnew{c_1},\atnew{c_2}} \vdash ((\atnew{c_1} \ \atnew{c_2}) \circ \newswap{a}{c_1}) \act t_0' \aeq{C} \newswap{d'}{c_2} \act t_1'$ where $\atnew{c_2}\notin \atm{\Upsilon_{\catnew{c}},a,d',t_0',t_1'}$.

            By Equivariance (Theorem~\ref{thm:miscellaneous}(\ref{thm:object-equivariance})), we have $\Upsilon_{\catnew{c},\atnew{c_1},\atnew{c_2}} \vdash (\newswap{d'}{c_2}^{-1} \circ (\atnew{c_1} \ \atnew{c_2}) \circ \newswap{a}{c_1}) \act t_0' \aeq{C} t_1'$, which is the same as $\Upsilon_{\catnew{c},\atnew{c_2},\atnew{c_1}} \vdash (a \ d' \ \atnew{c_1} \ \atnew{c_2}) \act t_0' \aeq{C} t_1'$.

            Now, by applying Lemma~\ref{alemma:two-fresh-names-fix}, we have $\Upsilon_{\catnew{c},\atnew{c_1},\atnew{c_2}} \vdash (\atnew{c_1} \ \atnew{c_2}) \act t_0' \aeq{C} t_0'$. Since $((a \ d' \ \atnew{c_2})^{-1} \circ (a \ d' \ \atnew{c_2} \ \atnew{c_1})) = (\atnew{c_1} \ \atnew{c_2})$, we conclude that $\Upsilon_{\catnew{c},\atnew{c_1},\atnew{c_2}} \vdash ((a \ d' \ \atnew{c_2})^{-1} \circ (a \ d' \ \atnew{c_2} \ \atnew{c_1})) \act t_0' \aeq{C} t_0'$.

            By Equivariance (Theorem~\ref{thm:miscellaneous}(\ref{thm:object-equivariance})), we get $\Upsilon_{\catnew{c},\atnew{c_1},\atnew{c_2}} \vdash (a \ d' \ \atnew{c_2} \ \atnew{c_1}) \act t_0' \aeq{C} (a \ d' \ \atnew{c_2}) \act t_0'$. This, together with the derivation $\Upsilon_{\catnew{c},\atnew{c_1},\atnew{c_2}} \vdash (a \ d' \ \atnew{c_2} \ \atnew{c_1}) \act t_0' \aeq{C} t_1'$, and Equivalence (Theorem~\ref{thm:miscellaneous}(\ref{thm:alpha-equivalence}))), yields $\Upsilon_{\catnew{c},\atnew{c_1},\atnew{c_2}} \vdash (a \ d' \ \atnew{c_2}) \act t_0' \aeq{C} t_1'$.

            Since $(a \ d' \ \atnew{c_2}) = \newswap{d'}{c_2} \circ \newswap{a}{c_2}$, we get $\Upsilon_{\catnew{c},\atnew{c_1},\atnew{c_2}} \vdash (\newswap{d'}{c_2} \circ \newswap{a}{c_2}) \act t_0' \aeq{C} t_1'$, which, through Equivariance (Theorem~\ref{thm:miscellaneous}(\ref{thm:object-equivariance})) combined with Equivalence (Theorem~\ref{thm:miscellaneous}(\ref{thm:alpha-equivalence}))), gives us $\Upsilon_{\catnew{c},\atnew{c_1},\atnew{c_2}} \vdash \newswap{d'}{c_2} \act t_1' \aeq{C} \newswap{a}{c_2} \act t_0'$. Finally, by applying rule $(\frule{\faeq{C}}{ab})$, we conclude that $\Upsilon_{\catnew{c},\atnew{c_1}} \vdash [d']t_1' \aeq{C} [a]t_0'$.

           Using the claim we just proved combined with Equivalence (Theorem~\ref{thm:miscellaneous}(\ref{thm:alpha-equivalence})), the derivations $\Upsilon_{\catnew{c},\atnew{c_1}} \vdash\newswap{a}{c_2}\act[a]t_0'\aeq{C} [d']t_1'$ and $\Upsilon_{\catnew{c},\atnew{c_1}} \vdash \newswap{a}{c_2}\act[d']t_1'\aeq{C} [a]t_0'$ yield
            \[
                \Upsilon_{\catnew{c},\atnew{c_1}} \vdash \newswap{a}{c_2}\act[a]t_0'\aeq{C} [a]t_0' \text{ and } \Upsilon_{\catnew{c},\atnew{c_1}} \vdash \newswap{a}{c_2}\act[d']t_1'\aeq{C} [d']t_1',
            \]
            and the result follows.

            \end{itemize}

        \item For the following cases:
        \begin{itemize}
            \item $t_0\equiv [d']t_0'$ and $t_1 \equiv [a]t_1'$.
            \item $t_0 \equiv [d']t_0'$ and $t_1 \equiv [d']t_1'$
            \item $t_0 \equiv [d_1]t_0'$ and $t_1 \equiv [d_2]t_1'$
            \item $t_0 \equiv [d_2]t_0'$ and $t_1 \equiv [d_1]t_1'$
        \end{itemize}
        The arguments for these cases are very similar to the previous ones, so we will omit them.

            \end{enumerate}
        \end{enumerate}

\end{description}
\end{proof}

\begin{lemma}\label{alemma:characterization-fix-c-general}
    $\Upsilon_{\catnew{c}} \vdash \pi_{\catnew{c}} \fix{C} t$ iff $\Upsilon_{\catnew{c},\atnew{c_1}} \vdash \newswap{a}{\atnew{c_1}}\fix{C} t$ for all $a\in \dom{\pi_{\catnew{c}}}\cup\catnew{c}$ where $\atnew{c_1}\notin \atm{\Upsilon_{\catnew{c}},\pi_{\catnew{c}},t}$.
\end{lemma}

\begin{proof}
    Proof by induction on structure of the term $t$. The non-trivial cases are presented below.
    \begin{itemize}
        \item $t \equiv \pi'\act X$.

        \begin{description}
            \item[$(\Rightarrow)$] In this case, we have $\Upsilon_{\catnew{c}} \vdash \pi_{\catnew{c}} \fix{C} \pi'\act X$. By Inversion (Theorem~\ref{thm:miscellaneous}(\ref{thm:inversion}), it follows that $\pi_{\catnew{c}}^{\pi'^{-1}} \in \PN{}{\Upsilon_{\catnew{c}}|_X}$. In particular, $\pi_{\catnew{c}}^{\pi'^{-1}} \in \PN{}{\Upsilon_{\catnew{c},\atnew{c_1}}|_X}$. Since $\pi_{\catnew{c}}^{\pi'^{-1}} = (\pi^{\pi'^{-1}})_{\pi'^{-1}(\catnew{c})}$ and $\pi'$ does not mention names from $\catnew{c}$, we have $ \pi_{\catnew{c}}^{\pi'^{-1}} = \pnew{\pi^{\pi'^{-1}}}{\catnew{c}}$, obtaining $\pnew{\pi^{\pi^{-1}}}{\catnew{c}} \in  \PN{}{\Upsilon_{\catnew{c},\atnew{c_1}}|_X}$.

            We aim to prove that $\Upsilon_{\catnew{c},\atnew{c_1}} \vdash \newswap{a}{c_1} \fix{C} \pi'\act X$ for all $a \in \dom{\pi_{\catnew{c}}} \cup \catnew{c}$. To achieve this, it suffices to show
            \[
                \newswap{\pi'^{-1}(a)}{c_1} = \newswap{a}{c_1}^{\pi'^{-1}} \in \PN{}{\Upsilon_{\catnew{c},\atnew{c_1}}|_X}.
            \]

            \begin{itemize}
                \item First, let's analyse the cases where $a$ is an atom from $\catnew{c}$. Notice that $\catnew{c},\atnew{c_1} \subseteq \atm{(\Upsilon_{\catnew{c},\atnew{c_1}}|_X)_{\fresh}}$. Thus, for every $\atnew{c'} \in \catnew{c}$, we have $(\atnew{c'} \ \atnew{c_1}) \in \PN{}{\Upsilon_{\catnew{c},\atnew{c_1}}|_X}$, implying $(\atnew{c'} \ \atnew{c_1}) \in \PN{}{\Upsilon_{\catnew{c},\atnew{c_1}}|_X}$. Since $\pi'^{-1}(\atnew{c'}) = \atnew{c'}$ for all $\atnew{c'} \in \catnew{c}$, this ensures that $\newswap{\pi'^{-1}(a)}{c_1} \in  \PN{}{\Upsilon_{\catnew{c},\atnew{c_1}}|_X}$ for all $a \in \catnew{c}$.

                \item  Now, suppose $a \in \dom{\pi_{\catnew{c}}}$ but $a \notin \catnew{c}$. Then $\pi'^{-1}(a) \in \dom{\pnew{\pi^{\pi'^{-1}}}{\catnew{c}}}$. From $\pnew{\pi^{\pi^{-1}}}{\catnew{c}} \in  \PN{}{\Upsilon_{\catnew{c},\atnew{c_1}}|_X}$, we get $\pnew{\pi^{\pi^{-1}}}{\catnew{c}} \in \Perm{\atm{(\Upsilon_{\catnew{c},\atnew{c_1}}|_X)_{\fresh}}}$. Thus $\pi'^{-1}(a) \in \atm{(\Upsilon_{\catnew{c},\atnew{c_1}}|_X)_{\fresh}}$ and so $\newswap{a}{c_1}^{\pi'^{-1}} \in \PN{}{\Upsilon_{\catnew{c},\atnew{c_1}}|_X}$. Consequently, the result follows by applying rule $(\frule{\faeq{C}}{var})$.
            \end{itemize}

            \item[$(\Leftarrow)$] Suppose $\Upsilon_{\catnew{c},\atnew{c_1}} \vdash \newswap{a}{c_1}\fix{C} \pi'\act X$ for all $a \in \dom{\pi_{\catnew{c}}} \cup \catnew{c}$, where $\atnew{c_1} \notin \atm{\Upsilon_{\catnew{c}},\pi_{\catnew{c}},\pi'}$. By Inversion (Theorem~\ref{thm:miscellaneous}(\ref{thm:inversion}), it follows that
            \[
                \newswap{\pi'^{-1}(a)}{c_1} = \newswap{a}{c_1}^{\pi'^{-1}} \in \PN{}{\Upsilon_{\catnew{c},\atnew{c_1}}|_X}.
            \]
            Since $\atnew{c_1} \notin \dom{\pi'}$, we have $\pi'^{-1}(a) \neq \atnew{c_1}$. Moreover, $\newswap{\pi'^{-1}(a)}{c_1}\in \Perm{\atm{(\Upsilon_{\catnew{c},\atnew{c_1}})_{\fresh}}}$, which implies $\pi'^{-1}(a) \in \atm{(\Upsilon_{\catnew{c},\atnew{c_1}})_{\fresh}}$.

            Equivalently, this means that $a \in \pi'\act\atm{(\Upsilon_{\catnew{c},\atnew{c_1}})_{\fresh}}$ for all $a \in \dom{\pi_{\catnew{c}}} \cup \catnew{c}$. As a result, we have $\pi_{\catnew{c}} \in \Perm{\pi'\act\atm{(\Upsilon_{\catnew{c},\atnew{c_1}})_{\fresh}}}$. Consequently,
            \[
                \pi_{\catnew{c}}^{\pi'^{-1}} \in \Perm{\atm{(\Upsilon_{\catnew{c},\atnew{c_1}})_{\fresh}}} \subseteq \PN{}{\Upsilon_{\catnew{c}}|_X}.
            \]
            Finally, the result follows by an application of rule $(\frule{\faeq{C}}{var})$.
        \end{description}

        \item $t \equiv \tf{f^C}(t_0,t_1)$.

        \begin{description}
            \item[$(\Rightarrow)$] In this case, we have $\Upsilon_{\catnew{c}} \vdash \pi_{\catnew{c}} \fix{C} \tf{f^C}(t_0,t_1)$ which is the same as $\Upsilon_{\catnew{c}} \vdash\tf{f^C}( \pi_{\catnew{c}}\act t_0, \pi_{\catnew{c}}\act t_1) \aeq{C} \tf{f^C}(t_0,t_1)$. By Inversion (Theorem~\ref{thm:miscellaneous}(\ref{thm:inversion}),

            \begin{itemize}
                \item either $\Upsilon_{\catnew{c}} \vdash \pi_{\catnew{c}}\act t_0 \aeq{C} t_0$ and $\Upsilon_{\catnew{c}} \vdash \pi_{\catnew{c}}\act t_1 \aeq{C} t_1$.

                Then by induction, it follows that $\Upsilon_{\catnew{c},\atnew{c_1}} \vdash \newswap{a}{c_1}\act t_0\aeq{C} t_0$ and $\Upsilon_{\catnew{c},\atnew{c_1}} \vdash \newswap{a}{c_1}\act t_1\aeq{C} t_1$ for all $a \in \dom{\pi_{\catnew{c}}} \cup \catnew{c}$, where $\atnew{c_1} \notin \atm{\Upsilon_{\catnew{c}},\pi_{\catnew{c}}, t}$. By rule $(\frule{\faeq{C}}{\tf{f^C}})$, we obtain $\Upsilon_{\catnew{c},\atnew{c_1}} \vdash \newswap{a}{c_1}\fix{C} \tf{f^C}(t_0,t_1) $, for all $a \in \dom{\pi_{\catnew{c}}} \cup \catnew{c}$.

                \item or $\Upsilon_{\catnew{c}} \vdash \pi_{\catnew{c}}\act t_0 \aeq{C} t_1$ and $\Upsilon_{\catnew{c}} \vdash \pi_{\catnew{c}}\act t_0 \aeq{C} t_1$.

                By Equivariance and Equivalence of Theorem~\ref{thm:miscellaneous}, we have $\Upsilon_{\catnew{c}} \vdash \pi_{\catnew{c}}^2\act t_0 \aeq{C} t_0$ and  $\Upsilon_{\catnew{c}} \vdash \pi_{\catnew{c}}^2\act t_1 \aeq{C} t_1$. Proceeding similarly to the analysis in the previous case, we conclude that $\Upsilon_{\catnew{c},\atnew{c_1}} \vdash\newswap{a}{c_1}\fix{C} \tf{f^C}(t_0,t_1)$ for all $a\in \dom{\pi_{\catnew{c}}^2}\cup\catnew{c}$.

                Since $\dom{\pi_{\catnew{c}}^2} \subseteq \dom{\pi_{\catnew{c}}}$, it remains to show that $\Upsilon_{\catnew{c},\atnew{c_1}} \vdash \newswap{a}{c_1}\fix{C} \tf{f^C}(t_0,t_1)$ for all $a \in \dom{\pi_{\catnew{c}}} \setminus (\dom{\pi_{\catnew{c}}^2} \cup \catnew{c})$. Assume, by contradiction, that there exists some $a \in \dom{\pi_{\catnew{c}}} \setminus (\dom{\pi_{\catnew{c}}^2} \cup \catnew{c})$ such that $\Upsilon_{\catnew{c},\atnew{c_1}} \vdash \newswap{a}{c_1}\act \tf{f^C}(t_0,t_1)\naeq{C} \tf{f^C}(t_0,t_1)$. Observe that, since $a \notin \dom{\pi_{\catnew{c}}^2}$, it follows that $\pi_{\catnew{c}}^2(a) = a$. Therefore, $\pi_{\catnew{c}}$ can be expressed as $(a \ \pi_{\catnew{c}}(a)) \circ \rho$, where $\rho$ is a  disjoint from $(a \ \pi_{\catnew{c}}(a))$. We are assuming $a \notin \catnew{c}$, so we analyse $\pi_{\catnew{c}}(a)$:
                \begin{itemize}
                    \item If $\pi_{\catnew{c}}(a) \notin \catnew{c}$, then $\pi_{\catnew{c}}$ would contain the disjoint cycle $(a \ \pi_{\catnew{c}}(a))$ that does not mention any names from $\catnew{c}$, which contradicts the structure of $\pi_{\catnew{c}}$.

                    \item Hence, $\pi_{\catnew{c}}(a)$ must be an atom from $\catnew{c}$, say $\pi_{\catnew{c}}(a) \in \catnew{c}$.
                \end{itemize}
                Since we proved that $\Upsilon_{\catnew{c},\atnew{c_1}} \vdash \newswap{a}{c_1}\fix{C} \tf{f^C}(t_0,t_1)$ for all $a\in \dom{\pi_{\catnew{c}}^2}\cup\catnew{c}$.  In particular it holds for all $a\in \catnew{c}$. Then $ \Upsilon_{\catnew{c},\atnew{c_1}} \vdash (\pi_{\catnew{c}}(a) \ \atnew{c_1})\fix{C}\tf{f^C}(t_0,t_1)$, which is the same as $\Upsilon_{\catnew{c},\atnew{c_1}} \vdash (a \ \atnew{c_1})^{\pi_{\catnew{c}}}\fix{C}\tf{f^C}(t_0,t_1)$. By Equivariance (Theorem~\ref{thm:miscellaneous}(\ref{thm:object-equivariance})), this implies $\Upsilon_{\catnew{c},\atnew{c_1}} \vdash (a \ \atnew{c_1})\fix{C}\pi_{\catnew{c}}\act\tf{f^C}(t_0,t_1)$ and so the result follows by Equivalence of Theorem~\ref{thm:miscellaneous}.
            \end{itemize}

            \item[$(\Leftarrow)$]  Suppose $\Upsilon_{\catnew{c},\atnew{c_1}} \vdash \newswap{a}{c_1}\fix{C} \tf{f^C}(t_0,t_1)$ for all $a\in \dom{\pi_{\catnew{c}}}\cup\catnew{c}$ where $\atnew{c_1}\notin \atm{\Upsilon_{\catnew{c}},\pi_{\catnew{c}},t}$. By Lemma~\ref{alemma:characterization-fix-commutative}, we have $\Upsilon_{\catnew{c},\atnew{c_1}} \vdash \newswap{a}{c_1}\fix{C} t_0$ and $\Upsilon_{\catnew{c},\atnew{c_1}} \vdash \newswap{a}{c_1}\fix{C} t_1$ for all $a\in \dom{\pi_{\catnew{c}}}\cup\catnew{c}$. By induction, we obtain $\Upsilon_{\catnew{c}} \vdash \pi_{\catnew{c}} \fix{C} t_0$ and $\Upsilon_{\catnew{c}} \vdash \pi_{\catnew{c}} \fix{C} t_1$ and thus by rule $(\frule{\faeq{C}}{\tf{f^C}})$ we get $\Upsilon_{\catnew{c}} \vdash \pi_{\catnew{c}}\fix{C} \tf{f^C}(t_0,t_1)$.
        \end{description}
    \end{itemize}
\end{proof}

\begin{lemma}\label{alemma:fix-point-split}
    $\Upsilon_{\catnew{c}} \vdash \pi \fix{C} t$ iff $\Upsilon_{\catnew{c}} \vdash \pi_{\catnew{c}} \fix{C} t$ and $\Upsilon_{\catnew{c}} \vdash \pi_{\neg\catnew{c}} \fix{C} t$.
\end{lemma}

\begin{proof}
    The left-to-right case follows by Lemma~\ref{alemma:fix-point-composition-inverse}(\ref{alemma:fix-point-composition}). For the right-to-left case, assume that  $\Upsilon_{\catnew{c}} \vdash \pi \fix{C} t$. By Soundness (Theorem~\ref{thm:soundness-completeness-fix}(\ref{thm:soundness-fix})), we have $\Upsilon_{\catnew{c}} \vDash \pi \fix{C} t$. As a consequence of Pitts' equivalence (Lemma~\ref{lemma:pitts-eq-generalized}), we have $\Upsilon_{\catnew{c}} \vDash \pi_{\catnew{c}} \fix{C} t$ and $\Upsilon_{\catnew{c}} \vDash \pi_{\neg\catnew{c}} \fix{C} t$. Therefore, the result follows by Completeness (Theorem~\ref{thm:soundness-completeness-fix}(\ref{thm:completeness-fix})).
\end{proof}

 \begin{lemma}\label{alemma:perm_dec}
     Let $I$ be a non-empty, finite set of indices. Suppose, for each $i\in I$, that $\Upsilon_{\catnew{c}} \vdash \pi_i\fix{C}t$. Let $\pi$ be a permutation and $\atnew{\pvec{c}'}\subseteq \catnew{c}$ such that
     \begin{equation}
         \pi \in \Perm{\left(\bigcup_{i\in I} \dom{\pnew{\pi_i}{\atnew{\pvec{c}'}}}\right)\cup \atnew{\pvec{c}'}}\alert{\circ}\pair{\{\npnew{\pi_i}{\atnew{\pvec{c}'}} \mid i\in I\}}.
     \end{equation}
     Then $\Upsilon_{\catnew{c}} \vdash \pi\fix{C}t$.
\end{lemma}

 \begin{proof}
     By Lemma~\ref{alemma:fix-point-split}, we have: $\Upsilon_{\catnew{c}}\vdash \pnew{\pi_i}{\atnew{\pvec{c}'}}\fix{C} t$ and $\Upsilon_{\catnew{c}}\vdash \npnew{\pi_i}{\atnew{\pvec{c}'}}\fix{C} t$, for all $i\in I$. Write $\pi$ as $\pi = \eta_1\circ \eta_2$ where $\eta_1 \in \Perm{\left(\bigcup_{i\in I} \dom{\pnew{\pi_i}{\atnew{\pvec{c}'}}}\right)\cup \atnew{\pvec{c}'}}$ and $\eta_2 \in\pair{\{\npnew{\pi_i}{\atnew{\pvec{c}'}} \mid i\in I\}}$. Thus, on one hand the former combined with Lemmas~\ref{alemma:characterization-fix-c-general} and~\ref{alemma:fix-point-formed-by-fresh-names} yields $\Upsilon_{\catnew{c}} \vdash \eta_1 \fix{C} t$. On the other hand, the latter combined with Lemma~\ref{alemma:fix-point-composition} yields $\Upsilon_{\catnew{c}}\vdash \eta_2\fix{C} t$. Then the result follows by Lemma~\ref{alemma:fix-point-composition}(\ref{alemma:fix-point-composition}).

 \end{proof}

Given contexts $\Upsilon_{\catnew{c}}$ and $\Upsilon'_{\atnew{\pvec{c}'}}$ and $\sigma_{\atnew{\pvec{c}''}}$ a $\new$-substitution. We introduce some notations:
\begin{itemize}
    \item $\Upsilon_{\catnew{c}}\vdash \Upsilon'_{\atnew{\pvec{c}'}}\sigma_{\atnew{\pvec{c}''}}$ for $\Upsilon_{\catnew{c}\cup\atnew{\pvec{c}'}\cup\atnew{\pvec{c}''}}\vdash \Upsilon'\sigma$.
    \item $\Upsilon_{\catnew{c}}\vdash s\sigma_{\atnew{\pvec{c}''}} \aeq{C} t\sigma_{\atnew{\pvec{c}''}}$ for $\Upsilon_{\catnew{c}\cup\atnew{\pvec{c}''}}\vdash s\sigma \aeq{C} t\sigma$.

    \item $\Upsilon_{\catnew{c}}\vdash \pi\fix{C} t\sigma_{\atnew{\pvec{c}''}}$ for $\Upsilon_{\catnew{c}\cup\atnew{\pvec{c}''}}\vdash \pi \fix{C} t\sigma$.
\end{itemize}

 \begin{lemma}[Preservation by substitution]\label{alemma:preservation-by-subs}
     Suppose $\Upsilon_{\catnew{c}}$ and $\Upsilon'_{\atnew{\pvec{c}'}}$ are contexts and let $\sigma_{\atnew{\pvec{c}''}}$ be a substitution such that $\Upsilon_{\catnew{c}}\vdash \Upsilon'_{\atnew{\pvec{c}'}}\sigma_{\atnew{\pvec{c}''}}$. If $\Upsilon'_{\atnew{\pvec{c}'}}\vdash s\aeq{C}t$ then $\Upsilon_{\catnew{c}\cup\atnew{\pvec{c}'}}\vdash s\sigma_{\atnew{\pvec{c}''}}\aeq{C}t\sigma_{\atnew{\pvec{c}''}}$.
 \end{lemma}

 \begin{proof}
    Proof by induction on the last rule applied. The only non-trivial case is the rule $(\frule{\faeq{C}}{var})$. In this case, $s \equiv \pi_1\act X$ and $t \equiv \pi_2\act X$. Hence, $\Upsilon'_{\atnew{\pvec{c}'}}\vdash \pi_1\act X\aeq{C} \pi_2\act X$. By Inversion (Theorem~\ref{thm:miscellaneous}(\ref{thm:inversion})), we have $\pi_2^{-1}\circ\pi_1\in \PN{}{\Upsilon'_{\atnew{\pvec{c}'}}|_X} = \Perm{\atm{(\Upsilon_{\atnew{\pvec{c}'}}|_X)_{\fresh}}}\alert{\circ}\pair{\perm{}{(\Upsilon_{\atnew{\pvec{c}'}}|_X)_{\fix{C}}}}$. By hypothesis, $\Upsilon_{\catnew{c}\cup\atnew{\pvec{c}'}\cup\atnew{\pvec{c}''}}\vdash \pi\fix{C} Y\sigma$ for every $\pi\fix{C} Y\in \Upsilon'$. In particular, $\Upsilon_{\catnew{c}\cup\atnew{\pvec{c}'}\cup\atnew{\pvec{c}''}} \vdash \pi\fix{C} X\sigma$ for every $\pi\fix{C} X\in \Upsilon'|_X$. By Lemma~\ref{alemma:perm_dec}, we obtain $\Upsilon_{\catnew{c}\cup\atnew{\pvec{c}'}\cup\atnew{\pvec{c}''}} \vdash (\pi_2^{-1}\circ\pi_1)\fix{C} X\sigma$, which is the same as  $\Upsilon_{\catnew{c}\cup\atnew{\pvec{c}'}\cup\atnew{\pvec{c}''}} \vdash (\pi_2^{-1}\circ\pi_1)\act X\sigma \aeq{C} X\sigma$. By Equivariance (Theorem~\ref{thm:miscellaneous}(\ref{thm:object-equivariance})), this is equivalent to $\Upsilon_{\catnew{c}\cup\atnew{\pvec{c}'}\cup\atnew{\pvec{c}''}} \vdash \pi_1\act X\sigma \aeq{C} \pi_2\act X\sigma$ Therefore, $\Upsilon_{\catnew{c}\cup\atnew{\pvec{c}'}} \vdash \pi\act X\sigma_{\atnew{\pvec{c}''}} \aeq{C} \pi_2\act X\sigma_{\atnew{\pvec{c}''}}$.
\end{proof}

\section{Proofs of Section~\ref{sec:nominal-c-unification} - Nominal \texorpdfstring{$\C$}{C}-Unification}\label{app:nominal-c-unification}

For the rest of this section, for a solvable problem $\probc = \newc{c}{}.Pr$ and a solution $\npair{\Psi,\sigma}{\atnew{\pvec{c}'}}\in \U{\probc}$, we will write $\Psi_{\atnew{\pvec{c}'}} \vdash (Pr)\sigma$ to abbreviate that $\Psi_{\atnew{\pvec{c}'}} \vdash s\sigma\aeq{C} t\sigma$ for all $s\aeq{C}^? t\in Pr$.

\begin{lemma}\label{alemma:quasiorder}
    The relation $\ins{}$ defines a quasiorder (i.e. reflexive and transitive) in $\mathcal{U}(\probc)$.
\end{lemma}

\begin{proof}
\begin{itemize}
    \item \emph{Reflexivity.}  Just note that for every pair $\npair{\Psi,\sigma}{\catnew{c}}\in \mathcal{U}(\probc)$, we have $\npair{\Psi,\sigma}{\catnew{c}} \ins{\id} \npair{\Psi,\sigma}{\catnew{c}}$.

    \item \emph{Transitivity.} Suppose $ \npair{\Psi_1,\sigma_1}{\atnew{\vec{c}_1}} ,\npair{\Psi_2,\sigma_2}{\atnew{\vec{c}_2}},\npair{\Psi_3,\sigma_3}{\atnew{\vec{c}_3}}\in \mathcal{U}(\probc)$ are such that $\npair{\Psi_1,\sigma_1}{\atnew{\vec{c}_1}} \ins{\delta_1} \npair{\Psi_2,\sigma_2}{\atnew{\vec{c}_2}}$ and $\npair{\Psi_2,\sigma_2}{\atnew{\vec{c}_2}} \ins{\delta_2} \npair{\Psi_3,\sigma_3}{\atnew{\vec{c}_3}}$. Then, by definition, we have $\atnew{\vec{c}_1} \subseteq \atnew{\vec{c}_2} \subseteq \atnew{\vec{c}_3}$ and
          \begin{enumerate}
              \item $(\Psi_2)_{\atnew{\vec{c}_2}}\vdash \Psi_1\delta_1$  and  $(\Psi_2)_{\atnew{\pvec{c}_2}}\vdash X\sigma_2\aeq{C}X\sigma_1\delta_1$ for all $X\in\V$.
              \item  $(\Psi_3)_{\atnew{\vec{c}_3}}\vdash \Psi_2\delta_2$  and  $(\Psi_3)_{\atnew{\vec{c}_3}}\vdash X\sigma_3\aeq{C}X\sigma_2\delta_2$  for all $X\in\V$.
          \end{enumerate}
          Take $\delta = \delta_1\delta_2$. We know $(\Psi_3)_{\atnew{\vec{c}_3}}\vdash \Psi_2\delta_2$. Moreover, we have $(\Psi_2)_{\atnew{\vec{c}_2}}\vdash \Psi_1\delta_1$. Applying Lemma~\ref{alemma:preservation-by-subs} yields $(\Psi_3)_{\atnew{\vec{c}_3}}\vdash \Psi_1\delta_1\delta_2$ because $\atnew{\vec{c}_2}\subseteq\atnew{\vec{c}_3}$.

          For all $X\in \V$, we have that $(\Psi_2)_{\atnew{\vec{c}_2}}\vdash X\sigma_2\aeq{C}X\sigma_1\delta_1$. By $(\Psi_3)_{\atnew{\vec{c}_3}}\vdash \Psi_2\delta_2$, Lemma~\ref{alemma:preservation-by-subs} and $\atnew{\vec{c}_2}\subseteq\atnew{\vec{c}_3}$, we obtain  $(\Psi_3)_{\atnew{\vec{c}_3}}\vdash X\sigma_2\delta_2\aeq{C}X\sigma_1\delta_1\delta_2$ for all  $X\in\V$. Moreover, we know that $(\Psi_3)_{\atnew{\vec{c}_3}}\vdash X\sigma_3\aeq{C}X\sigma_2\delta_2$ for all $X\in\V$. Therefore, by transitivity of $\aeq{C}$, it follows that  $(\Psi_3)_{\atnew{\vec{c}_3}}\vdash X\sigma_3\aeq{C}X\sigma_1\delta_1\delta_2$ for all $X\in \V$.
\end{itemize}

\end{proof}

\begin{lemma}[Closure by Instantiation]\label{alemma:closure-by-instantiation}
    Suppose $\npair{\Psi_1,\sigma_1}{\atnew{\vec{c}_1}}\in \mathcal{U}(\probc)$ and $\npair{\Psi_1,\sigma_1}{\atnew{\vec{c}_1}} \ins{} \npair{\Psi_2,\sigma_2}{\atnew{\vec{c}_2}}$. Then $ \npair{\Psi_2,\sigma_2}{\atnew{\vec{c}_2}} \in \mathcal{U}(\probc)$.
\end{lemma}

\begin{proof}
     Suppose $\probc = \newc{c}{}.Pr$. Let $\npair{\Psi_1,\sigma_1}{\atnew{\vec{c}_1}}\in \mathcal{U}(\probc)$. Then by definition, we have $(\Psi_1)_{\atnew{\vec{c}_1}} \vdash (Pr)\sigma_1$ and $(\Psi_1)_{\atnew{\pvec{c}_1}} \vdash X\sigma_1 \aeq{C} X\sigma_1\sigma_1$ for all $X\in \V$. From  $\npair{\Psi_1,\sigma_1}{\atnew{\vec{c}_1}} \ins{} \npair{\Psi_2,\sigma_2}{\atnew{\vec{c}_2}}$, we know that there exist a substitution $\delta$ such that:

     \begin{enumerate}
         \item $(\Psi_2)_{\atnew{\vec{c}_2}}\vdash \Psi_1\delta$.

         \item $(\Psi_2)_{\atnew{\vec{c}_2}}\vdash X\sigma_2\aeq{C}X\sigma_1\delta$, for all $X\in \V$;
     \end{enumerate}

     \begin{claim}[1]
          We claim that $(\Psi_2)_{\atnew{\vec{c}_2}}\vdash t\sigma_2\aeq{C}t\sigma_1\delta$ for all nominal term $t$. The proof proceeds by straightforward induction on the structure of the nominal term $t$ utilizing item 2 as a key component of the argument.
     \end{claim}

     \begin{claim}[2]
          We claim that $(\Psi_2)_{\atnew{\vec{c}_2}} \vdash (Pr)\sigma_2$. For each $s\aeq{C}^?  t\in Pr$, we have $(\Psi_1)_{\atnew{\vec{c}_1}} \vdash s\sigma_1 \aeq{C} t\sigma_1$. By Lemma~\ref{alemma:preservation-by-subs}, it follows that $(\Psi_2)_{\atnew{\vec{c}_2}} \vdash s\sigma_1\delta \aeq{C} t\sigma_1\delta$. By Claim 1 and transitivity of $\aeq{C}$, it follows that $(\Psi_2)_{\atnew{\vec{c}_2}} \vdash s\sigma_2 \aeq{C} t\sigma_2$. 
     \end{claim}
\end{proof}

For the  termination proof, we define the \emph{size of a nominal term} $t$, denoted by $|t|$, inductively by: $|a| := 1, |\pi\act X| := 1, |\tf{f}(\tilde{t})_n| := 1+|t_1|+\ldots+|t_n|$ for $\tf{f}\in\Sigma$, and $|[a]t| := 1+|t|$. 
The \emph{size of a constraint} is defined by $|s\aeq{C}^? t| := |s|+|t|$. As a consequence, for all term $t$ and permutation $\pi$, we have that $|\pi\act t| = |t|$.

\termination*

\begin{proof}
 Define a measure for the size of a problem $\probc = \newc{c}{}.Pr$ as $[\probc] = (n,M)$, where $n$ is the number of distinct variables occurring in the constraints of $Pr$ and $M$ is the multiset of the sizes of the equality constraints that are not of the form $\pi\act X\aeq{C}^? X$ that are occurring in $Pr$. To compare $[\probc]$ and $[\probc']$ we use the lexicographic order $>_{lex}$ that is defined upon the product of the usual order on natural numbers, $>$, and its multiset extension $>_{mul}$, that is,  $>_{lex} \; \triangleq \; > \times >_{mul}$. 
\end{proof}

\begin{lemma}\label{alemma:preservation-solution}
     Let $\probc$ be a problem such that $\probc\overset{*}{\Longrightarrow} \exprob$ without using instantiating rules. Assume $\probc \overset{*}{\Longrightarrow} \probc'$ where $\probc'\in \exprob$. Then
     \begin{enumerate}
         \item  $\U{\probc}=\U{\probc'}$. 

         \item If $\probc'$ contains inconsistent constraints, then $\U{\probc} = \emptyset$.
     \end{enumerate}
 \end{lemma}

\begin{proof}
    For the first item the proof follows by induction on the length of the derivation $\probc\overset{*}{\Longrightarrow} \probc'$.

    \paragraph{(Base).} The base of the induction corresponds to the case where $\probc\overset{0}{\Longrightarrow} \probc'$, which implies that $\probc = \probc'$ and hence  $\U{\probc} = \U{\probc'}$.

    \paragraph{(Inductive step).} Now, we are going to prove that the result holds for reductions of length $n>0$, that is, $\probc\overset{n}{\Longrightarrow} \probc'$. Our induction hypothesis is
           \begin{equation}\label{eq:induction-hyp-1}
                \text{If ${\probc}_1\overset{n-1}{\Longrightarrow} {\probc}_2$, then $\U{{\probc}_1}=\U{{\probc}_2}$ }\tag{I.H.}
            \end{equation}
            Note that we can rewrite  $\probc\overset{n}{\Longrightarrow} \probc'$ as  $\probc\overset{n-1}{\Longrightarrow} \probc''\Longrightarrow \probc'$. By (\ref{eq:induction-hyp-1}), it follows that $\U{\probc} = \U{\probc''}$. It remains to analyse the last step of the reduction, specifically $\probc''\Longrightarrow \probc'$. We present a illustrative case: The last rule is $(var)$. In this case, $\probc'' = \newc{c}{}.Pr_0\uplus\{\pi\act X \aeq{C}^? \pi'\act X\} \Longrightarrow \newc{c}{}.Pr_0 \cup\{\pi'^{-1}\circ\pi\act X\aeq{C}^? X\} = \probc'$, where $\pi'\neq \id$. If $\npair{\Psi,\sigma}{\atnew{\pvec{c}'}}\in\U{\probc''}$, then
                $\Psi_{\atnew{\pvec{c}'}} \vdash (\pi\act X)\sigma \aeq{C} \pi'\act X\sigma$. Since $(\rho\act X)\sigma \equiv \rho\act X\sigma$, it follows that $\Psi_{\atnew{\pvec{c}'}} \vdash \pi\act (X\sigma) \aeq{C} \pi'\act (X\sigma)$. By Equivariance (Theorem~\ref{thm:miscellaneous}(\ref{thm:object-equivariance})), we have $\Psi_{\atnew{\pvec{c}'}} \vdash (\pi'^{-1}\circ\pi)\act (X\sigma) \aeq{C} X\sigma$. Therefore, $\npair{\Psi,\sigma}{\atnew{\pvec{c}'}}\in\U{\prob'}$ and hence $\U{\probc''} \subseteq \U{\probc'}$. For the other inclusion the argument is the same.

        The second item of the lemma follows directly from the fact that inconsistent constraints are not derivable, therefore have no solutions.
 \end{proof}

\correctness*

\begin{proof}
  By induction on the number of steps in the reduction $\probc \overset{*}{\Longrightarrow} \nf{\probc}$.

       \paragraph{(Base).}  This corresponds to the case where $\probc\overset{0}{\Longrightarrow} \nf{\probc}$. In this case, $\nf{\probc} = \probc = \newc{c}{}.Pr$ and so $Pr$ contains only consistent constraints. Then $Pr = \{\pi_i\act X_i\aeq{C}^? X_i \mid i=1,\ldots,n\}$. By construction, the pair $\sol{\probc} = \{\npair{\Psi,Id}{\catnew{c}}\}$, where $\Psi_{\catnew{c}} = \newc{c}{}. \{\pi_i\fix{C} X_i \mid i=1,\ldots,n\}$, satisfies $\Psi_{\catnew{c}} \vdash \pi_i\act X_iId \aeq{C} X_iId$  for all $i=1,\ldots,n$. Therefore, $\sol{\probc} \subseteq \U{\probc}$.

           For the second item, for any other $\npair{\Phi,\tau}{\atnew{\pvec{c}''}} \in \U{\probc}$, by definition, $\tau$ is such that $\Phi_{\atnew{\pvec{c}''}} \vdash (Pr)\tau$ which is the same as $\Phi_{\atnew{\pvec{c}''}} \vdash \Psi\tau$. Moreover,  $\Phi_{\atnew{\pvec{c}''}} \vdash X\tau \aeq{C} XId\tau$ follows for all $X\in \V$ by reflexivity. Therefore, this proves that $\npair{\Psi,Id}{\catnew{c}} \ins{} \npair{\Phi,\tau}{\atnew{\pvec{c}''}}$ and the second item follows.

        \paragraph*{(Inductive step).}  Now, we are going to prove that the result holds for reductions of length $n>0$, that is, $\probc\overset{n}{\Longrightarrow} \nf{\probc}$. We rewrite it as follows: $\probc \Longrightarrow \probc''\overset{n-1}{\Longrightarrow} \nf{\probc}$.

           \begin{itemize}
                \item Suppose $\probc \Longrightarrow \probc''$ by some non-instantiating simplification. Then using Lemma~\ref{alemma:preservation-solution}, we know that $\U{\probc} = \U{\probc''}$.

                By induction, we have $\sol{\probc''} \subseteq \U{\probc''}$. Consequently, $\sol{\probc''} \subseteq \U{\probc}$. Since no instantiation rule was used, by the construction of $\sol{-}$ we have that $\sol{\probc} = \sol{\probc''}$ and the first item follows.

                For the second item, for any other $\npair{\Phi,\tau}{\atnew{\pvec{c}''}} \in \U{\probc}$, by induction, there is some $\npair{\Psi,\sigma}{\atnew{\pvec{c}'}} \in \sol{\probc''}$ such that $\npair{\Psi,\sigma}{\atnew{\pvec{c}'}} \ins{}\npair{\Phi,\tau}{\atnew{\pvec{c}''}}$ and the result follows because $\sol{\probc} = \sol{\probc''}$.

                \item Suppose $\probc \overset{\theta}{\Longrightarrow} \probc''$  by an instantiating rule, say $(inst_1)$. So $\probc = \newc{c}{}.Pr\uplus\{\pi\act X \aeq{C}^? t\} \overset{\theta}{\Longrightarrow} \newc{c}{}.Pr\theta = \probc''$ where $\theta = [X\mapsto \pi^{-1}\act t]$ and $X\notin \var{t}$. 
                
                Let $\npair{\Psi,\sigma}{\atnew{\pvec{c}'}} \in \sol{\probc}$, so by construction $\sigma = \theta\sigma'$ and $\npair{\Psi,\sigma'}{\atnew{\pvec{c}'}}\in \sol{\probc''}$. By induction, $\sol{\probc''} \subseteq \U{\probc''}$ which implies $\Psi_{\atnew{\pvec{c}'}} \vdash (Pr\theta)\sigma'$, that is the same as  $\Psi_{\atnew{\pvec{c}'}} \vdash (Pr)(\theta\sigma')$ and thus $\npair{\Psi,\sigma}{\atnew{\pvec{c}'}} \in \U{\probc}$ which proves the first item.

                It remains to prove that $\sol{\probc}$ is a complete set of solutions of $\probc$. So, take $\npair{\Phi,\tau}{\atnew{\pvec{c}''}} \in \U{\probc}$. By definition, $\Phi_{\atnew{\pvec{c}''}} \vdash (Pr\uplus\{\pi\act X\aeq{C}^? t\})\tau$. In particular, $\Phi'_{\atnew{\pvec{c}''}} \vdash (\pi\act X)\tau \aeq{C} t\tau$ which, by Equivariance (Theorem~\ref{thm:miscellaneous}(\ref{thm:object-equivariance})) and the fact that permutations and substitutions commute, results in $\Phi_{\atnew{\pvec{c}''}} \vdash X\tau \aeq{C} (\pi^{-1}\act t)\tau$ that is the same as $\Phi_{\atnew{\pvec{c}''}} \vdash X\tau \aeq{C} X\theta\tau$.
        
                Let $\tau'$ be the substitution that acts just like $\tau$, only $X\tau \equiv X$ (remember that $X\tau \not\equiv X$). Then since $X\notin \var{X\theta}$, we have $X\theta\tau \equiv X\theta\tau'$ and hence $\Phi_{\atnew{\pvec{c}''}} \vdash X\tau \aeq{C} X\theta\tau'$. For all other variable $Y$, we have $\Phi_{\atnew{\pvec{c}''}} \vdash Y\tau \aeq{C} Y\tau'$. Since $Y\theta \equiv Y$, we conclude that  $\Phi_{\atnew{\pvec{c}''}} \vdash Z\tau \aeq{C} Z\theta\tau'$ holds for all variable $Z\in \V$. Then $\Phi_{\atnew{\pvec{c}''}} \vdash s\tau \aeq{C} s\theta\tau'$ for all nominal term $s$. As a consequence, because $\Phi_{\atnew{\pvec{c}''}} \vdash (Pr)\tau$ we conclude $\Phi_{\atnew{\pvec{c}''}} \vdash (Pr)(\theta\tau')$. This is the same as $\Phi_{\atnew{\pvec{c}''}} \vdash (Pr\theta)\tau'$. So $\npair{\Phi,\tau'}{\atnew{\pvec{c}''}}\in \U{\probc''}$ and by inductive hypothesis there is some $\npair{\Psi,\sigma}{\atnew{\pvec{c}'}}\in\sol{\probc''}$ such that  $\npair{\Psi,\sigma}{\atnew{\pvec{c}'}}\ins{} \npair{\Phi,\tau'}{\atnew{\pvec{c}''}}$. By construction, $\npair{\Psi,\theta\sigma}{\atnew{\pvec{c}'}}\in\sol{\probc}$.

                \begin{claim}[1]
                    \sloppy{We claim that $\npair{\Psi,\theta\sigma}{\atnew{\pvec{c}'}} \ins{} \npair{\Phi,\theta\tau'}{\atnew{\pvec{c}''}}$. Since $\npair{\Psi,\sigma}{\atnew{\pvec{c}'}}\ins{} \npair{\Phi,\tau'}{\atnew{\pvec{c}''}}$ we know that there is some $\delta$ such that $\Phi_{\atnew{\pvec{c}''}} \vdash \Psi\delta$ and $\Phi_{\atnew{\pvec{c}''}} \vdash Z\tau' \aeq{C} Z\sigma\delta$ for all $Z\in \V$. Then, $\Phi_{\atnew{\pvec{c}''}} \vdash u\tau' \aeq{C} u\sigma\delta$ holds for all terms $u$, in particular, for every $u\equiv X\theta$ where $X\in V$. So $\Phi_{\atnew{\pvec{c}''}} \vdash X\theta\tau' \aeq{C} X\theta\sigma\delta$ for all $X\in \V$, proving the claim.}
                \end{claim}

                \begin{claim}[2]
                    We claim that $\npair{\Phi,\theta\tau'}{\atnew{\pvec{c}''}} \ins{} \npair{\Phi,\tau}{\atnew{\pvec{c}''}}$. In fact, take $\delta = Id$. Then $\Phi_{\atnew{\pvec{c}''}} \vdash \Phi Id$ trivially. Furthermore, we proved that $\Phi_{\atnew{\pvec{c}''}} \vdash Z\tau \aeq{C}^? Z\theta\tau'Id$ for all variable $Z\in\V$. 
                \end{claim}

        The result follows by combining Claims (1) and (2) and using Lemma~\ref{alemma:quasiorder}, concluding the proof of Correctness.

        \end{itemize}
\end{proof}

\end{document}